%% file: arxiv_main.tex
\documentclass[11pt]{article}
\usepackage[letterpaper,margin=1in]{geometry}
\usepackage{setspace}
\RequirePackage{tgtermes}
\RequirePackage{newtxtext}
\usepackage{amsmath}
\usepackage{amssymb}
\usepackage{amsthm}
\RequirePackage{newtxmath}
\RequirePackage{bm}
\RequirePackage{endnotes}

\usepackage{algorithm}
\usepackage{algpseudocode}
\usepackage{tikz}
\usetikzlibrary{calc}
\usetikzlibrary{shapes.geometric}
\usepackage{graphicx}
\usepackage{pgfplots}
\pgfplotsset{compat=1.17}
\usepackage{subcaption}
\usepackage{float}
\usepackage{array}
\usepackage[hypertexnames=false, plainpages=false, naturalnames=true]{hyperref}
\usepackage{cleveref}

\newcommand{\bN}{\bm{N}}
\newcommand{\bx}{\bm{x}}
\newcommand{\bY}{\bm{Y}}
\newcommand{\by}{\bm{y}}
\newcommand{\bve}{\bm{\varepsilon}}
\newcommand{\bg}{\bm{g}}

\newcommand{\mF}{\mathcal{F}}
\newcommand{\mB}{\mathcal{B}}
\newcommand{\tF}{\tilde{\mathcal{F}}}
\newcommand{\tB}{\tilde{\mathcal{B}}}
\newcommand{\Halmos}{}
\usepackage{natbib}
 \bibpunct[, ]{(}{)}{,}{a}{}{,}%

\usepackage{xcolor}

\newtheorem{theorem}{Theorem}
\newtheorem{lemma}{Lemma}
\newtheorem{proposition}{Proposition}
\newtheorem{corollary}{Corollary}

\newtheorem*{lemma*}{Lemma}
\newtheorem*{proposition*}{Proposition}
\theoremstyle{definition}
\newtheorem{definition}{Definition}
\newtheorem{assumption}{Assumption}

\theoremstyle{remark}
\newtheorem{remark}{Remark}

\title{{\fontsize{21}{25}\selectfont\bfseries Distributed Load Balancing with Workload-Dependent Service Rates}\thanks{An extended abstract of this paper appeared in the Proceedings of the 2025 ACM Conference on Economics and Computation.}}
\author{%
\begin{tabular}{c}
Wenxin Zhang\textsuperscript{1}, Santiago R.\ Balseiro\textsuperscript{1}, Robert Kleinberg\textsuperscript{2},\\
Vahab Mirrokni\textsuperscript{3}, Balasubramanian Sivan\textsuperscript{3}, Bartek Wydrowski\textsuperscript{3}\\[0.4em]
\small \textsuperscript{1}Graduate School of Business, Columbia University\\
\small \textsuperscript{2}Department of Computer Science, Cornell University\\
\small \textsuperscript{3}Google Research\\[0.2em]
\small \texttt{wz2574@columbia.edu}, \texttt{srb2155@columbia.edu}, \texttt{rdk@cs.cornell.edu}\\
\small \texttt{mirrokni@google.com}, \texttt{balusivan@google.com}, \texttt{bwydrowski@google.com}
\end{tabular}}
\date{}

\begin{document}

\maketitle
\begin{abstract}
\noindent Modern service systems, including cloud platforms and large language model inference endpoints, must distribute jobs across servers whose processing speeds depend on current workloads. At scale, centralized coordination is costly, while naive distributed policies can perform arbitrarily poorly. We study how to design a simple distributed load balancing policy that achieves globally optimal latency performance in such settings.
We model the system as a bipartite queueing network with an arbitrary compatibility graph and servers with concave, workload-dependent service rates. We propose the Greatest Marginal Service Rate (GMSR) policy, which routes jobs to a connected server where it has the largest marginal impact on service rate. In a discrete-time stochastic model, we show that as time discretization is refined (shrinking time step and job size proportionally), the scaled workload process converges almost surely to a fluid limit governed by a differential inclusion. In the fluid regime, GMSR reaches an $\epsilon$-suboptimal solution in $\mathcal{O}(\delta + \log(1/\epsilon))$ time from any $\delta$-suboptimal initial state, implying global convergence to the centrally optimal routing. When the system is overloaded, GMSR maximizes throughput, maximizes the number of stabilized backends among throughput-optimal policies, and minimizes total workload over those stabilized backends.

GMSR yields a practical routing rule that requires neither demand-rate knowledge nor centralized coordination. By relying only on local information, service providers can achieve near-optimal latency performance through decentralized decisions, making the policy well suited to large-scale cloud computing, LLM serving, and other distributed service environments where centralized control is costly or infeasible.
\end{abstract}



\section{Introduction}
{The infrastructure powering artificial intelligence is both enormously expensive and facing surging demand. On the supply side, an estimated \$5--8 trillion in data center capital expenditure will be needed by 2030 to meet worldwide demand for AI compute \citep{mckinsey2025cost}. On the demand side, AI inference workloads---the real-time serving of trained models to end users---are projected to account for two-thirds of all AI compute by 2026, with demand growing four- to five-fold annually \citep{deloitte2025compute}. These systems rely on computational resources deployed across a global network of data centers, and a key lever for efficiently utilizing this infrastructure is \emph{load balancing}: how should multiple frontends (routers) distribute jobs across a pool of backends (servers) to minimize latency? Better load balancing reduces latency at a given capacity, improving user experience, or equivalently, allows the same quality of service with fewer resources, improving return on investment in AI infrastructure.}

We consider a model where a set of frontends $\mF$ route jobs to a set of backends $\mB$ with the goal to minimize the average latency across all jobs. The connectivity between two sides is governed by an arbitrary bipartite graph $\mathcal{G} = (\mathcal{F}, \mathcal{B}, \mathcal{E})$: frontend $f$ can send jobs only to backends in $\mathcal{B}(f)$, reflecting constraints such as data residency regulations and geographical proximity.
Each backend $b\in\mB$ has a \emph{workload-dependent} service rate $\mu_b(N_b) \geq 0$, modeled as a non-decreasing and concave function of the workload $N_b$ (the current number of requests at the backend), with $\mu_b(0) = 0$ (an empty queue has no jobs to serve). {This modeling choice is motivated by the complexities of modern applications. For example, serving LLM inference queries requires expensive GPU computation where throughput depends nonlinearly on the number of concurrent jobs due to continuous batching and scheduling mechanisms, while contention for shared resources such as memory bandwidth causes the processing rate to increase concavely in the workload \citep{yu2022orca,kwon2023efficient,agrawal2024taming,zhong2024distserve,recasens2025mind}---an effect predicted by the Universal Scalability Law \citep[Chapter~6]{gunther2007guerrilla}. While LLM inference is our primary motivating application, our model and algorithms apply broadly to any service system with increasing, concave service rate functions.} 

Classic load-balancing policies are insufficient for this setting. For example, Join-the-Shortest-Queue (JSQ) ignores the heterogeneity of backends entirely: it drives the system toward equalizing workloads across all backends, which can be highly suboptimal when backends differ in service capacity. A natural fix is to account for service rates, as in heterogeneity-aware variants such as Generalized JSQ \citep{banawan1989load,selen2016approximate}: for instance, routing each job to the backend where it would experience the shortest expected latency. However, when service rates are workload-dependent, this greedy approach is also flawed. The root cause is an \emph{externality}: sending an additional job to a backend changes the backend's service rate and potentially degrades the service received by other jobs already present. Routing greedily to minimize one's own expected latency ignores this congestion externality, leading to outcomes that are socially suboptimal (see Pigou's classic example in Appendix~\ref{apx: failure of expected latency}). Figure~\ref{fig:intro-latency-utilization} illustrates this gap.

\begin{figure}[t]
    \centering
    \includegraphics[width=0.72\linewidth]{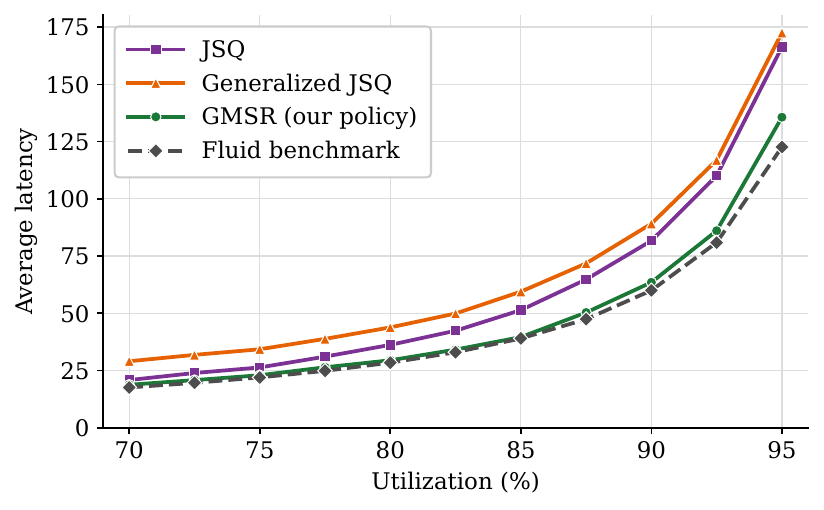}
    \caption{Stochastic simulation in a two-backend system with service rates $\mu_1(N)=N/(1+N)$ and $\mu_2(N)=N/(16+N)$, with arrival rate swept over $\lambda\in[1.4,1.9]$ and utilization defined as $\lambda/(\mu_1(\infty) + \mu_2(\infty))$. The dashed curve is the centralized fluid benchmark from the fluid relaxation (defined in \eqref{eq: FLU}). GMSR remains close to this benchmark, while JSQ and Generalized JSQ incur substantially higher latency at high utilization. Our theory shows that the GMSR--benchmark gap vanishes under the fluid scaling. Section~\ref{sec: randomized simulation} reports simulations for additional randomly generated setups.}
    \label{fig:intro-latency-utilization}
\end{figure}

An additional challenge is that centralized coordination is costly if not infeasible at scale: it would require each frontend to know the arrival rates at all other frontends, imposing communication overhead that grows with the system. We therefore focus on \emph{distributed} control, where each frontend makes routing decisions independently using only local information about its connected backends. Beyond circumventing the need for global information, distributed control offers robustness (no single point of failure), scalability (frontends and backends can be added or removed seamlessly), and reduced communication overhead---all critical for modern web services. However, distribution introduces its own difficulty: a routing decision that looks locally optimal for one frontend can degrade throughput system-wide.

\subsection{Contributions}
We develop a simple distributed load balancing policy called the Greatest Marginal Service Rate policy (GMSR), which facilitates coordination among the frontends and achieves stability and asymptotic latency optimality under minimal assumptions on the problem setting. Our stability analysis introduces a novel Lyapunov function and relies on a combinatorial argument, which may be of independent interest.

When a job arrives at a frontend, GMSR routes it to a connected backend with the highest marginal service rate, i.e., the backend where an additional job would have the highest impact on the current service rate.
The policy is fully distributed (each frontend needs only the workloads at its connected backends) and agnostic to arrival rates, in contrast to LP-based policies that require knowledge of arrival rates (e.g.,~\citealt{bassamboo2006design}).
GMSR is motivated by the structure of the centrally coordinated optimum: at the optimal solution, the marginal service rates are equal for any two backends receiving positive flow from the same frontend (Lemma~\ref{lem: optimal FLU structure}). GMSR drives the system toward this balanced-gradient condition.

\paragraph{\textbf{Main results}.}
We establish the following results for the GMSR policy:
\begin{itemize}
    \item \textit{Convergence to fluid model.} In a large-system limit where job sizes shrink to zero and arrival rates are scaled correspondingly, the discrete-time stochastic system dynamics converge almost surely to a solution of the fluid differential inclusion (Theorem~\ref{thm: converge to fluid model}).
    \item \textit{Global stability and latency optimality.} GMSR is shown to minimize the expected latency jobs experience in the fluid limit. Specifically, the fluid model under GMSR is globally strongly asymptotically stable: regardless of the initial workloads, every trajectory converges to latency minimizing equilibrium point over time (Theorem~\ref{thm:stability}). 
    \item \textit{Convergence rate.} From $\delta$-suboptimal initial workloads, i.e., the absolute difference between the sum of the initial workloads and the sum of the optimal workloads is $\delta > 0$, GMSR attains an $\epsilon$-suboptimal solution in $\mathcal{O}(\delta + \log{1/\epsilon})$ time (Theorem~\ref{thm:convergence}, $0<\epsilon < \delta$). The convergence has two phases: if the system starts far from optimality, at least one backend's workload decreases linearly until the system enters a regime of exponentially fast convergence.
    \item \textit{Overloaded systems.} When the system cannot be fully stabilized, GMSR lexicographically maximizes throughput, maximizes the number of stabilized backends among throughput-optimal policies, and minimizes total workload over those stabilized backends.
\end{itemize}
GMSR can be viewed as a generalization of JSQ: for identical servers with strictly concave service rates, the backend with the highest marginal rate is the one with the smallest workload, so GMSR reduces to JSQ. Our convergence result then implies the asymptotic optimality of JSQ for identical servers with general bipartite graphs---which, to our knowledge, had not previously been established.

\paragraph{\textbf{Main challenges and our techniques}.}
GMSR resembles two classical frameworks, but fits neither.

The first is \emph{selfish routing / Wardrop equilibria} \citep{wardrop1952road}. In a routing game, each agent picks a route to minimize its own latency, where latency is a fixed function of the aggregate flow on each route. The system is static: given a routing decision, latencies are determined instantaneously, and the classical result is that marginal-cost pricing (Pigouvian taxes) induces the socially optimal flow \citep{beckmann1956studies,pigou2017economics}. GMSR uses a similar marginal-cost signal---the marginal service rate---but operates in a fundamentally dynamic setting. Routing decisions change workloads over time, which change service rates, which change the optimal routing. There is no fixed ``latency function'' to optimize against; the landscape shifts with every routing decision.

Second, GMSR differs from \emph{first-order optimization} methods such as Frank-Wolfe or gradient descent. Those methods repeatedly evaluate an objective at the current iterate and move with a controlled, typically diminishing, step size. GMSR does neither. Each arriving job is routed greedily and irrevocably to the backend with the highest current marginal service rate, while the state used to make that decision is itself evolving stochastically under past routing choices.

General concave service rates make the analysis substantially harder. With piecewise-linear rates (e.g., M/M/$k$ queues), faster backends are preferred until they saturate. Under general concavity, marginal service rates vary continuously as workloads evolve, causing the preferred backend to change frequently and the routing structure to reconfigure in complex ways. At the same time, decentralized frontends may herd to the same backend, and ties in marginal service rates create discontinuities in the dynamics.

To handle these challenges, we analyze GMSR through its fluid limit: we show the discrete-time stochastic system converges almost surely to a \emph{differential inclusion}---a set-valued ordinary differential equation that captures tie-breaking ambiguity \citep{aubincellina1984differential} (Theorem~\ref{thm: converge to fluid model}). We then prove that the centrally optimal routing is the unique equilibrium of this differential inclusion (Lemma~\ref{lem: Equilibrium Points are Fluid Optimal}).
To prove global convergence, we construct a novel Lyapunov function measuring total flow imbalance (Lemma~\ref{lem: m V positive definite}) and introduce a combinatorial \emph{tier graph} (Definitions~\ref{def:tier} and~\ref{def:tiergraph}) that tracks how frontends' routing decisions evolve as workloads move across equal-gradient hypersurfaces, including sliding, splitting, and reconfiguration events related to sliding mode control \citep{utkin2013sliding}.
The key step is to show that this Lyapunov function decreases despite these structural changes.

\subsection{Organization of the Paper}
The rest of the paper is organized as follows. In Section \ref{sec: related work}, we review related literature. Section \ref{sec: problem formulation} describes the discrete-time model and the fluid relaxation problem. The GMSR policy is introduced in Section \ref{sec: GMSR}, and we present the convergence to the fluid model in Section \ref{sec: convergence to fluid}. The stability analyses are in Section \ref{sec: stability analysis}. Section~\ref{sec: overload system} analyzes overloaded systems. Section~\ref{sec: randomized simulation} reports additional simulations, and Section~\ref{sec: conclusions} concludes with further research directions. All proofs are deferred to the appendix.

\section{Literature Review}\label{sec: related work}

Our work sits at the intersection of three literatures: workload-dependent service rates, load balancing in stochastic systems, and AI inference serving efficiency. 

\paragraph{\textbf{Workload-dependent service rates.}}
Empirical evidence shows that service rates in practice are often workload-dependent: \cite{edie1954traffic,batt2012doctors,tan2014does} document this phenomenon among human workers in toll booths, hospitals, and call centers.
On the modeling side, queueing systems with state-dependent service rates have been studied both descriptively \citep{harris1967queues,bekker2004queues,abouee2016state} and prescriptively: for admission control \citep{bekker2006optimal}, staffing \citep{dong2013slowdown}, concurrency limits under processor sharing \citep{gupta2022approximations}, and service rate control \citep{xia2017optimal}.
A classical line of work also studies routing to parallel queues with state-dependent rates \citep{johri1989optimality,menich1991optimality,sparaggis1993extremal}, but these are restricted to homogeneous servers with a single router.
We study a substantially different setting: distributed load balancing across heterogeneous servers with workload-dependent rates over an arbitrary bipartite graph.

\paragraph{\textbf{Load balancing.}}
The Join-the-Shortest-Queue (JSQ) policy directs each job to the server with the fewest jobs. Its optimality for a single router with identical servers has been established under various conditions \citep{winston1977optimality,weber1978optimal,johri1989optimality,liu2022steady}. Extensions to heterogeneous servers include heavy-traffic optimality results \citep{chen2012asymptotic,hurtado2021throughput} and faster-server-first policies in the Halfin-Whitt regime \citep{armony2005dynamic,tezcan2008optimal}; these works assume constant service rates and a single router.

JSQ can also be applied to systems with multiple routers in a distributed fashion, where each router sends jobs to the connected server with the shortest queue \citep{foss1998stability,cruise2020stability,stolyar2005optimal,van2018scalable,cardinaels2022heavy}. Among these, \cite{weng2020optimal,rutten2023load,zhao2024exploiting} are closest to our model: they study bipartite load balancing with distributed JSQ variants. Our work differs in three respects: we allow workload-dependent service rates, operate in the fluid regime with a fixed network (versus many-server scaling), and require only a minimal necessary stability condition (Assumption~\ref{assmp: system stable}) rather than additional graph connectivity assumptions.
The closest work is \citet{horvath2019mean}, which models servers with state-dependent rates but studies a single router, analyzes a different policy (Join-Below-Threshold), and uses mean-field scaling. Our setting and the policy analyzed are substantially different.

\paragraph{\textbf{AI inference serving efficiency.}}
There is growing interest in applying operations research tools to improve the efficiency of AI inference. Recent work addresses throughput-optimal scheduling \citep{li2025throughput,bari2025optimal}, online scheduling with KV cache constraints \citep{jaillet2025online,ao2025optimizing}, robust scheduling under output-length uncertainty \citep{chen2025adaptively}, queueing-theoretic modeling of LLM inference latency \citep{yang2024queueing,mitzenmacher2025queueing}, and batch scheduling with variable prefill and decode lengths \citep{wang2025llm}. To our knowledge, the present paper is the first to study the distributed load balancing problem in this context, bridging the gap between the queueing-theoretic foundations of load balancing and the operational needs of large-scale AI inference.

\section{Problem Formulation}\label{sec: problem formulation}
In this section, we introduce a discrete-time model for the load balancing problem with a bipartite queueing system. Consider a bipartite graph $\mathcal{G} = (\mathcal{F}, \mathcal{B}, \mathcal{E})$, where $\mathcal{F}$ is the set of frontends, $\mathcal{B}$ is the set of backends, and $\mathcal{E}$ is the set of edges connecting $\mF$ to $\mB$: if $(f,b) \in \mathcal{E}$ for $f\in \mathcal{F}, b\in \mathcal{B}$, then frontend $f$ can route jobs to backend $b$. We denote by $\mathcal{F}(b)$ and $\mathcal{B}(f)$ the set of frontends connected to backend $b\in \mathcal{B}$ and the set of backends connected to frontend $f\in \mathcal{F}$, respectively. We assume that no frontends or backends are isolated, i.e., $\mathcal{F}(b) \neq \emptyset$, $\mathcal{B}(f) \neq \emptyset$ for all $f \in \mathcal{F}, b\in \mathcal{B}$.

In each period $i\geq 0$, the number of jobs arriving at frontend $f$ is a random variable independently drawn from a distribution that can be time-dependent with mean $\lambda_f$, denoted by $W_f(i)$. Upon arrival, a job can either queue at the frontend or be sent to a connected backend for service. The job service rate at the backend $b\in \mathcal{B}$ is denoted by $\mu_b(N_b)\geq 0$, which is a function of the workload $N_b$ at this backend $b$. Let $G_f(i)$ denote the number of jobs queued at frontend $f$ and let $N_b(i)$ denote the number of jobs at backend $b$ at time period $i$, the system dynamics are
given by
\begin{align*}
G_f(i+1) = G_f(i) + W_f(i) - \sum_{b\in \mathcal{B}(f)} A_{f,b}(i), \quad N_b(i+1) = N_b(i) + \sum_{f\in \mathcal{F}(b)} A_{f,b}(i) - D_b(i),
\end{align*}
where
\begin{itemize}
    \item $A_{f,b}(i)$ denotes the number of arriving jobs routed from frontend $f$ to backend $b$ at time period $i$, and the total number of jobs routed from frontend $f$ cannot exceed the total number of jobs queueing there, i.e., $\sum_{b\in \mathcal{B}(f)} A_{f,b}(i) \leq G_f(i) + W_f(i)$.
    \item $D_b(i)$ denotes the number of job departures from backend $b$ at time period $i$, which satisfies $\mathbb{E}[D_b(i)\mid N_b(i)] = \mu_b(N_b(i))$ and $D_b(i) \leq N_b(i)$.
\end{itemize}

The decision-maker needs to design an online policy that routes arriving jobs to the backends to minimize the long-run average latency each job experiences. If the system can be stabilized, i.e., the workload does not explode, then Little's Law implies the objective is equivalent to minimizing the long-run average number of jobs in the system. 

Let $\pi$ denote any online policy for our load balancing problem, then the long-run average number of jobs in the system under policy $\pi$ is defined as 
\begin{align}\label{eq: def performance of pi}
    \limsup_{k\to\infty} \frac{1}{k} \mathbb{E}_\pi \left[\sum_{i=0}^{k-1}  \left( \sum_{f\in \mathcal{F}} G_f(i)+  \sum_{b\in \mathcal{B}}N_b(i)\right)\right],
\end{align}
where the expectation is taken over the arrival process, job service process, and the probability measure induced by the policy $\pi$.

To analyze the performance of a policy $\pi$, we benchmark against the fluid relaxation of the problem. The fluid optimization problem can be formulated as follows:
\begin{align}
\textsf{OPT =}    \min_{\bN, \bx} \quad & \sum_{b\in \mathcal{B}} N_b \label{eq: FLU}\\
    \text{s.t.} \quad & \sum_{f\in \mathcal{F}} \lambda_f x_{f,b} = \mu_b(N_b), \forall b\in\mathcal{B}, \notag\\
    & \sum_{b\in \mathcal{B}}x_{f,b} = 1, \forall f \in \mathcal{F}, \notag\\
    & x_{f,b} \geq 0, \forall (f,b) \in \mathcal{E}, \notag \\
    & x_{f,b} = 0, \forall (f,b) \notin \mathcal{E}, \notag
\end{align}
where $N_b$ denotes the steady-state load level at each backend $b\in \mathcal{B}$, and $x_{f,b}$ denotes the proportion of jobs that is sent from frontend $f$ to backend $b$. 
The first constraint imposes flow balance at each backend, i.e., at the equilibrium, the total flow into backend $b$, $\sum_{f\in \mathcal{F}} \lambda_f x_{f,b}$ equals the total flow out of backend $b$, $\mu_b(N_b)$. The second constraint imposes flow balance at each frontend, i.e., all jobs are sent to backends. The last two sets of constraints impose the connectivity constraints.

We conduct our analysis under the following assumptions.
\begin{assumption}\label{assmp: increasing concave smooth ell}
    The service rate functions $\{\mu_b\}_{b\in \mathcal{B}}$ are strictly increasing, strictly concave, bounded, and differentiable, with $\mu_b(0)= 0$.
\end{assumption}

Service rates are non-decreasing because higher loads lead to higher service rates: the service rate never decreases when new jobs arrive because the backend always has the option to delay working on those jobs and continue processing the ones that were already in progress. Concavity reflects decreasing returns to scale: the available resources to process the next-arriving job are a decreasing function of the number of jobs already being served. When the workload is zero, so is the service rate. We assume boundedness to reflect the finite capacity of real-world systems. Differentiability and strict concavity are technical assumptions that simplify the analysis. The assumption that service rate functions are strictly increasing is a direct consequence of strict concavity combined with non-decreasingness. 

Conceptually, $\mu_b(\cdot)$ acts as a \emph{performance envelope} for backend~$b$: it summarizes the aggregate processing capacity as a function of the total workload, abstracting away the internal mechanisms---for instance, continuous batching, disaggregated prefill-decode, memory management in LLM inference application---that determine how that capacity arises. Our model is agnostic to how service capacity is allocated among individual requests.

\begin{assumption}\label{assmp: system stable}
    The optimization problem \eqref{eq: FLU} is feasible.
\end{assumption}
The feasibility assumption is necessary for the system to be stable, i.e., the service capacity of the backends is enough to serve all arriving jobs without allowing the queue sizes to explode. A necessary and sufficient condition for feasibility is that for every subset $P\subseteq \mathcal{F}$ of frontends, the total arrival rate does not exceed the aggregate maximum service capacity of all reachable backends:
\[\sum_{f\in P}\lambda_f \leq \sum_{b\in \cup_{f\in P} \mB(f)}\mu_b(\infty).\]
This is the analogue of Hall's marriage condition \citep{tassiulas1990stability} for our setting with workload-dependent service rates.
Assumption~\ref{assmp: system stable} can be relaxed to admit overloaded systems, where GMSR lexicographically maximizes throughput, maximizes the number of backends that remain stable, and minimizes their total workload; see Section~\ref{sec: overload system} and Appendix~\ref{apx: overload system} for details.


The fluid optimum $\textsf{OPT}$ serves as a lower bound on the performance of any online policy:
{ \begin{lemma}\label{lem: OPT lower bound}
    For any online policy $\pi$,
\[\textsf{OPT} \leq \liminf_{k\to\infty} \frac{1}{k} \mathbb{E}_\pi \left[ \sum_{i=0}^{k-1}  \left( \sum_{f\in \mathcal{F}} G_f(i)+  \sum_{b\in \mathcal{B}}N_b(i)\right) \right].\]
\end{lemma}
As $\liminf_{k\to\infty} a_k \leq \limsup_{k\to\infty} a_k$ for any sequence $\{a_k\}$, Lemma \ref{lem: OPT lower bound} establishes that $\textsf{OPT}$ is a lower bound on the long-run average number of jobs under any online policy~\eqref{eq: def performance of pi}.}
This holds regardless of whether the policy is centralized or distributed.

Next, we explore two properties of the optimal solution to the fluid optimization problem.

\begin{lemma}\label{lem: unique N star}
The optimal $\bm{N}^*$ to the fluid optimization problem \eqref{eq: FLU} is unique. 
\end{lemma}
The uniqueness of the optimal workload holds as the service rate functions are strictly increasing and strictly concave, and we prove this by contradiction.
Note that the optimal flow assignment $\bx^*$ could be non-unique if the graph has cycles because we can always push flow through a circulation without changing the objective value. 

\begin{lemma}\label{lem: optimal FLU structure}
Let $(\bN, \bx)$ denote a feasible solution to the fluid optimization problem \eqref{eq: FLU}, then $(\bN, \bx)$ is optimal if and only if for each frontend $f\in \mathcal{F}$ there exists a constant $c_f$ such that for $b\in \mathcal{B}(f)$ we have $1/\mu_b'(N_b) \geq c_f$ with equality holding if $x_{f,b} >0$.
\end{lemma}
This lemma motivates our GMSR policy, as we explain in detail in the next section. 

\section{Greatest Marginal Service Rate Policy and System Dynamics}\label{sec: GMSR}

Lemma \ref{lem: optimal FLU structure} implies that for any two backends $b_1, b_2$ connected to the same frontend $f$, if the optimal proportions of jobs sent to them are positive, i.e., $x_{f,b_1}^* >0, x_{f,b_2}^* >0$, then their service rate gradients at the optimal workloads must be equal:
\[\mu'_{b_1}({N_{b_1}^*}) = \mu'_{b_2}({N_{b_2}^*}).\]

Motivated by this observation---the gradients of service rate functions are balanced in the optimal solution---we consider the following policy, which we refer to as the Greatest Marginal Service Rate policy (GMSR): when a job arrives at a frontend, send it to a connected backend with the highest gradient $\mu_b'(N_b)$.
Because gradients are decreasing in workload, routing jobs to the highest-gradient backend would force the system to equalize gradients across backends.
In this policy, frontends make independent decisions without communication with each other, and each frontend only needs to know the workloads at its connected backends. Notably, the policy does not require knowledge of the job arrival rates $\{\lambda_f\}_{f\in \mathcal{F}}$.

Mathematically, frontends send jobs to a connected backend with the highest gradient, i.e., for all $f\in \mathcal{F}$ we pick $\bx_f \in X_f(\bN)$ where
\begin{align}
    X_f(\bN) := \arg\max & \sum_{b\in\mathcal{B}(f)} {x_{f,b}}{\mu'_b(N_b)} \label{eq: opt X_f}\\
    \text{s.t.} & \sum_{b\in\mathcal{B}} x_{f,b} = 1, \notag \\
    & x_{f,b} \geq 0, \forall (f,b) \in \mathcal{E}, \notag\\
    & x_{f,b} = 0, \forall (f,b) \notin \mathcal{E}. \notag
\end{align}
Here $x_{f,b}$ represents the proportions of jobs routing from frontend $f$ to backend $b$. It is important to note that the system dynamics under GMSR may not be uniquely determined at all times due to potential ties in the gradients $\mu'_b(N_b)$. Thus, $X_f(\bN)$ is a \emph{set-valued} function.

We remark that GMSR can break ties between backends in any way and how ties are broken in the discrete-time stochastic system does not impact our analysis. For example, if there is one frontend connected to two backends and $\mu_1'(N_1(i)) = \mu_2'(N_2(i))$, then the feasible proportion of jobs that are sent to the two backends can be any vector satisfying $x_{1,1}(i) + x_{1,2}(i) = 1, x_{1,1}(i),x_{1,2}(i)\geq 0$. Upon picking $\bx_f$, the frontend can implement this routing proportion in any way as long as the number of jobs distributed follows $\bx_f$ in expectation. For example, the control can be implemented in a weighted round-robin or probabilistic fashion.

As GMSR is an output-queuing policy, i.e., $W_f(i) = \sum_{b\in \mB(f)} A_{f,b}(i)$, the queues at the frontends are assumed to be always empty.
The system dynamics under GMSR is given by
\[G_f(i) = 0, \quad N_b(i+1) = N_b(i) + \sum_{f\in \mathcal{F}(b)} A_{f,b}(i) - D_b(i),\]
where
\begin{itemize}
    \item $A_{f,b}(i)$ denotes the number of arriving jobs routed from frontend $f$ to backend $b$, which satisfies $\mathbb{E}[A_{f,b}(i)\mid \bN(i)] = \lambda_f x_{f,b}(i)$ with $\bx_f(i) \in X_f(\bN(i))$.
    \item $D_b(i)$ denotes the number of job departures from backend $b$, which satisfies $\mathbb{E}[D_b(i)\mid N_b(i)]= \mu_b(N_b(i))$ and $D_b(i) \leq N_b(i)$.
\end{itemize}

\section{Convergence to the Fluid Model}\label{sec: convergence to fluid}
In this section, we show that the discrete system dynamics under GMSR converge to a fluid model in a large-system limit in which we shrink the job size to zero and scale the arrival rates correspondingly. In this fluid model, jobs are modeled as infinitely divisible continuous flows and there is no stochasticity in arrivals or service times.

To formalize this, fix a time horizon $T>0$, and consider a sequence of discrete systems indexed by a ``system size'' parameter $c \in \mathbb{N}$. The dynamics within time interval $[0,T]$ are as follows:
    \[N_b^{(c)}(i+1) = N_b^{(c)}(i) + \sum_{f\in \mathcal{F}(b)} A_{f,b}^{(c)}(i) - D_b^{(c)}(i),\]
where 
\begin{itemize}
    \item $i$ is the index of the discrete time step, each time step has a physical time length of $1/c$ in the $c^{th}$ system. Hence $i$ ranges over $\{0,1,\ldots, \lfloor Tc\rfloor\}$;
    \item $A^{(c)}_{f,b}(i)$ satisfies $\mathbb{E}[A^{(c)}_{f,b}(i)\mid \bN^{(c)}(i)] = \lambda_f x^{(c)}_{f,b}(i)$ with $\bx^{(c)}_f(i) \in X_f(\bN^{(c)}(i)/c)$, and $\sum_{b\in \mathcal{B}(f)} A^{(c)}_{f,b}(i) = W^{(c)}_f(i)$;
    \item $D_b^{(c)}(i)$ satisfies $\mathbb{E}[D_b^{(c)}(i)\mid N_b^{(c)}(i)] = \mu_b(N_b^{(c)}(i)/c)$ and $D_b^{(c)}(i) \leq N_b^{(c)}(i)$.
\end{itemize}
The stochastic recursion above is often referred to as a \emph{stochastic recursive inclusion.}
For both job arrivals and departures, we allow general distributions and impose only the following condition.
\begin{assumption}\label{assmp: moment bound}
    The arrival and departure processes satisfy
    \[\sup_{c \in \mathbb{N}} \sup_{1\leq i\leq Tc}\mathbb{E}[(W^{(c)}_f(i))^4] < \infty \quad \text{for all } f \in \mathcal{F}, \qquad \sup_{c \in \mathbb{N}} \sup_{1\leq i\leq Tc}\mathbb{E}[(D_b^{(c)}(i))^4] < \infty \quad \text{for all } b \in \mathcal{B}.\]
\end{assumption}
This uniform fourth-moment bound ensures that the stochastic fluctuations in arrivals and departures vanish in the large-system limit. It is the key ingredient in establishing a uniform law of large numbers for the martingale difference arrays that arise in the convergence proof (Lemma~\ref{lem: SLLN for epsilon}).


Our scaling shrinks the length of each time step and the size of jobs by a factor of $1/c$, where we scale $c$ to infinity. Let $Y_b^{(c)}(i) := N_b^{(c)}(i)/c$ for all $b\in \mathcal{B}$ be the normalized workloads. To analyze the convergence, let $t$ denote the physical time index, which should be distinguished from the algorithmic discrete time step index $i$. Specifically, in the $c^{th}$ system, the physical time for one discrete time step is $1/c$, thus the physical time after $i$ steps is $i/c$. On the other hand, given physical time $t > 0$, the number of steps passed is given by $\lceil t c \rceil$. 

We construct a continuous piecewise-linear interpolation $\bar{\bY}^{(c)}(t)$ of the discrete process ${\bY}^{(c)}(i)$ over $t\in [0,T]$, where $\bar{\bY}^{(c)}(i/c) = \bY^{(c)}(i), i\geq 0$, and
\[\bar{\bY}^{(c)}(t) = \bY^{(c)}(i) + \left(\bY^{(c)}(i+1) - \bY^{(c)}(i)\right) \frac{t-i/c}{1/c}, t\in [i/c, (i+1)/c).\]

Our convergence result is as follows. We define $C([0,T], \mathbb{R}^{|\mathcal{B}|})$ to be the set of continuous functions with uniform norm topology, i.e., set of continuous functions $Y: [0,T] \mapsto \mathbb{R}^{|\mB|}$ with norm $\|Y\| = \sup_{t\in [0,T]}|Y(t)|$.
\begin{theorem}\label{thm: converge to fluid model}
    Under Assumption~\ref{assmp: moment bound}, almost surely, the set of sub-sequential limits of $\{\bar{\bY}^{(c)}(\cdot)\}$ as $c\to\infty$ is non-empty in $C([0,T], \mathbb{R}^{|\mathcal{B}|})$, and for every such limit point $\bN(\cdot)$, there exists $\bx_f(t) \in X_f(\bN(t))$ for all $f\in \mathcal{F}$, $t\in [0,T]$ such that $\bN(\cdot)$ satisfies 
\[N_b(t) = N_b(0) + \int_0^t \left(\sum_{f\in \mathcal{F}(b)}\lambda_f x_{f,b}(\tau)- \mu_b(N_b(\tau)) \right) \, d\tau \text{ for all } b \in \mathcal{B}, t \in [0,T].\]
\end{theorem}

The proof builds on Theorem 5.2 of \citet{borkar2008stochastic} and Theorem 3.7 of  \citet{duchi2018stochastic}, who study stochastic recursive inclusions. The main difference is that they consider diminishing time step sizes and analyze the system as time goes to infinity; our results focus on a sequence of systems indexed by $c$ and examine the system as $c$ goes to infinity. This shift necessitates the analysis of a uniform law of large numbers for triangular arrays of martingale differences, which we prove using the Burkholder inequality (see Lemma \ref{lem: SLLN for epsilon}). The proof has two steps. First, we use our uniform law of large numbers to show that the stochastic process $\{\bar{\bY}^{(c)}(\cdot)\}_{c\in \mathbb{N}}$ converges uniformly to its compensator, i.e., the integral of its drift. Because the compensator is equicontinuous, we use the Arzel\`a-Ascoli theorem to extract a converging subsequence, which implies, in turn, the relative compactness of the original process $\{\bar{\bY}^{(c)}(\cdot)\}_{c\in \mathbb{N}}$. Second, we apply some functional analysis results to show that every limit point is a solution to the differential inclusion.

\begin{remark}
Note that $\bx_f(t) \in X_f(\bN(t))$ is a set-valued function due to the potential non-uniqueness of optimal routing decisions when gradients are tied, thus the limit point $\bN(\cdot)$ is a solution to a differential inclusion, which is a mathematical formulation that generalizes differential equations to allow derivatives to belong to a set of possible values (for details, see \citealt{aubincellina1984differential,smirnov2022introduction}). Specifically, the dynamics of the fluid limit is defined by:
\begin{align}\label{eq: differential inclusion}
    \dot{\bN}(t) \in H(\bN(t)),
\end{align}
where 
\begin{align}\label{eq: def H(N)}
    H(\bN(t)):= \left\{(v_1, \ldots, v_{|\mathcal{B}|}): v_b = \sum_{f\in \mathcal{F}}\lambda_f x_{f,b}(t) - \mu_b(N_b(t)), \bx_f \in X_f(\bN), \forall f\in \mathcal{F}, b\in \mathcal{B}\right\}.
\end{align}
\end{remark}

Theorem~\ref{thm: converge to fluid model} shows the stochastic recursive inclusion can be seen as a noisy discretization of the differential inclusion \eqref{eq: differential inclusion}: $\{\bar{\bY}^{(c)}\}$ ``track" a solution with probability 1 on every finite horizon. For proof of the existence of the global solution, i.e., on $[0,\infty)$, see Appendix \ref{apx: DI}

While the differential inclusion allows for arbitrary tie-breaking in routing decisions, the set of feasible routing decisions in a solution to the differential inclusion is in fact quite constrained by the requirement that the solution be absolutely continuous. This is due to the self-correcting feature of the policy. For instance, consider a scenario with one frontend and two identical backends, both starting with zero workloads. If the frontend initially routes unequal flows to the backends, the backend receiving more flow will accumulate more workloads, resulting in a lower gradient and making it less attractive to the frontend. When jobs are discrete, the system state will oscillate around the equal-gradient curve, where both backends have the same gradient. {However, in the fluid model, when we model jobs as infinitely divisible continuous flows, the system always stays on the equal gradient curve. To do so, in the one-frontend-two-backends example, the differential inclusion must pick the unique routing proportions $(x_{1,1}(t), x_{1,2}(t)) \in X_1((N_1(t), N_2(t)))$ that produce a drift tangent to the equal-gradient curve. The routing proportions $\bx_1(t)$ in the differential inclusion can be interpreted as a ``local'' time-average of the routing proportions $\bx_1^{(c)}(i)$ of the stochastic recursive inclusion. }

In the following lemma, we show that the equilibrium of the dynamical system \eqref{eq: differential inclusion} coincides with the unique optimal solution to the fluid optimization problem. 

\begin{lemma}\label{lem: Equilibrium Points are Fluid Optimal}
The dynamic system \eqref{eq: differential inclusion} has a unique equilibrium point $\bN^{eq}$, which is optimal for the fluid optimization problem \eqref{eq: FLU}.
\end{lemma}

Therefore, the unique equilibrium of the dynamic system \eqref{eq: differential inclusion} is precisely $\bN^*$, the unique optimal solution to the fluid optimization problem \eqref{eq: FLU}.

\section{Stability and Convergence Rate Analyses}\label{sec: stability analysis}
In the previous section, we establish the convergence result, and that the equilibrium of the differential inclusion corresponds to the optimal fluid solution to \eqref{eq: FLU}. We now analyze the stability of this equilibrium and the convergence rate, which implies that GMSR is asymptotically optimal and achieves effective coordination in the fluid model.
Specifically, we show that $\bN^*$ is globally strongly asymptotically stable: regardless of the initial state (globally), every trajectory (strongly) of the differential inclusion converges to the equilibrium point $\bN^*$ as $t\to\infty$ (asymptotically).

\noindent Our main results are the following two theorems.
\begin{theorem}[Stability]\label{thm:stability}
The equilibrium $\bN^*$ of the differential inclusion \eqref{eq: differential inclusion} is globally strongly asymptotically stable: every solution $\bN(\cdot)$ with every initial state $\bN(0)$ satisfies $\lim_{t\to\infty} \bN(t) = \bN^*$.
\end{theorem}

{
\begin{theorem}[Convergence Rate]\label{thm:convergence}
    If initial workloads are $\delta$-suboptimal $(\delta >0)$, i.e., 
    \[\left|\sum_{b\in \mathcal{B}} N_b(0) - \sum_{b\in \mathcal{B}} N_b^*\right| = \delta,\] 
    it takes GMSR $\mathcal{O}( \delta + \log{1/\epsilon})$ time to obtain an $\epsilon$-suboptimal solution ($0<\epsilon <\delta$).
\end{theorem}}

To prove the stability result, we use the Lyapunov direct method. Specifically, we construct a function that serves as a measure of the system's ``energy,'' which is zero at the equilibrium point $\bN^*$, and demonstrate this function consistently decreases over time. The Lyapunov function is formally defined in Section~\ref{sec:preliminaries}.

{To prove the convergence rate result, we notice that the drift of the Lyapunov function restricted to a subset of backends with the same gradient is proportional to the Lyapunov function itself and the gradient of these backends. Therefore, if the gradients can be lower bounded by a constant, the Lyapunov function will converge exponentially fast. We define a set of system states in which gradients are uniformly bounded from below and show it is an invariant set, i.e., once the system reaches the set, it stays there forever. Then, we argue that the system state reaches this invariant set in linear time if starting outside, which gives our two-phase convergence process. We remark that the constants in the convergence rate can depend on the initial workloads and the initial parameters (the arrival rates, service rate functions, and compatibility graph).}

{The proofs for the following results, however complicated in the form, all rely on a simple yet powerful observation: frontends always send jobs to the most preferable connected backends. If a backend with a very high marginal service rate is not receiving jobs from a frontend, then either they are not connected, or the frontend is connected to another backend with an even higher marginal service rate.}

{Our analysis of the performance of GMSR primarily focuses on the fluid model. 
Back to the discrete model, the system may not converge precisely to the centrally coordinated optimal routing due to stochasticity and may oscillate around the optimal solution. The extent to which it oscillates should decrease as we scale the system size, and we leave the analyses of the discrete model with finite $c$ for future work. }

\subsection{Preliminaries}\label{sec:preliminaries}

The next lemma introduces the Lyapunov function and proves its positive definiteness. Recall the system state at time $t$ is $\bN(t) = (N_1(t), N_2(t), \ldots, N_{|\mathcal{B}|}(t))$, the workloads at each backend at time $t$, and the system dynamics follows $\dot{N}_b(t) = \sum_{f\in \mathcal{F}(b)}\lambda_f x_{f,b}(t) -\mu_b(N_b(t))$, where  $\bx_f(t) = (x_{f,1}(t), \ldots, x_{f,|\mathcal{B}|}(t)) \in X_f(\bN(t))$ for $f\in \mathcal{F}$. Let $\bx(t)$ denote the matrix $\{x_{f,b}(t)\}_{f\in \mathcal{F}, b\in \mathcal{B}}$.

\begin{lemma}\label{lem: m V positive definite}
   The function $V(\bN, \bx) = \sum_{b\in\mathcal{B}} \left|\sum_{f\in \mathcal{F}(b)}\lambda_f x_{f,b} - \mu_b(N_b)\right|$ is positive definite, i.e., $V(\bN, \bx) \ge 0$ for all $\bN, \bx$ and for all  $\bx_f \in X_f(\bN)$ we have
   \[V(\bN, \bx) = 0 \text{ if and only if }(\bN, \bx)  = (\bN^*, \bx^*),\]
   where $\bx^*$ satisfies $\sum_{f\in \mathcal{F}(b)}\lambda_f x_{f,b}^* -\mu_b(N_b^*) = 0$ for all $b\in \mathcal{B}$.
\end{lemma}

We next introduce some definitions that are based on workloads, $\bN$, at the backends. See Figure \ref{fig: ex G and best backend graph} for an illustration. 

\begin{definition}[Best Backend Graph]\label{def: best backend graph}
Given a connectivity bipartite graph $\mathcal{G} = (\mathcal{F}, \mathcal{B}, \mathcal{E})$ and workloads at the backends $\bm{N}$, we define the best backend graph at $\bN$ to be a subgraph of the original graph $\mathcal{G}$ that involves only those edges that connect each frontend $f$ to its most preferable backends according to $\mu_b'(N_b)$. Mathematically,
\[{G}(\bm{N}) = (\mathcal{F}, \mathcal{B}, {E}(\bm{N})) \text{ with } {E}(\bm{N}) = \{(f,b) \in \mathcal{E} | b \in S_f(\bm{N})\},\] 
where $S_f(\bm{N}) = \{b\in\mathcal{B}(f) | \mu_b'(N_b) = \max_{j\in \mathcal{B}(f)} \mu_j'(N_j)\}$ is the set of best backends for frontend $f\in \mathcal{F}$ at $\bN$.
\end{definition}

\begin{definition}[Tier]\label{def:tier}
    We say that $(F, B)$ with $F \subseteq \mathcal{F}, B\subseteq \mathcal{B}$ is a \emph{tier} at $\bN$ if 
    \begin{itemize}
        \item $F\cup B$ is connected in the best backend graph ${G}(\bN)$,
        \item for all frontends $f\in F$ there is no backend $b\in \mathcal{B} \setminus B$ such that $(f,b) \in {E}(\bN)$,
        \item for all backends $b\in B$ there is no frontend $f\in \mathcal{F} \setminus F$ such that $(f,b) \in {E}(\bN)$.
    \end{itemize}
\end{definition}
In graph theory, the above definition is equivalent to saying $F\cup B$ is the node set of a \textit{connected component} of the best backend graph $G(\bN)$. Note that $F$ can be an empty set for a tier, and this will happen when $B = \{b\}$ with $b\notin \cup_{f\in \mathcal{F}} S_f(\bN)$, i.e., no frontend finds the backend preferable. On the other hand, $B$ can never be an empty set as $\mathcal{B}(f) \neq \emptyset$ for all $f\in \mathcal{F}$, thus every frontend will find at least one backend preferable. By definition, all backends in the same tier share the same gradient value, so we let $\mu'_B(\bN)$ denote $\mu_{b}'(N_b)$ for $b\in B$ at $\bN$.

\begin{figure}[h!]
    \centering
    \begin{subfigure}[t]{0.48\textwidth}
        \centering
        \begin{tikzpicture}[
          every node/.style={circle, draw, minimum size=6mm},
          every edge/.style={draw, thick},
          scale=0.95
        ]
        \node (f1) at (0, 4) {$f_1$};
        \node (f2) at (0, 3) {$f_2$};
        \node (f3) at (0, 2) {$f_3$};
        \node (f4) at (0, 1) {$f_4$};
        \node (b1) at (3, 4) {$b_1$};
        \node (b2) at (3, 3) {$b_2$};
        \node (b3) at (3, 2) {$b_3$};
        \node (b4) at (3, 1) {$b_4$};
        \node (b5) at (3, 0) {$b_5$};
        \draw (f1) -- (b1);
        \draw (f1) -- (b2);
        \draw (f1) -- (b3);
        \draw (f2) -- (b2);
        \draw (f2) -- (b3);
        \draw (f3) -- (b3);
        \draw (f3) -- (b4);
        \draw (f4) -- (b4);
        \draw (f4) -- (b5);
        \draw (f4) -- (b1);
        
        \node[anchor=west, draw=none, opacity=0] at (3.5, 4) {\(\mu'_1(N_1) = 3\)};
        \node[anchor=west, draw=none, opacity=0] at (3.5, 3) {\(\mu'_2(N_2) = 2\)};
        \node[anchor=west, draw=none, opacity=0] at (3.5, 2) {\(\mu'_3(N_3) = 1\)};
        \node[anchor=west, draw=none, opacity=0] at (3.5, 1) {\(\mu'_4(N_4) = 2\)};
        \node[anchor=west, draw=none, opacity=0] at (3.5, 0) {\(\mu'_5(N_5) = 3\)};
        \end{tikzpicture}
        \caption{}
    \end{subfigure}%
    \hfill
    \begin{subfigure}[t]{0.48\textwidth}
        \centering
        \begin{tikzpicture}[
          every node/.style={circle, draw, minimum size=6mm},
          every edge/.style={draw, thick},
          scale=0.95
        ]
        \node (f1) at (0, 4) {$f_1$};
        \node (f2) at (0, 3) {$f_2$};
        \node (f3) at (0, 2) {$f_3$};
        \node (f4) at (0, 1) {$f_4$};
        \node (b1) at (3, 4) {$b_1$};
        \node (b2) at (3, 3) {$b_2$};
        \node (b3) at (3, 2) {$b_3$};
        \node (b4) at (3, 1) {$b_4$};
        \node (b5) at (3, 0) {$b_5$};
        \node[anchor=west, draw=none] at (3.5, 4) {\(\mu'_1(N_1) = 3\)};
        \node[anchor=west, draw=none] at (3.5, 3) {\(\mu'_2(N_2) = 2\)};
        \node[anchor=west, draw=none] at (3.5, 2) {\(\mu'_3(N_3) = 1\)};
        \node[anchor=west, draw=none] at (3.5, 1) {\(\mu'_4(N_4) = 2\)};
        \node[anchor=west, draw=none] at (3.5, 0) {\(\mu'_5(N_5) = 3\)};
        \draw (f1) -- (b1);
        \draw (f2) -- (b2);
        \draw (f3) -- (b4);
        \draw (f4) -- (b5);
        \draw (f4) -- (b1);
        \end{tikzpicture}
        \caption{}
    \end{subfigure}
    \caption{(a) A connectivity bipartite graph with four frontends (left side) and five backends (right side) with edges denoting feasible flow routing; (b) The best backend graph at a system state $\bN$ with $\mu_1'(N_1) = 3, \mu_2'(N_2) = 2, \mu_3'(N_3) = 1, \mu_4'(N_4)= 2, \mu_5'(N_5) = 3$.}
    \label{fig: ex G and best backend graph}
\end{figure}

\begin{definition}[TierGraph]\label{def:tiergraph}
 Let $k(\bN)$ denote the number of disjoint tiers at $\bN$,
 a TierGraph $\mathcal{T}(\bN)$ at $\bN$ is defined as follows:
    \begin{itemize}
        \item the graph has $k(\bN)$ vertices, with each vertex $v_i$ representing a tier $(F_i, B_i)$ at $\bN$,
        \item there is an arc pointing from vertex $v_i = (F_i, B_i)$ to vertex $v_j = (F_j, B_j)$ ($i\neq j$), which is denoted as $v_i\to v_j$, if there exists $f\in F_i, b \in B_j$ such that $(f,b) \in \mathcal{E}$, i.e., they are connected in the connectivity bipartite graph $\mathcal{G} = (\mathcal{F},\mathcal{B}, \mathcal{E})$.
    \end{itemize}
We say $(F_i, B_i) \prec (F_j, B_j)$ if vertex $v_i$ can reach vertex $v_j$ in $\mathcal{T}(\bN)$. 
\end{definition}

{For an example of tiers and TierGraph, please refer to Figure \ref{fig: ex connected graph with tiers} and Figure \ref{fig: ex tiergraph} respectively. In the remainder of the paper, we denote the vertices in the TierGraph and their corresponding tiers using the same index.
We next present some fundamental facts about the structure of the TierGraph and how the reachability relation on the TierGraph reveals the relative order of gradient values of tiers.}

\begin{corollary}\label{cor:gradient decreases along tiers}
    For two tiers $(F_i, B_i), (F_j, B_j)$ with $F_i, F_j \subseteq \mathcal{F}, B_i, B_j \subseteq \mathcal{B}$ at $\bN$, if $(F_i, B_i) \prec (F_j, B_j)$, then $\mu_{B_i}'(\bN) > \mu_{B_j}'(\bN)$. 
\end{corollary}

\begin{corollary}\label{cor: TierGraph DAG}
The TierGraph $\mathcal{T}(\bN)$ is a directed acyclic graph. 
\end{corollary}

\begin{figure}[h!]
    \centering
    \begin{minipage}{0.4\textwidth}
    \centering
\begin{tikzpicture}[
  every node/.style={circle, draw, minimum size=6mm},
  every edge/.style={draw, thick},
  scale=0.95
]

\node (f1) at (0, 4) {$f_1$};
\node (f2) at (0, 2) {$f_2$};
\node (f3) at (0, 1) {$f_3$};
\node (f4) at (0, 3) {$f_4$};

\node (b1) at (3, 4) {$b_1$};
\node (b5) at (3, 3) {$b_5$};
\node (b2) at (3, 2) {$b_2$};
\node (b3) at (3, 0) {$b_3$};
\node (b4) at (3, 1) {$b_4$};

\draw (f1) -- (b1);
\draw (f2) -- (b2);
\draw (f3) -- (b4);
\draw (f4) -- (b5);
\draw (f4) -- (b1);

\draw[dotted, thick, blue] ($(f1.north west) + (-0.2,0.2)$) rectangle ($(b5.south east) + (0.2,-0.2)$);
\node[anchor=north west, blue, draw=none] at ($(f1.north west) + (-1.5,0)$) {Tier 1};

\draw[dotted, thick, red]  ($(f2.north west) + (-0.2,0.2)$) rectangle ($(b2.south east) + (0.2,-0.2)$);
\node[anchor=north west, red, draw=none] at ($(f2.north west) + (-1.5,0)$) {Tier 2};

\draw[dotted, thick, orange] ($(f3.north west) + (-0.2,0.2)$) rectangle ($(b4.south east) + (0.2,-0.2)$);
\node[anchor=north west, orange, draw=none] at ($(b4.north west) + (1,0)$) {Tier 3};

\draw[dotted, thick, purple] ($(b3.north west) + (-0.2,0.2)$) rectangle ($(b3.south east) + (0.2,-0.2)$);
\node[anchor=north west, purple, draw=none] at ($(b3.north west) + (1,0)$) {Tier 4};

\draw[dashed] (f1) -- (b2);
\draw[dashed] (f1) -- (b3);
\draw[dashed] (f2) -- (b3);
\draw[dashed] (f3) -- (b3);
\draw[dashed] (f4) -- (b4);
\end{tikzpicture}
\caption{A connectivity bipartite graph from Figure \ref{fig: ex G and best backend graph}  (a) \textit{rearranged} with solid edges denoting that the frontend $f$ is connected to a best backends in $S_f(\bN)$ and dashed edges denoting connected but not best backends. The system state $\bN$ is the same as Figure \ref{fig: ex G and best backend graph} (b). } 
\label{fig: ex connected graph with tiers}
 \end{minipage}%
    \hfill
    \begin{minipage}{0.45\textwidth}
    \centering
\begin{tikzpicture}[
  every node/.style={circle, draw, minimum size=6mm},
  every edge/.style={draw, thick},
  scale=0.95
]

\node[ellipse] (T1) at (0, 4) {$v_1$ (Tier 1)};
\node[ellipse] (T2) at (-2.5, 2.5) {$v_2$ (Tier 2)};
\node[ellipse] (T3) at (2.5, 2.5) {$v_3$ (Tier 3)};
\node[ellipse] (T4) at (0, 1) {$v_4$ (Tier 4)};

\draw[->] (T1) -- (T4);
\draw[->] (T1) -- (T2);
\draw[->] (T1) -- (T3);
\draw[->] (T2) -- (T4);
\draw[->] (T3) -- (T4);
\end{tikzpicture}
\caption{The TierGraph $\mathcal{T}(\bN)$ at $\bN$ that corresponds to the graph in Figure \ref{fig: ex connected graph with tiers}. Here the tiers are $(\{f_1, f_4\}, \{b_1, b_5\}), (\{f_2\}, \{b_2\}), (\{f_3\}, \{b_4\})$ and $(\emptyset,\{b_3\})$, which corresponds to $v_1, v_2, v_3, v_4$ respectively in the TierGraph.}
\label{fig: ex tiergraph}
    \end{minipage}

\end{figure}

\subsection{Stability Analysis}
We proceed to prove the stability of the solution to the differential inclusion \eqref{eq: differential inclusion}. There are two steps in the proof:
\begin{enumerate}
    \item We show that if a tier does not change during a time interval, the workloads at the backends in this tier must all increase or all decrease, depending on the total flow imbalance within the tier (see Lemma \ref{lem: tier N same direction}). Specifically, when the total arrival rates exceed the total service rates, the workloads will increase and vice versa. We also establish flow imbalance inequalities for any subset of the tier, which will be used in the second step (see Lemma \ref{lem: K+ lambda big} and Lemma \ref{lem: K- lambda small}).
    
    \item We next analyze the Lyapunov function. When the tiers do not change during a time interval, we show the Lyapunov function is the sum of absolute values of total flow imbalance within each tier, which decreases with time (see Lemma \ref{lem: slide}). When the tiers' configurations change, we show the value of the Lyapunov function decreases (see Lemma \ref{lem:split} and Lemma \ref{lem:sync}) by analyzing how the tiers change, which relies on some fundamental structural properties of the TierGraph established in the previous section. { It is worth noting that our tier definition shares conceptual similarities with complete resource pooling subsystems in \cite{afeche2022optimal} and the routing components in \cite{ding2021fluid}, but we use it for very different analytical purposes.
    \cite{afeche2022optimal} use this concept to derive steady-state waiting times in heavy-traffic, and \cite{ding2021fluid} focus on characterizing the switch times at which routing components change, while we track the transient evolution of tiers to characterize the changes in the Lyapunov function under GMSR, independent of when switches occur.}
\end{enumerate}

\subsubsection{Analysis of a Tier's Dynamics}
First, we show that within a tier, the workloads of the backends evolve in the same direction. This is because,
by definition, all backends in the same tier will share the same marginal service rate. Therefore to stay in the
same tier, their gradients must remain aligned, and so are the workloads’ drifts.
\begin{lemma}\label{lem: tier N same direction}
    If there exists $t_1 < t_2$ and $F \subseteq \mathcal{F}, B\subseteq \mathcal{B}$ such that $(F,B)$ is a tier at $\bN(t)$ for $t\in [t_1, t_2]$, then for any $b_1, b_2 \in B$, $\frac{d}{dt}N_{b_1}(t) \cdot \frac{d}{dt}N_{b_2}(t) \geq 0$ for $t\in [t_1, t_2]$.
\end{lemma}

Next, we establish flow imbalance inequalities for subsets of a tier, which are critical for analyzing the Lyapunov function. {Within a tier, when the total arrival rates are higher than the total service rates, which we refer to as having ``positive flow imbalance," jobs arrive faster than jobs depart and the workload at \emph{every backend} must increase. Locally, we prove each backend must experience positive flow imbalance too, i.e., in-flow dominates out-flow, and for this to be feasible, the flow imbalance inequalities for any subset of backends and their connected frontends in the tier must also hold. This result leverages that backends are sliding in their equal-gradient hypersurface.}

\begin{lemma}\label{lem: K+ lambda big}
If there exists $t_1 < t_2$ and $F \subseteq \mathcal{F}, B \subseteq \mathcal{B}$ such that for $t\in [t_1, t_2]$,
\begin{itemize}
    \item $(F,B)$ is a tier at $\bN(t)$,
    \item $\sum_{f\in F} \lambda_f - \sum_{b\in B}\mu_b(N_b(t)) \geq 0$.
\end{itemize}
Then, for $t\in [t_1, t_2]$, the following holds:
\begin{itemize}
    \item $N_b(t)$ is non-decreasing in $t$ for all $b\in B$,
    \item for any subset of backends $Q\subseteq B$,
    \[\sum_{f\in F\cap (\cup_{b\in Q} \mathcal{F}(b))} \lambda_f \geq \sum_{b\in Q} \mu_b(N_b(t)).\]
\end{itemize}
\end{lemma}

{Similarly, we have the following results for the ``negative flow imbalance'' case, i.e., when the total arrival rates are lower than the total service rates. In this case, jobs arrive slower than jobs depart, so the workloads at the backends must decrease. This implies similar local inequalities for any subset of frontends and their connected backends in the tier.
}

\begin{lemma}\label{lem: K- lambda small}
If there exists $t_1 < t_2$ and $F \subseteq \mathcal{F}, B \subseteq \mathcal{B}$ such that
\begin{itemize}
    \item $(F,B)$ is a tier at $\bN(t), t\in [t_1, t_2]$,
    \item $\sum_{f\in F} \lambda_f - \sum_{b\in B}\mu_b(N_b(t)) \leq 0$.
\end{itemize}
Then, for $t\in [t_1, t_2]$, the following holds:
\begin{itemize}
    \item $N_b(t)$ is non-increasing in $t$ for all $b\in B$,
    \item for any subset of frontends $P\subseteq F$,
    \[\sum_{f\in P} \lambda_f \leq \sum_{b\in B \cap (\cup_{f\in P}\mathcal{B}(f))} \mu_b(N_b(t)).\]
\end{itemize}

\end{lemma}

\subsubsection{Analysis of the Lyapunov Function}

In the following, given $B\subseteq \mathcal{B}$, we define
\[
    V_B(\bN, \bx) = \sum_{b\in B}\left|\sum_{f\in \mathcal{F}(b)}\lambda_f x_{f,b} - \mu_b(N_b)\right|\,,
\]
i.e., the sum of absolute values of drifts of backends within the set $B$, which we refer to as \textit{Total Absolute Drifts}.

Before proving the general result that the value of the Lyapunov function decreases over time, we first focus on two specific cases: (1) the \textit{Constant Tier} case, where $V_B$ is shown to decrease during intervals when the tier structure remains unchanged, and (2) the \textit{Single-Tier Splitting} case, where $V_B$ decreases when a single tier $(F, B)$ splits into multiple tiers. These cases help build the foundation for the more general result, where we show that the Lyapunov function decreases when several tiers merge and subsequently split, which we refer to as the \textit{Tiers Reconfiguration} case.

\paragraph{\textbf{Constant Tier Case}}
{When the tier structure does not change, by Lemma~\ref{lem: tier N same direction}, the workloads at the backends change in the same direction, thus the total absolute drifts equals the absolute total drifts, i.e., the absolute value of total flow imbalance. Note that flow imbalance drives workloads towards equilibrium at which flow is balanced, with jobs arriving at the same rate as jobs departing. Therefore, the absolute flow imbalance decreases.}
\begin{lemma}\label{lem: slide}
Suppose there exists $t_1 < t_2$ and $F \subseteq \mathcal{F}, B \subseteq \mathcal{B}$ such that for $t\in [t_1, t_2]$, $(F,B)$ is a tier at $\bN(t)$.
Then, the total absolute drifts for the backends in $B$ satisfy
\[V_B(\bN(t), \bx(t)) =  \left|\sum_{f\in F} \lambda_f - \sum_{b\in B}\mu_b(N_b(t))\right|,\]
and $V_B(\bN(t), \bx(t))$ is {non-increasing} in $[t_1,t_2]$.
\end{lemma}

\paragraph{\textbf{Single-Tier Splitting Case}}
{When a tier $(F,B)$ splits into multiple tiers at time $\tau$, we show that $(F,B)$ and the resulting tiers will all have the same sign of flow imbalance. The proof relies on the flow imbalance inequalities presented in Lemma \ref{lem: K+ lambda big} and Lemma \ref{lem: K- lambda small} and analyzing the structure of the TierGraph induced by the tiers after splitting. Specifically, if $(F,B)$ has a positive flow imbalance, by the flow imbalance inequalities in Lemma \ref{lem: K+ lambda big}, we can show that the tier that corresponds to a source vertex in the TierGraph must also have positive flow imbalance. Then by the fundamental properties of TierGraph, we can show other tiers must also have positive flow imbalance.
Therefore, we can apply Lemma \ref{lem: slide} to show that $V_B(\bN(\tau^-), \bx(\tau^-)) = V_B(\bN(\tau^+), \bx(\tau^+))$.}

\begin{lemma}\label{lem:split}
Suppose there exists $t_1 < \tau < t_2$ and $F \subseteq \mathcal{F}, B \subseteq \mathcal{B}$ such that
\begin{itemize}
    \item $(F,B)$ is a tier at $\bN(t), t\in[t_1,\tau)$,
    \item $(F,B)$ splits into ${\hat{k}}>1$ tiers $(\hat{F}_1, \hat{B}_1), \ldots, (\hat{F}_{\hat{k}}, \hat{B}_{\hat{k}})$ at $\bN(t), t\in (\tau, t_2]$,
    \item $F = \hat{F}_1 \cup \cdots \cup \hat{F}_{\hat{k}}$ and $B = \hat{B}_1 \cup \cdots \cup \hat{B}_{\hat{k}}$.
\end{itemize}
Then, the total absolute drifts for the backends in $B$ satisfy
\[V_B(\bN(t), \bx(t)) = \left|\sum_{f\in F} \lambda_f - \sum_{b\in B}\mu_b(N_b(t))\right|,\]
and $V_B(\bN(t), \bx(t))$ is {non-increasing} in $[t_1,t_2]$.
\end{lemma}

Having established the behavior of the Lyapunov function during a single-tier split, we now consider the most general case, where the tiers reconfigure by first merging and then splitting into multiple tiers.

\paragraph{\textbf{Tiers Reconfiguration Case.}}
{
If, prior to the reconfiguration, all tiers have the same sign of flow imbalance—whether positive or negative—the proof is similar as in Lemma \ref{lem:split}.
The challenge arises when the tiers involved in the reconfiguration have opposing flow imbalances, i.e., some tiers have positive flow imbalances while others have negative flow imbalances. In this case, the dynamics become more complex, as the connectivity graph structure, arrival rates, and current workloads at the time of reconfiguration all influence the system's dynamics.

In this complex case, the proof relies on splitting the frontends and backends in $(F,B)$ into subsets based on the local flow imbalance each of them experiences. This is done separately for $[t_1, \tau)$ and $(\tau, t_2]$ so we will have two subsets of frontends and two subsets of backends defined based on $\bN(t)$ at $[t_1, \tau)$ and the same for $\bN(t)$ at $(\tau, t_2]$. We present an explicit expression of $V_B(\bN(\tau^-), \bx(\tau^-)) - V_B(\bN(\tau^+), \bx(\tau^+))$ that is based on these eight subsets, and show this difference is non-negative by analyzing the structure of the TierGraph induced by the tiers prior to reconfiguration and after reconfiguration.
}

\begin{lemma}\label{lem:sync}
Suppose there exists $t_1 < \tau < t_2$ and $F_1, \ldots, F_k \subseteq \mathcal{F}, B_1,\ldots,  B_k \subseteq \mathcal{B}$ such that
\begin{itemize}
    \item $(F_1,B_1), \ldots, (F_k, B_k)$ are tiers at $\bN(t), t\in[t_1,\tau)$,
    \item $\mu'_{B_1}(\bN(\tau)) = \mu'_{B_2}(\bN(\tau)) =\cdots = \mu'_{B_k}(\bN(\tau))$,
    \item $(F_1,B_1), \ldots, (F_k, B_k)$ split into $(\hat{F}_1,\hat{B}_1), \ldots, (\hat{F}_{\hat{k}}, \hat{B}_{\hat{k}})$ that are tiers at $\bN(t), t\in (\tau, t_2]$,
    \item $F:= F_1 \cup \cdots \cup F_k = \hat{F}_1 \cup \cdots \cup \hat{F}_k$ and $B:= B_1 \cup \cdots \cup B_k =\hat{B}_1 \cup \cdots \cup \hat{B}_k$,
\end{itemize}

Then, the total absolute drifts for the backends in $B$, $V_B(\bN(t), \bx(t))$ is {non-increasing} in $[t_1,t_2]$.
\end{lemma}

Combining the three cases, we have shown that the Lyapunov function consistently decreases over time until the system state reaches the equilibrium $\bN^*$.

\subsubsection{Putting Things Together}
We have shown that $V_B(\bN(t), \bx(t))$, the total absolute drifts of the backends within a tier $(F,B)$, is proportional to the total flow imbalance within the tier when the tier structure remains constant. Thus, when the flow is balanced at the tier, $V_B(\bN(t), \bx(t))$ can remain constant. Note in the statements of Lemma \ref{lem: slide}, Lemma \ref{lem:split} and Lemma \ref{lem:sync}, we use ``non-increasing."
Nevertheless, as the flow imbalances are zero for all tiers only at the equilibrium (formalized in the following corollary), this implies that when $\bN(t) \neq \bN^*$, there exists at least one tier at $\bN(t)$ with non-zero flow imbalance. Therefore, the Lyapunov function $V(\bN(t), \bx(t))$, which is the sum of $V_B(\bN(t), \bx(t))$ over all tiers at $\bN(t)$, is strictly decreasing until $\bN(t) = \bN^*$.

\begin{corollary}\label{cor: non N star drift positive}
If there exists $F_1, \ldots, F_k \subseteq \mathcal{F}, B_1, \ldots, B_k \subseteq \mathcal{B}$ such that
\begin{itemize}
    \item $(F_1,B_1), \ldots, (F_k, B_k)$ are tiers at $\bN$,
    \item $\mathcal{F}= F_1 \cup \cdots \cup F_k$ and $\mathcal{B}= B_1 \cup \cdots \cup B_k $,
    \item $\sum_{f\in F_i} \lambda_f = \sum_{b\in B_i} \mu_b(N_b)$ for all $i\in [k]$.
\end{itemize}
Then $\bN = \bN^*$, i.e., $\bN$ is the equilibrium.
\end{corollary}

\subsection{Convergence Rate}\label{sec: convergence rate}
In this section, we discuss the convergence rate of the GMSR policy. Notably, the Euclidean distance between the workloads at time $t$, $\bN(t)$, and the equilibrium workloads $\bN^*$ does not decrease monotonically with $t$.
Otherwise, we could use this distance as our Lyapunov function.
Nevertheless, we can show the total workloads $\sum_{b\in \mB}N_b(t)$ converge exponentially fast to the optimum $\sum_{b\in \mB}N_b^*$ when $\bN(t)$ reaches an invariant set to be defined later. A set is said to be invariant if once the state enters the set, it never leaves.
A challenging part of our analysis is dealing with the fact that the Lyapunov function depends on both the system state $\bN$ and the routing matrix, $\bx$, whose dependence on $\bN$ is only semi-continuous. In particular, if the optimal solution lies in an equal gradient hypersurface, there are workloads arbitrarily close yet outside of this hypersurface that induce large drifts under the GMSR policy.

To establish the convergence rate, we first define the invariant set $K$, within which the marginal service rates of all backends are bounded from below. Within the set, the Lyapunov function converges to zero exponentially fast (see Proposition \ref{prof: exponential convergence V in K}), and this convergence rate translates to exponential convergence of $\left|\sum_{b\in \mathcal{B}} N_b(t) - \sum_{b\in \mathcal{B}} N_b^*\right|$, as we show this difference can be bounded by the Lyapunov function (see Proposition \ref{prop:exponential-convergence-sum-rifs}). If the system state is outside of $K$, we show at least one backend's workload decreases linearly (see Proposition \ref{prop: linear converges to K}), pushing the system towards $K$ in finite time.

We begin by showing that there exists a workload vector $\tilde{\bN}$ where 1) all backends' marginal service rates are lower bounded by a constant, and 2) capacity slackness holds.
\begin{lemma}[Capacity Slack]\label{lem: capacity slack}
  Under Assumption \ref{assmp: increasing concave smooth ell} and Assumption \ref{assmp: system stable}, there exist constants $\kappa>0$ and $\Delta >0$ that are independent of the initial workload, such that for every subset $P \subseteq \mathcal{F}$ of frontends we have
    \[\sum_{f\in P} \lambda_f + \Delta \leq \sum_{b\in (\cup_{f\in P} \mathcal{B}(f))} \mu_b(\tilde{N}_b),\]
    where $\tilde{N}_b$ is the unique workload satisfying $\mu_b'(\tilde{N}_b) = \kappa$.
\end{lemma}

With the existence of $\kappa$ established, we define an invariant set within which the marginal service rates are bounded from below.
\begin{proposition}\label{prop: K invariant set}
    Define $K = \{\bN \in \mathbb{R}^{|\mathcal{B}|}_+: \mu_b'(N_b) \geq \kappa, \forall b\in \mathcal{B}\}$, then $K$ is an invariant set, i.e., for all $\bN_0 \in K$, all the trajectories of the differential inclusion \eqref{eq: differential inclusion} with $\bN(0) = \bN_0$, remain in $K$.
\end{proposition}
The invariant set can be equivalently defined as  $K = \{\bN \in \mathbb{R}^{|\mathcal{B}|}_+: N_b \leq \tilde{N}_b, \forall b\in \mathcal{B}\}$.
To see why $K$ is invariant, suppose by contradiction that some backend $b_1$ escapes $K$, so its workload increases above $\tilde{N}_{b_1}$. Since $\mu_{b_1}'$ is strictly decreasing, $b_1$'s gradient falls below $\kappa$ while all backends still in $K$ have gradient $\geq \kappa$. The escaping backend's tier must therefore have all backends at gradient $\leq \kappa$, implying all are at workload $\geq \tilde{N}_b$. This tier has positive flow imbalance ($\sum \lambda_f > \sum \mu_b(\tilde{N}_b)$), contradicting the capacity slack established in Lemma~\ref{lem: capacity slack}.

{The following result establishes that the Lyapunov function $V(\bN(t), \bx(t))$ converges to zero exponentially fast in time $t$ if the system starts in the invariant set $K$.}
\begin{proposition}\label{prof: exponential convergence V in K}
  If $\bN(0) \in K$, the value of the Lyapunov function $V(\bN(t),\bx(t))$ converges to zero exponentially fast for $t\geq 0$, i.e.,
  \[\frac{d}{dt}V(\bN(t), \bx(t)) \leq -\kappa V(\bN(t), \bx(t)) {\text{ almost everywhere.}}\]
\end{proposition}

{Next, we establish the relationship between 1) the Lyapunov function, and  2) the absolute difference between the total workloads and the optimal workloads when the system state lies in the invariant set $K$.}
\newcommand{\invrate}{r}
\newcommand{\inflow}{w}
\newcommand{\binflow}{\bm{w}}
\begin{proposition}
    \label{prop:exponential-convergence-sum-rifs}
    For all $(\bN,\bx)$ with $\bN \in K$ and $\bx_f \in X_f(\bN)$ for all $f \in \mathcal{F}$,
    we have
    $$ \left|\sum_{b \in \mathcal{B}} N_b \; - \; \sum_{b \in \mathcal{B}} N^*_b \right| \leq \frac{V(\bN,\bx)}{\kappa} . $$
\end{proposition}
{Therefore, Proposition \ref{prof: exponential convergence V in K} and Proposition \ref{prop:exponential-convergence-sum-rifs} together imply that the total workloads converge exponentially fast to the optimal workloads if the system state lies in the invariant set $K$.}

Lastly, we show that if the system state starts outside of the invariant set $K$, it approaches $K$ in finite time.
\begin{proposition}\label{prop: linear converges to K}
If $\bN(0) \notin K$, there exists 
\[T \leq \frac{\sum_{b\in \mB} \max\{N_b(0) - \tilde{N}_b, 0\}}{\min\{\Delta,  \min_{b\in \mB} \mu_b(\tilde{N}_b)\}}\] such that $\bN(t)$ enters the invariant set $K$ no later than $T$.
\end{proposition}

\section{Overloaded System in Fluid Regime Under GMSR}\label{sec: overload system}
In this section, we extend the fluid-regime analysis to scenarios where the system may not be strictly feasible, relaxing Assumption~\ref{assmp: system stable}.

When the system is overloaded, total workload inevitably grows unbounded. Therefore, we focus on characterizing the equilibrium in terms of the limiting service rates $\lim_{t\to\infty} \mu_b(N_b(t))$ for every $b\in \mB$.
Consider the following optimization problem that maximizes system throughput:
\begin{align}
\textsf{OPT-TP} =    \max_{\bm{L}, \bx} \quad & \sum_{b\in \mathcal{B}} L_b \label{eq: FLU TP}\\
    \text{s.t.} \quad & \sum_{f\in \mathcal{F}} \lambda_f x_{f,b} \geq L_b , \forall b\in\mathcal{B}, \notag\\
    & L_b \leq \mu_b(\infty), \forall b\in \mB, \notag\\
    & \sum_{b\in \mathcal{B}}x_{f,b} = 1, \forall f \in \mathcal{F}, \notag\\
    & x_{f,b} \geq 0, \forall (f,b) \in \mathcal{E}, \notag \\
    & x_{f,b} = 0, \forall (f,b) \notin \mathcal{E}. \notag
\end{align}
Here $L_b$ denotes the service rate at each backend $b\in \mathcal{B}$, and $x_{f,b}$ denotes the proportion of jobs sent from frontend $f$ to backend $b$. The objective is to maximize the total service rates across all backends.
The first constraint, combined with the objective function, ensures that $L_b$ is set as close to the total incoming flow as possible, while not exceeding the service rate upper bound $\mu_b(\infty)$ as given by the second constraint.
The third constraint imposes flow balance at each frontend, i.e., all jobs are sent to backends. The last two sets of constraints impose the connectivity constraints. This optimization problem is always feasible; for example, a feasible solution is $\bm{L}=\bm{0}$, $x_{f,b} = 1/|\mB(f)|$ for $(f,b) \in \mathcal{E}$.
This formulation can be equivalently recast as minimizing the total queue overflow (i.e., $\sum_{f\in \mF}\lambda_f x_{f,b}-L_b$), a metric studied in previous works \citep{georgiadis2006optimal,li2014dynamic,wu2022overload}.

When the system is not overloaded, i.e., the fluid optimization problem \eqref{eq: FLU} is feasible, all policies that can stabilize the system also maximize throughput, which equals the total arrival rates. That is, all feasible solutions to \eqref{eq: FLU} are optimal to \eqref{eq: FLU TP}. However, when the system cannot be fully stabilized, the optimal service rate vector saturates some backends at their capacity $\mu_b(\infty)$.

We now establish the optimality of GMSR through several key results.
\begin{lemma}\label{lem: OPT TP bounds service rate}
For any online policy $\pi$,
    \[\textsf{OPT-TP}  \geq \limsup_{k\to\infty}\frac{1}{k}\sum_{i=0}^{k-1}\mathbb{E}_{\pi}\left[\sum_{b\in \mB}\mu_b(N_b(i))\right],\]
where $N_b(i)$ denotes the number of jobs at backend $b$ at time period $i$ in the discrete-time stochastic model.
\end{lemma}
This lemma establishes that $\textsf{OPT-TP}$ provides an upper bound on the throughput achievable by any online policy. We next show that GMSR can achieve this maximum throughput in the fluid model.
\begin{proposition}[Throughput Optimality]\label{pro: maximum throughput}
Each solution $\bN(\cdot)$ of the differential inclusion with any initial state $\bN(0)$ satisfies $\lim_{t\to\infty} \mu_b(N_b(t)) = L_b^*$ for every $b\in \mB$,
where $\bm{L}^*$ is an optimal solution to $\textsf{OPT-TP}$.
\end{proposition}
We defer the precise values of the equilibrium service rates $\bm{L}^*$ under GMSR to Theorem \ref{thm: general N convergence result} in Appendix~\ref{apx: overload system}, and focus on discussing the properties of $\bm{L}^*$ here.
Note that multiple optimal service rates typically exist for \eqref{eq: FLU TP}, as the objective function only concerns maximizing total throughput. Therefore, the marginal impact on the throughput of routing one additional unit of job to any backend (when service capacity permits) is the same across backends, creating flexibility in how jobs are distributed among backends. This is in stark contrast to \eqref{eq: FLU}, where the marginal impact is different due to the concavity of service rate functions, forcing a unique optimal workload.
Among all optimal solutions to \eqref{eq: FLU TP}, our policy GMSR will induce the system to converge to one satisfying the following two properties.

First, GMSR stabilizes the largest number of backends among all throughput-optimal policies at equilibrium.
\begin{proposition}[Maximal Stabilized Backends]\label{pro: maximum stabilizing backend}
Let $\mathcal{L}^*$ denote the set of optimal $\bm{L}$ to the optimization problem \eqref{eq: FLU TP}, and $B(\bm{L}): = \{b\in \mB: L_b <\mu_b(\infty)\}$ denote the set of backends that are stabilized. We have for every $\bm{L} \in \mathcal{L}^*$, $B(\bm{L})\subseteq B(\bm{L}^*)$.
\end{proposition}

Next, GMSR minimizes the total workload at stabilized backends, among all policies
that achieve maximum throughput and maximum number of stabilized backends at equilibrium.
\begin{proposition}[Least Workload at Stabilized Backends]\label{prof: minimum stabilized workload}
For every $\bm{L} \in \mathcal{L}^*$ with $B(\bm{L}) = B(\bm{L}^*)$,
    \[\sum_{b\in B(\bm{L}^*)} \mu_b^{-1}(L_b^*) \leq \sum_{b\in B(\bm{L})} \mu_b^{-1}(L_b).\]
\end{proposition}

Detailed analyses can be found in Appendix~\ref{apx: overload system}; here we provide the high-level ideas.
Our analyses build on a concept called \textit{stability decomposition} (Definition \ref{def: stability decomposition}) which decomposes frontends and backends based on the graph structure, arrival rates, and service capacities. This decomposition identifies which backends' workloads will inevitably grow without bound and which ones can remain stable under GMSR. We show this decomposition is equivalent to a minimum s-t cut on the graph augmented by a source connecting to all frontends and a sink connecting to all backends (Lemma \ref{lem: stability decomposition is a cut}), and we show it is unique (Corollary \ref{cor: unique stability composition}) and can be constructed using maximum flow algorithms (Corollary \ref{cor: construct mF using max flow}).

We then use the same Lyapunov function as in the stability analysis to show that the service rates in the system converge to an equilibrium $\bm{L}^*$ (Theorem \ref{thm: general N convergence result}), where the system splits into two parts:
\begin{itemize}
    \item A subset of backends that can be stabilized, with total workload converging to a finite value;
    \item A subset of backends whose workloads grow to infinity, with total service capacity converging to the maximum possible level, i.e., $\mu_b(\infty)$.
\end{itemize}

Finally, we use the properties of the stability decomposition to show the three properties of the equilibrium service rate $\bm{L}^*$ described in the propositions above.

\section{{Simulation of Policy Performances}}\label{sec: randomized simulation}
Figure~\ref{fig:intro-latency-utilization} in the introduction reports a representative instance with clear gains from GMSR. We next report simulations for additional randomly generated setups. In each setup, the numbers of frontends and backends are sampled uniformly from $\{1,\ldots,10\}$ and $\{2,\ldots,20\}$, respectively. We consider both complete graphs, where every frontend connects to every backend, and sparse graphs, where each edge is included independently with probability $0.5$ and isolated nodes are repaired by adding one uniformly random incident edge. Motivated by Universal Scalability Law \citep[Chapter~6]{gunther2007guerrilla}, backend service rates are
\[
    \mu_b(N)=\frac{N}{N+\alpha_b},
\]
where $\lfloor |\mB|/2\rfloor$ backends have $\alpha_b$ sampled uniformly from $\{1,\ldots,10\}$ and the rest from $\{10,\ldots,20\}$. Frontend arrival weights are sampled uniformly from $\{1,\ldots,10\}$. For complete graphs, the total arrival rate is $\rho\sum_{b\in \mB} \mu_b(\infty) = \rho|\mB|$, where $\rho$ denotes utilization, and this total is split in proportion to these weights. For sparse graphs, we scale the weights so that bottleneck utilization equals $\rho$, where bottleneck utilization is the maximum, over nonempty subsets $P\subseteq\mF$, of arrivals into $P$ divided by the limiting capacity ($\mu_b(\infty)$) of the backends reachable from $P$.

Arrivals to each frontend are independent Poisson random variables. If backend $b$ has current workload $N_b$, its departures are Binomial with parameters $N_b$ and $1/(N_b+\alpha_b)$, so the expected number of departures equals $\mu_b(N_b)$. We compare JSQ, Generalized JSQ, and GMSR against the fluid benchmark defined in \eqref{eq: FLU}. Each setup-utilization pair is simulated for 25{,}000 periods with 8{,}000 warmup periods, using one sample path and common frontend arrival paths across policies. We choose the initial state to be the fluid-optimal workloads (rounded to the nearest integer) to focus on the steady-state performance of the policy as opposed to transient behavior.

Figure~\ref{fig:rational-service-two-family-excess} reports aggregate results over 1{,}000 systems per graph family. At $\rho=0.9$, GMSR reduces latency by 6.9\% and 9.5\% relative to JSQ and Generalized JSQ in complete graphs, and by 5.9\% and 8.5\% in sparse graphs. Sparse graphs are harder (GMSR's mean benchmark excess is 2.7\%, versus 1.6\% for complete graphs). The smaller policy improvement is consistent with limited routing flexibility: if a frontend is connected to only a few backends, or if a local bottleneck forces traffic into a particular neighborhood, then JSQ, GJSQ, and GMSR are partly constrained to make similar allocations.

The monotone patterns in Figure~\ref{fig:rational-service-two-family-excess} reflect how the three policies allocate scarce service capacity.
For the rational service family $\mu_b(N)=N/(N+\alpha_b)$, the marginal service rate is $\alpha_b/(N+\alpha_b)^2$.
As the optimal workload equalizes marginal service rates (Lemma \ref{lem: optimal FLU structure}), $N_b+\alpha_b$ should scale with $\sqrt{\alpha_b}$ across active backends.
GMSR is aligned with this condition, thus its excess latency is largely a finite-system stochastic gap.
JSQ, by contrast, equalizes raw queue lengths $N_b$, which becomes increasingly costly at high utilization because heterogeneous backends need different workload levels to contribute comparable throughput.
Generalized JSQ roughly equalizes $N_b+\alpha_b$, which over-favors low-$\alpha_b$ backends: the marginal optimum instead allows high-$\alpha_b$ backends to carry larger adjusted workloads. By steering too much traffic away from high-$\alpha_b$ backends, Generalized JSQ can therefore perform worse than JSQ. At high utilization, queues grow and the additive $\alpha_b$ term becomes less important, so Generalized JSQ increasingly resembles JSQ.
Overall, GMSR and JSQ reflect marginal-rate versus queue-length balancing, while the performance of GJSQ is more service-family dependent.

\begin{figure}[H]
    \centering
    \setlength{\abovecaptionskip}{2pt}
    \setlength{\belowcaptionskip}{0pt}
    \includegraphics[width=0.8\linewidth]{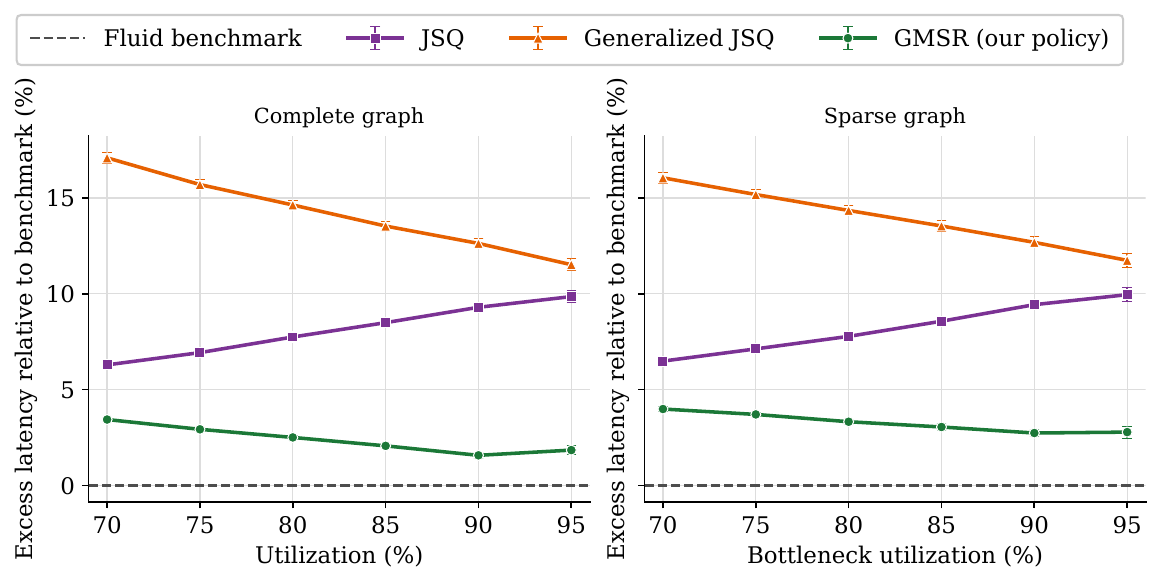}
    \caption{Policy comparison under complete and sparse compatibility graphs. Each panel reports mean excess latency relative to the fluid benchmark over 1{,}000 systems, with error bars denoting standard errors across setups.}
    \label{fig:rational-service-two-family-excess}
\end{figure}

\section{Conclusions}\label{sec: conclusions}
In this paper, we investigate distributed load balancing for general bipartite queueing systems. We propose the Greatest Marginal Service Rate policy that routes jobs based on the marginal service rates of connected backends. Our analyses show that GMSR minimizes expected latency asymptotically, matching the performance of a centrally coordinated optimal routing policy.

We believe this work opens new avenues for future research. In some settings, it could take time for jobs to travel from frontends to backends and for backends to communicate their state to the frontends. How can we design latency-aware routing policies that are robust to this delay? In practice, there are often multiple layers of routers, e.g., each backend also has a local router that decides how to dispatch jobs to individual physical machines---how to design multi-layer service networks and the routing policies?
Lastly, in our model, we assume the service rate functions are known to the decision-maker, but in reality, these functions need to be estimated. How to route while learning is a practically relevant and theoretically interesting future research direction.

\bibliographystyle{informs2014}
\bibliography{references}

\appendix
\section{Failure of Routing based on Expected Latency}\label{apx: failure of expected latency}
A natural idea is to route the job to a connected backend with the shortest expected latency. However, this approach can lead to suboptimal outcomes. Let's illustrate this using the famous Pigou's Example~\citep{pigou2017economics}.
Suppose there is one frontend connected to two backends. The arrival rate to the frontend is $1$. If we route $x_1$ proportion of jobs to backend 1, the job's expected latency is $1$, regardless of $x_1$. For backend 2, the expected latency equals the proportion of jobs routed there, which is denoted by $x_2$. Therefore, the total expected latency is $x_1\cdot 1 + x_2\cdot x_2$ subject to $x_1, x_2 \geq 0$ and $x_1+x_2 = 1$. 
The optimal solution is $x_1 = x_2 = 1/2$, which yields the minimal total expected latency of $3/4$. However, if the frontend routes jobs based on expected latency, then $x_1 = 0, x_2 = 1$, and the resulting expected latency is $1$.

Moreover, the gap can be made arbitrarily large. Suppose backend 2's expected latency is $x_2^p$ for an integer $p \geq 1$ (the case above is $p=1$). Routing based on expected latency still sends all jobs to backend 2---since $x_2^p < 1$ whenever $x_2 < 1$---yielding total expected latency $1$ for every $p$. The optimal routing instead sets $x_2 = (p+1)^{-1/p}$, with total expected latency $1 - \frac{p}{p+1}(p+1)^{-1/p} \to 0$ as $p \to \infty$. Hence routing based on expected latency can be arbitrarily suboptimal, a manifestation of the unbounded price of anarchy \citep{roughgarden2002bad}. In systems with multiple frontends operating independently, the problem becomes even more complex.

\section{Fluid Optimization Problem}
\subsection{Proof of Lemma \ref{lem: OPT lower bound}}

\begin{proof}
Fix an arbitrary online policy $\pi$. For simplicity, we drop the subscript $\pi$ from $\mathbb{E}_\pi$ and write $\mathbb{E}$ throughout this proof.
For any subset $P\subseteq \mathcal{F}$, let $Q := \cup_{f\in P} \mathcal{B}(f)$ denote the set of backends connected to frontends in $P$.

By the system dynamics, the number of jobs at frontend set $P$ and at the backend set $Q$ at time period $k$ is given by 
\begin{align*}
    \sum_{f\in P}G_f(k) + \sum_{b \in Q }N_b(k) & \stackrel{(a)}{=} \sum_{f\in P}G_f(0) + \sum_{i=0}^{k-1} \sum_{f\in P} \left( W_f(i) - \sum_{b\in \mathcal{B}(f)} A_{f,b}(i)\right) \\
    & \quad \quad + \sum_{b \in Q }N_b(0) + \sum_{i=0}^{k-1} \sum_{b \in Q }\left( \sum_{f\in \mathcal{F}(b)} A_{f,b}(i) - D_b(i) \right) \\
    & = \sum_{f\in P}G_f(0) + \sum_{b \in Q }N_b(0) + \sum_{i=0}^{k-1} \left(\sum_{f\in P}  W_f(i) - \sum_{b \in Q } D_b(i) \right) \\
    & \quad \quad + \sum_{i=0}^{k-1} \left( \sum_{b \in Q } \sum_{f\in \mathcal{F}(b)} A_{f,b}(i) -\sum_{f\in P}\sum_{b\in \mathcal{B}(f)} A_{f,b}(i) \right) \\
    & \stackrel{(b)}{\geq} \sum_{i=0}^{k-1} \left(\sum_{f\in P}  W_f(i) - \sum_{b \in Q } D_b(i) \right) +\sum_{i=0}^{k-1} \left( \sum_{b \in Q } \sum_{f\in \mathcal{F}(b)} A_{f,b}(i) -\sum_{f\in P}\sum_{b\in \mathcal{B}(f)} A_{f,b}(i) \right) \\
    & \stackrel{(c)}{\geq} \sum_{i=0}^{k-1} \left(\sum_{f\in P}  W_f(i) - \sum_{b \in Q } D_b(i) \right),
\end{align*}
where $(a)$ follows from the system dynamic recursions; $(b)$ follows as $G_f(0), N_b(0) \geq 0$; $(c)$ follows as 1) $A_{f,b}(i) \geq 0$ for $i\geq 1$, and 2) $Q =  \cup_{f\in P} \mathcal{B}(f)$, which means $P \subseteq\cup_{b\in Q} \mathcal{F}(b)$, thus $\{(f,b): f\in P, b \in Q \} \subseteq \{(f,b): f\in \cup_{b\in Q} \mathcal{F}(b), b\in Q\}$. 

Dividing both sides by $k$ and taking expectation and $\limsup$ over both sides, we obtain
\begin{align}
    \limsup_{k\to\infty} & \frac{1}{k}\mathbb{E}\left[\sum_{f\in P}G_f(k)+ \sum_{b \in Q}N_b(k)\right] \notag \\
    &\geq  \limsup_{k\to\infty}  \frac{1}{k} \mathbb{E}\left[\sum_{i=0}^{k-1} \left( \sum_{f\in P} W_f(i) - \sum_{b \in Q} D_b(i) \right)\right] \stackrel{(a)}{=}\limsup_{k\to\infty} \left(\sum_{f \in P}\lambda_f -  \frac{1}{k} \sum_{b \in Q}\mathbb{E}\left[\sum_{i=0}^{k-1} D_b(i)\right]\right)  \notag \\
   &  \stackrel{(b)}{\geq}  \limsup_{k\to\infty} \left(\sum_{f \in P}\lambda_f - \sum_{b \in Q}\mu_b \left(\mathbb{E}\left[\frac{1}{k}\sum_{i=0}^{k-1} N_b(i)\right]\right)\right)\stackrel{(c)}{=} \sum_{f \in P}\lambda_f + \limsup_{k\to\infty} \left(- \sum_{b \in Q}\mu_b \left(\mathbb{E}\left[\frac{1}{k}\sum_{i=0}^{k-1} N_b(i)\right]\right)\right) \notag\\
   & = \sum_{f \in P}\lambda_f - \liminf_{k\to\infty} \left(\sum_{b \in Q}\mu_b \left(\mathbb{E}\left[\frac{1}{k}\sum_{i=0}^{k-1} N_b(i)\right]\right)\right) \stackrel{(d)}{\geq}   \sum_{f \in P}\lambda_f -  \sum_{b \in Q}\mu_b \left(\liminf_{k\to\infty}\mathbb{E}\left[\frac{1}{k}\sum_{i=0}^{k-1} N_b(i)\right]\right)   \label{eq: bound liminf Qn Nn}
\end{align}
where $(a)$ holds as $W_f(i)$ has mean $\lambda_f$ for all $i\geq 0$; $(b)$ follows from \eqref{eq: bound D using ellN} below; $(c)$ follows from Lemma \ref{lem: liminf ineq bounded} as $\mu_b(\cdot)$ is bounded; $(d)$ follows from Lemma \ref{lem: fliminf > liminff}. 
To see $(b)$ holds, note that
\begin{align}
\mathbb{E}\left[\frac{1}{k}\sum_{i=0}^{k-1} D_b(i)\right]  & \stackrel{(a)}{=} \mathbb{E}\left[\frac{1}{k}\sum_{i=0}^{k-1}\mathbb{E}\left[ D_b(i)\mid N_b(i)\right]\right] \stackrel{(b)}{=}  \mathbb{E}\left[\frac{1}{k}\sum_{i=0}^{k-1}  \mu_b(N_b(i))\right] \\
& \stackrel{(c)}{\leq} \mathbb{E}\left[\mu_b \left(\frac{1}{k}\sum_{i=0}^{k-1} N_b(i)\right)\right] \stackrel{(d)}{\leq} \mu_b \left(\mathbb{E}\left[\frac{1}{k}\sum_{i=0}^{k-1} N_b(i)\right]\right). \label{eq: bound D using ellN}
\end{align}
where $(a)$ follows from the tower property of conditional expectation; $(b)$ follows from $\mathbb{E}[D_b(i)\mid N_b(i)] = \mu_b(N_b(i))$; $(c)$ and $(d)$ both follow from Jensen's inequality as $\mu_b(\cdot)$ is concave.

Let $\bar{N}_b := \liminf_{k\to\infty}\mathbb{E}\left[ \sum_{i=0}^{k-1} N_b(i)\right]/k$, if $\sum_{f\in P}\lambda_f >  \sum_{b \in Q}\mu_b \left(\bar{N}_b\right)$, then from \eqref{eq: bound liminf Qn Nn},
\[\limsup_{k\to\infty} \frac{1}{k}\mathbb{E}\left[\sum_{f\in P}G_f(k)+ \sum_{b \in Q}N_b(k)\right] \geq \sum_{f\in P}\lambda_f - \sum_{b\in Q}\mu_b(\bar{N}_b) > 0. \]
Since $G_f(k) \geq 0$, this implies $\limsup_{k\to\infty} \mathbb{E}\left[\sum_{b \in Q}N_b(k)\right]/k > 0$, which means $\sum_{b\in Q} N_b(k)$ is not mean rate stable (Definition \ref{def: mean rate stable}). By Theorem 2.8 from \citet{neely2022stochastic}, as the expected arrivals are bounded by a constant, mean rate stability is necessary for strong stability (Definition \ref{def: strongly stable}).

Therefore, a necessary condition for the system to be strongly stable is for any $P\subseteq \mathcal{F}$,
\begin{align}\label{eq: stability necessary condition}
    \sum_{f\in P}\lambda_f \leq  \sum_{b \in Q}\mu_b \left(\bar{N}_b\right).
\end{align}

We proceed to show that for any $\{\bar{N}_b\}_{b\in \mathcal{B}}$ that satisfies \eqref{eq: stability necessary condition}, there exists a routing matrix $\{x_{f,b}\}_{f\in \mathcal{F}, b\in \mathcal{B}}$ such that
$(\{\bar{N}_b\}_{b\in \mathcal{B}}, \{x_{f,b}\}_{f\in \mathcal{F}, b\in \mathcal{B}})$ is a feasible solution to the following optimization problem:
    \begin{align}
    \textsf{OPT' =}    \min_{\bN, \bx} \quad & \sum_{b\in \mathcal{B}} N_b \label{eq: FLU ineq relax}\\
    \text{s.t.} \quad & \sum_{f\in \mathcal{F}} \lambda_f x_{f,b} \leq \mu_b(N_b), \forall b\in\mathcal{B}, \notag\\
    & \sum_{b\in \mathcal{B}}x_{f,b} = 1, \forall f \in \mathcal{F}, x_{f,b} \geq 0, \forall (f,b) \in \mathcal{E}, x_{f,b} = 0, \forall (f,b) \notin \mathcal{E}. \notag
    \end{align}
    
To see this, consider a graph with a source connected to all frontends with capacities $\lambda_f$ and a sink connected to all backends with capacities $\mu_b(\bar N_b)$, and the arcs $\mathcal{E}$ having infinite capacities. Note that there are three types of cuts with finite capacity in this network: 1) cuts involving only the edges that connect the source to all front-end nodes, with capacity $\sum_{f\in \mathcal{F}} \lambda_f$; 2) cuts involving only the edges that connect all back-end nodes to the sink $\sum_{b\in \mathcal{B}} \mu_b(\bar{N_b})$; 3) cuts that combine edges from both the first and second types; that is, cuts formed by selecting edges from the source-to-$\mF\setminus P$ connections and some from the $ Q$-to-sink connections, see Figure \ref{fig:maxflowmincut}. Here $P$ is an arbitrary subset of $\mathcal{F}$ and $Q := \cup_{f\in P} \mathcal{B}(f)$ as before.
The capacity of the third type of cut is 
\[\sum_{f\in \mathcal{F} \setminus P} \lambda_f + \sum_{b\in Q} \mu_b(\bar{N}_b) = \sum_{f\in \mathcal{F} } \lambda_f  - \sum_{f\in P} \lambda_f + \sum_{b\in Q} \mu_b(\bar{N}_b)  \geq \sum_{f\in \mathcal{F}}\lambda_f.\]

Therefore, by the Max-Flow Min-Cut theorem, the max flow is $\min\{\sum_{f\in \mathcal{F}} \lambda_f, \sum_{b\in \mathcal{B}} \mu_b(\bar{N_b})\} = \sum_{f\in \mathcal{F}} \lambda_f$. Therefore, a feasible routing matrix $\{x_{f,b}\}_{f\in \mathcal{F}, b\in \mathcal{B}}$ exists. 
Putting things together,
\begin{align*}
    \textsf{OPT} & \stackrel{(a)}{=} \textsf{OPT}' \stackrel{(b)}{\leq} \sum_{b\in \mB} \bar{N}_b \stackrel{(c)}{=} \sum_{b\in \mB} \liminf_{k\to\infty}\frac{1}{k}\mathbb{E}\left[\sum_{i=0}^{k-1} N_b(i)\right] \\
    &  \stackrel{(d)}{\leq} \liminf_{k\to\infty}\sum_{b\in \mB}\frac{1}{k}\mathbb{E}\left[ \sum_{i=0}^{k-1}N_b(i)\right]  \stackrel{(e)}{\leq}\liminf_{k\to\infty} \frac{1}{k} \mathbb{E}\left[ \sum_{i=0}^{k-1}  \left( \sum_{f\in \mathcal{F}} G_f(i)+  \sum_{b\in \mathcal{B}}N_b(i)\right) \right].
\end{align*}
where $(a)$ follows as for any feasible solution $(\bN, \bx)$ to \eqref{eq: FLU ineq relax}, as $\mu_b$ are increasing, we can construct a feasible solution $(\bN', \bx)$ to \eqref{eq: FLU} such that $\sum_{b\in \mathcal{B}} N_b' \leq \sum_{b\in \mathcal{B}} N_b$. The results follow as the optimal solution to \eqref{eq: FLU} is also feasible to \eqref{eq: FLU ineq relax}; $(b)$ follows as $\bar{\bN}$ is feasible to the optimization problem \eqref{eq: FLU ineq relax}; $(c)$ follows from definition of $\bar{N}_b$; $(d)$ follows as for every $m\geq 1$ we have 
$\sum_{b\in \mB} \inf_{k\geq m} \frac{1}{k}\mathbb{E}\left[\sum_{i=0}^{k-1} N_b(i)\right] \leq \inf_{k\geq m} 
 \sum_{b\in \mB} \frac{1}{k}\mathbb{E}\left[\sum_{i=0}^{k-1} N_b(i)\right]$,
the inequality thus holds by taking $m\to\infty$; $(e)$ holds as $G_f(i) \geq 0$.
\Halmos \end{proof}
\begin{definition}[Definition 2.3 \citet{neely2022stochastic}, Mean Rate Stable]\label{def: mean rate stable}
A discrete-time process $N(k)$ is mean rate stable if:
$\lim_{k\to\infty} \frac{\mathbb{E}[|N(k)|]}{k} = 0$. 
\end{definition}

\begin{definition}[Definition 2.7 \citet{neely2022stochastic}, Strongly Stable]\label{def: strongly stable}
A discrete-time process $N(i)$ is strongly stable if:
$\limsup_{k\to\infty} \frac{1}{k} \sum_{i=1}^k \mathbb{E}[|N(i)|] < \infty$. 
\end{definition}


\begin{lemma}\label{lem: liminf ineq bounded}
    Given two bounded real sequence $\{x_k\}_{k\geq 1}$ and $\{y_k\}_{k\geq 1}$ such that $y_k$ converges to $y$, we have 
    \[\limsup_{k\to\infty}(x_k + y_k) = \limsup_{k\to\infty}x_k + y.\]
\end{lemma}

\begin{lemma}\label{lem: fliminf > liminff}
    Given a continuous function $g(\cdot)$, and a real sequence $\{x_k\}_{k\geq 1}$, we have 
    \[g\left(\liminf_{k\to\infty} x_k\right) \geq  \liminf_{k \to\infty} g(x_k).\]
\end{lemma}
\begin{proof}
By definition, 
    \[g\left(\liminf_{k\to\infty} x_k\right) = g\left(\lim_{K\to\infty} \inf_{k\geq K} x_k\right) = \lim_{K\to\infty}g\left(\inf_{k\geq K} x_k\right) \geq \lim_{K\to \infty} \inf_{k\geq K}g(x_k) = \liminf_{k} g(x_k). \Halmos \]
\end{proof}

\begin{figure}
    \centering
    \includegraphics[width=0.5\linewidth]{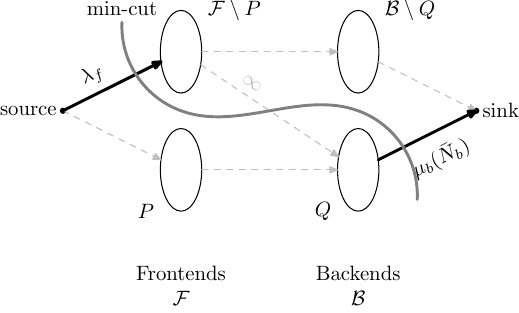}
    \caption{Illustration of the Third Type of Cut}
    \label{fig:maxflowmincut}
\end{figure}

\subsection{Proof of Lemma \ref{lem: unique N star}}

\begin{proof}
\textit{Existence.} Under Assumption~\ref{assmp: system stable}, the feasible set of~\eqref{eq: FLU} is non-empty. Since every feasible solution satisfies $\mu_b(N_b) = \sum_{f\in \mathcal{F}(b)} \lambda_f x_{f,b} \leq \sum_{f\in \mathcal{F}} \lambda_f$ and $\mu_b$ is non-decreasing with $\mu_b(0)=0$, the workloads $N_b$ are bounded above. The feasible set is therefore closed and bounded, and the continuous objective $\sum_b N_b$ attains its minimum.

\textit{Uniqueness.} Suppose $(\bm{x}, \bm{N})$ and $(\bm{x}', \bm{N}')$ are both optimal with $\bN \neq \bN'$. Since both have finite objective value, no backend is saturated, so $\mu_b^{-1}$ is well-defined on the relevant range. Consider $\tilde{\bm{x}} = ({\bm{x}+\bm{x}'})/{2}$. For all $b\in \mathcal{B}$,
\[\tilde{N}_b = \mu^{-1}_b \left(\sum_{f \in \mathcal{F}(b)} \lambda_f \frac{x_{f,b} + x'_{f,b}}{2}\right) < \frac{1}{2}\mu_b^{-1}\left(\sum_{f\in \mathcal{F}(b)} \lambda_f x_{f,b}\right) + \frac{1}{2}\mu_b^{-1}\left(\sum_{f\in \mathcal{F}(b)} \lambda_f x'_{f,b}\right) =  \frac{N_b+ N'_b}{2},\]
where the strict inequality holds for any $b$ with $N_b \neq N'_b$, since strict concavity of $\mu_b$ implies strict convexity of $\mu^{-1}_b$.
Therefore $\sum_b \tilde{N}_b < (\sum_b N_b + \sum_b N'_b)/2 = \textsf{OPT}$, contradicting the optimality of $(\bm{x}, \bm{N})$ and $(\bm{x}', \bm{N}')$.\Halmos \end{proof}

\subsection{Proof of Lemma \ref{lem: optimal FLU structure}}
\begin{proof}
    We first argue the conditions in the statement are necessary. As the service rate function is strictly increasing by assumption, it is invertible so we can rewrite the optimization problem with $\bx$ as the decision variables. Let $N_b(\bx) = \mu_b^{-1}\left(\sum_{f\in \mathcal{F}(b)} \lambda_f x_{f,b}\right)$. Pre-multiplying the flow balance constraint of frontend $f$ by $\lambda_f$, we obtain by using the method of Lagrange multipliers: 
    \[L(\bx, \bm{c}, \bm{\alpha}) = \sum_{b\in \mathcal{B}} N_b(\bx) + \sum_{f\in \mathcal{F}} c_f \lambda_f \left(1-\sum_{b\in \mathcal{B}(f)} x_{f,b}\right) - \sum_{(f,b) \in \mathcal{E}} \alpha_{f,b}x_{f,b}\,,\]
where $c_f$ are the Lagrange multipliers of the frontend flow balance constraints and $\alpha_{f,b}$ are the Lagrange multipliers of the non-negativity constraints.
    The Karush–Kuhn–Tucker conditions are 
    \begin{align*}
        &\frac{\partial N_b(\bx)}{\partial x_{f,b}} - c_f \lambda_f - \alpha_{f,b} = 0, \forall (f,b) \in \mathcal{E},\\
        &\sum_{b\in \mathcal{B}(f)}x_{f,b} = 1, \forall f \in \mathcal{F}, \\
        &\alpha_{f,b}x_{f,b} = 0,  \forall (f,b) \in \mathcal{E},\\
        &\alpha_{f,b}\geq 0, x_{f,b} \geq 0,  \forall (f,b) \in \mathcal{E},\\
        &x_{f,b} = 0, \forall (f,b) \not\in \mathcal{E}.
    \end{align*}
    By the implicit function theorem, $\partial N_b(\bx)/\partial x_{f,b} = \lambda_f/\mu_b'(N_b)$. Therefore, for $x_{f,b}^* >0$ complementary slackness implies $\alpha_{f,b} =0$ and  
    $\frac{\lambda_f}{\mu_b'(N_b^*)}= \lambda_f c_f\,.$
    For  $(f,b) \in \mathcal{E}$ with $x_{f,b}^* = 0$ we have $\alpha_{f,b}\geq 0$ and thus
    $\frac{\lambda_f}{\mu_b'(N_b^*)} \geq \lambda_f c_f.$
    The result follows by dividing both conditions by $\lambda_f > 0$.

Sufficiency follows because the problem is convex since the feasible set is linear and $N_b(\bx)$ is convex. Therefore, the first-order conditions are sufficient for optimality.
\Halmos \end{proof}

\subsection{Proof of Theorem \ref{thm: converge to fluid model}}
\paragraph{First Part}
We have that $Y_b^{(c)}(i) := N_b^{(c)}(i)/c$ satisfies the following stochastic recursive inclusion
\begin{align*}
    &Y_b^{(c)}(i+1)= Y_b^{(c)}(i) + \frac{1}{c} \left( \sum_{f\in \mathcal{F}(b)} A_{f,b}^{(c)}(i) - D_b^{(c)}(i)\right) \\
    & = Y_b^{(c)}(i) + \frac{1}{c} \underbrace{\left( \sum_{f\in \mathcal{F}(b)}\lambda_f x^{(c)}_{f,b}(i) -\mu_b(Y_b^{(c)}(i)) \right)}_{g^{(c)}_b(i)} + \frac{1}{c} \underbrace{\left( \sum_{f\in \mathcal{F}(b)} A_{f,b}^{(c)}(i) - D_b^{(c)}(i) - \left( \sum_{f\in \mathcal{F}(b)}\lambda_f x^{(c)}_{f,b}(i) -\mu_b(Y_b^{(c)}(i))\right)\right)}_{\varepsilon_b^{(c)}(i+1)}.
\end{align*}
We hereafter move the time step index $i$ to the subscript, i.e., 
\begin{align}\label{eq: Yn+1 from Yn}
    \bY^{(c)}_{i+1} = \bY^{(c)}_i + \frac{1}{c}\left(\bg^{(c)}_i + \bve^{(c)}_{i+1}\right),
\end{align}
where 
$\bve^{(c)}_{i+1}$ is the difference between a random variable and its mean, which can be viewed as the noise added to the system. Specifically, $\{\bve^{(c)}_i\}$ is a martingale difference sequence with respect to the past:
$\mathbb{E}[\bve^{(c)}_{i+1}| \sigma(\bY^{(c)}_m, \bve^{(c)}_m, m \leq i)] = 0 \text{ a.s., }i \geq 0.$

To prove the relative compactness of the linear interpolations $\{\bar{\bY}^{(c)}(\cdot)\}_{c\in \mathbb{N}}$, we show an equivalence between it and the following differential sequence, whose relative compactness is straight-forward. 

Let's define a continuous, piecewise constant $\bar{\bg}^{(c)}(t)$, $t \in [0,T]$ where $\bar{\bg}^{(c)}(i/c) = \bg^{(c)}_i$, $i\geq 0$ and on each interval $t\in [i/c, (i+1)/c)$,
$\bar{\bg}^{(c)}(t) = \bg^{(c)}_i.$

Let $\by^{(c)}(\cdot)$ denote the integrals of the piecewise linear $\bar{\bg}^{(c)}(\cdot)$ that start with $\bar{\bY}^{(c)}(0)$ at time 0, i.e.,  
${\by}^{(c)}(t) = \by^{(c)}(0) + \int_0^t \bar{\bg}^{(c)}(u)\,du, \by^{(c)}(0) = \bar{\bY}^{(c)}(0).$

Here process $\by^{(c)}(\cdot)$ is the compensator of the process $\bY^{(c)}(\cdot)$, i.e., the integral of the drift; and we will work with $\by^{(c)}(\cdot)$ as it allows us to work with the expected values. Using that $\sup_{i \leq \lceil T c\rceil} \|\bg^{(c)}_i\| < \infty$ a.s., we have that $\{\by^{(c)}(\cdot)\}_{c\in \mathbb{N}}$ is equicontinuous and pointwise bounded. By the Arzel\`a-Ascoli Theorem, it is relatively compact in $C([0,T]; \mathbb{R}^{|\mathcal{B}|})$. The following lemma allows us to pass the relative compactness to $\bY^{(c)}(\cdot)$.

\begin{lemma}\label{lem: barY to y}
For any $T >0$, let $\|\cdot \|$ denote Euclidean norm, we have 
\[\lim_{c\to\infty} \sup_{t \in [0,T]}\|\bar{\bY}^{(c)}(t) - \by^{(c)}(t)\| = 0, \text{ a.s..}\]
\end{lemma}
\begin{proof}
   Let's fix a sample path and fix $t\in [0,T]$, let $[t]: = \max\{i/c: i/c \leq t, i\in \mathbb{N}\}$, then 
\begin{align*}
    \by^{(c)}(t) & \stackrel{(a)}{=} \by^{(c)}(0) + \int_0^t \bar{\bg}^{(c)}(u)\,du \stackrel{(b)}{=} \bar\bY^{(c)}(0) + \int_0^{[t]} \bar{\bg}^{(c)}(u)\,du + \int_{[t]}^{t} \bar{\bg}^{(c)}(u)\,du \\
    & \stackrel{(c)}{=} \bar\bY^{(c)}(0) + \sum_{k=0}^{[t]-1}\frac{1}{c}{\bg}^{(c)}_k + \int_{[t]}^{t} \bar{\bg}^{(c)}(u)\,du \stackrel{(d)}{=} \bar\bY^{(c)}([t]) - \sum_{k=1}^{[t]}\frac{1}{c}\bve^{(c)}_k + \int_{[t]}^{t} \bar{\bg}^{(c)}(u)\,du \\
    & \stackrel{(e)}{=} \bar{\bY}^{(c)}(t) - \sum_{k=1}^{[t]}\frac{1}{c}\bve^{(c)}_k + \left(\bar{\bY}^{(c)}([t]) + \int_{[t]}^{t} \bar{\bg}^{(c)}(u)\,du  - \bar{\bY}^{(c)}(t) \right),
\end{align*}
where $(a)$ follows as $\by^{(c)}(\cdot)$ is the solution to the differential equation $\dot{\by}^{(c)}(t) = \bar{\bg}^{(c)}(t)$; $(b)$ follows as $\by^{(c)}(0) = \bar{\bY}^{(c)}(0)$ by definition; $(c)$ follows as $\bar{\bg}^{(c)}(t) = \bg^{(c)}_i, t\in [i/c, (i+1)/c)$ are piecewise constant in each interval; $(d)$ follows as $\bar\bY^{(c)}(0) = \bY^{(c)}_0$, $\bar\bY^{(c)}([t]) = \bY^{(c)}_{[t]}$ and the stochastic recursion \eqref{eq: Yn+1 from Yn}; $(e)$ follows by adding and subtracting $\bar{\bY}^{(c)}(t)$.

We proceed to bound the norm of the term in the bracket:
\begin{align*}
    \left\|\bar{\bY}^{(c)}([t]) + \int_{[t]}^{t} \bar{\bg}^{(c)}(u)\,du  - \bar{\bY}^{(c)}(t) \right\| & \stackrel{(a)}{\leq} \|\bar{\bY}^{(c)}([t])  - \bar{\bY}^{(c)}(t)\| + \int_{[t]}^{t} \|\bar{\bg}^{(c)}(u)\|\,du \\
    & \stackrel{(b)}{\leq} \frac{1}{c} \left(\|\bve^{(c)}_{[t]+1}\| + 2 \|{\bg}^{(c)}_{[t]}\|\right),
\end{align*}
where $(a)$ follows the triangle inequality; $(b)$ follows from $\bar{\bY}^{(c)}([t])  - \bar{\bY}^{(c)}(t) \leq ([t]-t)( {\bg}^{(c)}_{[t]}+ \bve^{(c)}_{[t]+1})$ and $[t]-t < 1/c$, and $\int_{[t]}^{t} \|\bar{\bg}^{(c)}(u)\|\,du = (t-[t]) \|{\bg}^{(c)}_{[t]}\| \leq \frac{1}{c}\|{\bg}^{(c)}_{[t]}\|$.

By Lemma \ref{lem: change in N bounded} and $\mu_b$ are bounded by assumption, we have 
$\lim_{c\to \infty} \frac{1}{c} \sup_{1\leq [t] \leq Tc}(\|\bve^{(c)}_{[t]}\| + 2 \|{\bg}^{(c)}_{[t]}\|) = 0$.

Combining with Lemma \ref{lem: SLLN for epsilon}, the result thus follows. \Halmos \end{proof}

Therefore, Lemma \ref{lem: barY to y} implies that $\{\bar{\bY}^{(c)}(\cdot)\}_{c\in \mathbb{N}}$ is also relatively compact, otherwise this contradicts the compactness of $\{\by^{(c)}(\cdot)\}_{c\in \mathbb{N}}$.

Thus, there exists a subsequence\footnote{{Throughout this proof, whenever we pass to a subsequence, we will denote it by the same symbol as the original sequence to simplify notation.}} such that 
$\bar{\bY}^{(c)}(\cdot) \to \by(\cdot) \text{ in } C([0,T]; \mathbb{R}^{|\mathcal{B}|}).$

As $\{\bar{\bg}^{(c)}(\cdot)\}_{c\in \mathbb{N}}$ is a subset of $L_2([0,T], \mathbb{R}^{|\mathcal{B}|})$ with bounded norm, by the Banach-Alaoglu Theorem \citep{rudin1991functional} it is weakly relatively sequentially compact. Thus, the subsequence has a further subsubsequence such that 
$\bar{\bg}^{(c)}(\cdot) \to \bg(\cdot) \text{ weakly in } L^2([0,T]; \mathbb{R}^{|\mathcal{B}|}).$

Recall that
$\by^{(c)}(t) = \by^{(c)}(0) + \int_0^t \bar{\bg}^{(c)}(u)\, du,$
taking $c\to \infty$ we have 
$\by(t) = \by(0) + \int_0^t \bg(u)\, du.$

We conclude by showing that $\bg(u) \in H(\by(u))$ for $u\in [0,t]$. 

By the Banach-Saks Theorem \citep{partington1977banach}, since $\bar{\bg}^{(c)}(\cdot) \to \bg(\cdot)$ weakly in $L^2([0,T]; \mathbb{R}^{|\mathcal{B}|})$, passing to a further subsequence (denoted by the same symbol), the Ces\`{a}ro means converge almost everywhere:
\[\frac{1}{c} \sum_{i=1}^{c} \bar{\bg}^{(i)}(u) \to \bg(u) \quad \text{a.e.\ } u \in [0,T].\]

We now argue that $\bg(u) \in H(\by(u))$ for a.e.\ $u$. Write $[u] := \max\{k/c: k/c \leq u, k\in \mathbb{N}\}$, so $\bar{\bg}^{(c)}(u) = \bg^{(c)}_{[u]} \in H(\bY^{(c)}_{[u]})$ by definition of GMSR. From the linear-interpolation definition of $\bar{\bY}^{(c)}$,
\[\bar{\bY}^{(c)}(u) = \bY^{(c)}_{[u]} + (u-[u])\bigl(\bg^{(c)}_{[u]} + \bve^{(c)}_{[u]+1}\bigr),\]
so $\|\bY^{(c)}_{[u]} - \bar{\bY}^{(c)}(u)\| \leq \frac{1}{c}\bigl(\|\bg^{(c)}_{[u]}\| + \|\bve^{(c)}_{[u]+1}\|\bigr) \to 0$, since $\|\bg^{(c)}_{[u]}\|$ is bounded (service rates are bounded) and $\frac{1}{c}\|\bve^{(c)}_{[u]+1}\| \to 0$ by Lemma~\ref{lem: change in N bounded}. Combined with $\bar{\bY}^{(c)}(u) \to \by(u)$, this gives $\bY^{(c)}_{[u]} \to \by(u)$ as $c\to\infty$. By the upper semi-continuity of $H$,
$\mathrm{dist}\!\left(\bar{\bg}^{(c)}(u),\, H(\by(u))\right) \to 0 \quad \text{as } c\to\infty.$
Since $H(\by(u))$ is convex, the distance function $v \mapsto \mathrm{dist}(v,\, H(\by(u)))$ is convex, so
\[\mathrm{dist}\!\left(\frac{1}{c}\sum_{i=1}^{c} \bar{\bg}^{(i)}(u),\, H(\by(u))\right) \leq \frac{1}{c}\sum_{i=1}^{c} \mathrm{dist}\!\left(\bar{\bg}^{(i)}(u),\, H(\by(u))\right) \to 0,\]
where the last step uses Ces\`{a}ro's theorem (the average of a sequence converging to 0 also converges to 0). Since $\frac{1}{c}\sum_{i=1}^{c}\bar{\bg}^{(i)}(u) \to \bg(u)$ and $H(\by(u))$ is closed, continuity of distance gives $\mathrm{dist}(\bg(u),\, H(\by(u))) = 0$, and hence $\bg(u) \in H(\by(u))$ for a.e.\ $u \in [0,T]$.

\begin{lemma}\label{lem: change in N bounded}
Let $Z_b^{(c)}(i) = \sum_{f\in \mathcal{F}(b)} A_{f,b}^{(c)}(i) - D_b^{(c)}(i)$, for any $T >0$, $\lim_{c\to\infty} \sup_{1\leq i\leq Tc} \frac{1}{c}  |Z_b^{(c)}(i)| = 0 \text{ a.s. }$.
\end{lemma}
\begin{proof}
    For any $\epsilon >0$, let $2< q \leq 4$, 
    \begin{align*}
        \mathbb{P}\left(\sup_{1\leq i\leq Tc} \frac{1}{c}  |Z_b^{(c)}(i)|>\epsilon\right) & \stackrel{(a)}{\leq} \frac{\mathbb{E}\left[\sup_{1\leq i\leq Tc} |Z_b^{(c)}(i)|^q\right]}{c^q \epsilon^q} \stackrel{(b)}{\leq} \frac{\sum_{1\leq i\leq Tc}\mathbb{E}\left[ |Z_b^{(c)}(i)|^q\right]}{c^q \epsilon^q} \stackrel{(c)}{\leq} \frac{Tc K_z}{c^q \epsilon^q},
    \end{align*}
where $(a)$ follows Chebyshev’s inequality; $(b)$ follows as for any sequence $\{a_i\}$, we have $\sum_{i} |a_i| \geq \sup_i |a_i|$; $(c)$ follows by Minkowski’s Inequality and the definition of $Z_b^{(c)}(i)$:
\begin{align}\label{eq: Z momnets bounded}
    \left(\mathbb{E}\left[ |Z_b^{(c)}(i)|^q\right]\right)^{1/q} \leq \sum_{f\in \mathcal{F}(b)}\left(\mathbb{E}\left[ |A_{f,b}^{(c)}(i)|^q\right]\right)^{1/q}  +\left(\mathbb{E}\left[ |D_b^{(c)}(i)|^q\right]\right)^{1/q} \leq K_z^{1/q}, 
\end{align}
    here $K_z^{1/q}$ is a finite constant as $A_{f,b}^{(c)}(i) \leq \sum_{b\in \mathcal{B}(f)}A_{f,b}^{(c)}(i) = W^{(c)}_f(i)$, and both $W^{(c)}_f(i)$ and $D_b^{(c)}(i)$ has finite fourth moment.
Therefore, 
\[\sum_{c=1}^\infty \mathbb{P}\left(\sup_{1\leq i\leq Tc} \frac{1}{c}  |Z_b^{(c)}(i)|>\epsilon\right) \leq \sum_{c=1}^\infty \frac{T K_z}{c^{q-1} \epsilon^q} < \infty.\]
By Borel-Cantelli lemma, as $\epsilon >0$ was arbitrary, the result follows.
\Halmos \end{proof} 
\begin{lemma}\label{lem: SLLN for epsilon}
For any $T >0$, 
\[\lim_{c\to\infty} \sup_{1\leq k \leq Tc} \frac{1}{c}\sum_{i=1}^{k}\bve^{(c)}_i = 0 \text{ a.s..}\]
\end{lemma}
\begin{proof}
As $\{\bve^{(c)}_i\}$ is martingale difference, we have $M_k^{(c)} = \sum_{i=1}^{k}\bve^{(c)}_i$ is a martingale with respect to $\{\sigma(\bve^{(c)}_1, \ldots, \bve^{(c)}_{k-1})\}$. 
We then have 
\begin{align*}
    \mathbb{E}\left[ (M_k^{c})^4 \right] & \stackrel{(a)}{\leq} K_4 \mathbb{E}\left[\left(\sum_{i=1}^k (\bve^{(c)}_i)^2\right)^2\right] = K_4 \mathbb{E}\left[\sum_{1\leq i, j \leq k} (\bve^{(c)}_i)^2  (\bve^{(c)}_j)^2\right] \\
    & \stackrel{(b)}{=} K_4 \sum_{1\leq i, j \leq k}  \mathbb{E}\left[(\bve^{(c)}_i)^2  (\bve^{(c)}_j)^2\right] \stackrel{(c)}{\leq} K_4 \sum_{1\leq i, j \leq k}  \sqrt{\mathbb{E}[(\bve^{(c)}_i)^4]}  \sqrt{\mathbb{E}[(\bve^{(c)}_j)^4]} \stackrel{(d)}{\leq} K_m k^2.
\end{align*}
where $(a)$ follows from Burkholder’s Inequality (\citet{burkholder1966martingale} Theorem 9) with $K_4>0$ being a constant; $(b)$ follows from linearity of expectation; $(c)$ follows from Cauchy–Schwarz  inequality; $(d)$ follows from \eqref{eq: Z momnets bounded} where we show $\mathbb{E}[(\bve^{(c)}_i)^4] < \infty$, $K_m >0$ is a constant.
Therefore, by Doob's Maximal inequality, 
\[\mathbb{P}\left(\sup_{1\leq k \leq Tc} |M_k^{(c)}| \geq c\epsilon\right) \leq \frac{\mathbb{E}[\max\{(M_{\lceil Tc \rceil}^{(c)})^4, 0\}]}{(c\epsilon)^4} = \frac{\mathbb{E}[(M_{\lceil Tc \rceil}^{(c)})^4]}{(c\epsilon)^4} \leq \frac{K_m T^2}{\epsilon^4 c^2},\]
which implies 
$\sum_{c=1}^\infty\mathbb{P}\left(\sup_{1\leq k \leq Tc} |M_k^{(c)}|  \geq c\epsilon \right) < \infty$.
By Borel-Cantelli Lemma, the result follows.
\Halmos \end{proof}

\begin{theorem}[Theorem 9 \citet{burkholder1966martingale}]
    Let $1<p<\infty$. There are positive real numbers $k_p, K_p$ such that if $\{M_k\}$ is a martingale then 
    $k_p \mathbb{E}[S_k^p] \leq \mathbb{E}[|M_k|^p] \leq K_p \mathbb{E}[S_k^p]$,
    where $S_k = \sqrt{\sum_{i=1}^k (M_{i}-M_{i-1})^2}$, $M_0 = 0$.
\end{theorem}

\section{Differential Inclusion}\label{apx: DI}
\begin{definition}[Definition 2.2.1 \citet{kunze2000nonsmooth}]
    Let $I \subseteq \mathbb{R}$ be an interval with $0 \in I$, $\mathcal{N} \subseteq \mathbb{R}^n, \bN_0 \in \mathcal{N}$, and $H: \mathcal{N}\mapsto \mathcal{V}$ a multi-valued mapping. A function $\bN$ such that 
    \begin{itemize}
        \item $\bN: I \mapsto \mathcal{N}$ is absolutely continuous on $I$,
        \item $\bN(0) = \bN_0$, and 
        \item $\dot{\bN}(t) \in H(\bN(t))$ for almost everywhere $t\in I$
    \end{itemize}
    is called a solution of the differential inclusion $\dot{\bN} \in H(\bN)$ a.e., $\bN(0) = \bN_0$.
\end{definition}

\subsection{Proof of Lemma \ref{lem: Equilibrium Points are Fluid Optimal}}
\begin{proof}
At an equilibrium point $\bN^{eq}$, we must have the temporal derivatives equal zero, i.e., $\bm{0} \in H(\bN^{eq})$. Therefore, there exists $\bx_f\in X_f(\bN^{eq})$ for all $f\in \mathcal{F}$ such that
\[\sum_{f\in\mathcal{F}(b)} \lambda_f x_{f,b} = \mu_b(N_b^{eq}), \forall b \in \mathcal{B}.\]
Therefore, $(\bN^{eq}, \bx)$ is a feasible solution to the optimization problem \eqref{eq: FLU}. We claim that $\mu_b'(N_b^{eq}) > 0$ for all backends $b$ with positive flow. Under GMSR, if $\mu_b'(N_b^{eq}) = 0$ for some backend $b$, no frontend routes to $b$ unless every connected backend also has gradient zero (since any connected backend with positive gradient is strictly preferred). Let $P = \{f \in \mathcal{F}: \mu_{b'}'(N_{b'}^{eq}) = 0 \text{ for all } b' \in \mathcal{B}(f)\}$. For each backend $b$ connected to a frontend in $P$, since $\mu_b'(N_b^{eq}) = 0$ and $\mu_b$ is concave and non-decreasing, $\mu_b$ is constant on $[N_b^{eq}, \infty)$, so $\mu_b(N_b^{eq}) = \mu_b(\infty)$. No frontend $f \notin P$ sends flow to such backends (since $f \notin P$ has a connected backend with gradient $> 0$, which is strictly preferred). Flow balance then gives $\sum_{f\in P}\lambda_f = \sum_{b\in \cup_{f\in P}\mathcal{B}(f)} \mu_b(\infty)$, contradicting Assumption~\ref{assmp: system stable} (strict feasibility). Therefore $P = \emptyset$, and $\mu_b'(N_b^{eq}) > 0$ for all $b$ with positive flow.

By definition of $X_f(\bm{N})$, if $x_{f,b}> 0$, $\mu'_b(N_b) = \max_{j \in \mathcal{B}(f)} \mu'_j(N_j)$. Then by Lemma \ref{lem: optimal FLU structure}, $\bN^{eq}$ satisfies the first-order conditions, which we prove is also sufficient, thus $\bN^{eq}$ is optimal for \eqref{eq: FLU}. Since $\bN^*$ is the unique optimum (Lemma~\ref{lem: unique N star}), we conclude $\bN^{eq} = \bN^*$. Conversely, $\bN^*$ is an equilibrium because flow balance holds at optimality. \Halmos \end{proof}

\section{Stability Analyses}\label{sec:appendix-stability}


\subsection{Proof of Lemma \ref{lem: m V positive definite}}
\begin{proof}
If $(\bN, \bx)  = (\bN^*, \bx^*)$, then $V(\bN, \bx)= \sum_{b \in\mathcal{B}} |\sum_{f\in \mathcal{F}(b)}\lambda_f x_{f,b}^* -\mu_b(N_b^*)| = 0$ by the definition of $\bx^*$.

If $V(\bN, \bx)=0$, then we must have $\sum_{f\in \mathcal{F}(b)}\lambda_f x_{f,b}  - \mu_b(N_b) = 0$ for all $b\in\mathcal{B}$. Meanwhile, as $\bx_f \in X_f(\bN)$, this implies $\sum_{b\in \mathcal{B}(f)} x_{f,b} = 1$ and $\bx_f \geq \bm{0}$ for all $f\in \mathcal{F}$. Therefore, $\bN$ is a feasible solution to the optimization problem \eqref{eq: FLU}.

Moreover, for any $f \in \mathcal{F}$, as $\bx_f \in X_f(\bN)$,
$\mu_{b_1}'(N_{b_1}) = \mu_{b_2}'(N_{b_2}) \text{ for } b_1, b_2 \in \{b\in \mathcal{B}(f): x_{f,b} >0\},$
Therefore, by Lemma \ref{lem: optimal FLU structure}, $\bN = \bN^*$, which implies $\bx = \bx^*$.\Halmos \end{proof}

\subsection{TierGraphs}

\paragraph{Proof of Corollary \ref{cor:gradient decreases along tiers}}
We first prove a useful lemma.
\begin{lemma}\label{lem:low tier low gradient}
    For two tiers $(F_i, B_i), (F_j, B_j)$ with $F_i, F_j \subseteq \mathcal{F}, B_i, B_j \subseteq \mathcal{B}$ at $\bN$, let vertex $v_i$ and vertex $v_j$ denote their corresponding vertices in the TierGraph $\mathcal{T}(\bN)$. If $v_i\to v_j$, then $\mu_{B_i}'(\bN) > \mu_{B_j}'(\bN)$.
\end{lemma}
\begin{proof}
    By Definition \ref{def:tiergraph}, $v_i\to v_j$ implies there exists $f \in F_i, b \in B_j$ such that $(f, b) \in \mathcal{E}$. Therefore, by Definition \ref{def:tier}, $b \notin S_{f}(\bN)$, i.e., $\mu_{B_i}'(\bN) > \mu_{B_j}'(\bN)$.\Halmos \end{proof}

We next state the proof of Corollary \ref{cor:gradient decreases along tiers}.
\begin{proof}
    By Definition \ref{def:tiergraph}, $(F_i, B_i) \prec (F_j, B_j)$ implies that in the TierGraph $\mathcal{T}(\bN)$, vertex $v_i$ can reach vertex $v_j$, i.e., there exists a sequence of vertices $v_{a_1}, v_{a_2},\ldots, v_{a_k}$ such that
    $v_i\to v_{a_1} \to v_{a_2} \to \cdots \to v_{a_k} \to v_j$
    in the TierGraph $\mathcal{T}(\bN)$. By Lemma \ref{lem:low tier low gradient}, we then have
    \[\mu_{B_i}'(\bN) > \mu_{B_{a_1}}'(\bN) > \mu_{B_{a_2}}'(\bN) > \cdots > \mu_{B_{a_k}}'(\bN) > \mu_{B_j}'(\bN).\Halmos \]
\end{proof}

\paragraph{Proof of Corollary \ref{cor: TierGraph DAG}}
\begin{proof}
Suppose not, i.e., there exists a sequence of vertices such that $v_1 \to v_2 \to \cdots \to v_m \to v_1$ in the TierGraph $\mathcal{T}(\bN)$. Then by Corollary \ref{cor:gradient decreases along tiers}, we have
$\mu_{B_{1}}'(\bN) > \mu_{B_{2}}'(\bN) > \cdots > \mu_{B_{m}}'(\bN) > \mu_{B_{1}}'(\bN),$
which is a contradiction.\Halmos \end{proof}

\subsection{Proof of Lemma \ref{lem: tier N same direction}}
\begin{proof}
   As $(F,B)$ is a tier at $\bN(t)$ for $t\in [t_1, t_2]$, for any $b_1, b_2 \in B$, we have
$\mu_{b_1}'(N_{b_1}(t)) = \mu_{b_2}'(N_{b_2}(t)), \forall t \in [t_1, t_2].$
If $\frac{d}{dt}N_{b_1}(t) = 0$ or $\frac{d}{dt}N_{b_2}(t) = 0$, the product $\frac{d}{dt}N_{b_1}(t) \cdot \frac{d}{dt}N_{b_2}(t) \geq 0$ holds trivially. Otherwise, suppose for contradiction that $\frac{d}{dt}N_{b_1}(t_0) > 0$ and $\frac{d}{dt}N_{b_2}(t_0) < 0$ for some $t_0 \in (t_1, t_2)$. Since the trajectories are continuous, for sufficiently small $\epsilon > 0$ we have $N_{b_1}(t_0 + \epsilon) > N_{b_1}(t_0)$ and $N_{b_2}(t_0 + \epsilon) < N_{b_2}(t_0)$. Since $\mu_b'$ is strictly decreasing (by strict concavity), this gives $\mu_{b_1}'(N_{b_1}(t_0+\epsilon)) < \mu_{b_1}'(N_{b_1}(t_0))$ and $\mu_{b_2}'(N_{b_2}(t_0+\epsilon)) > \mu_{b_2}'(N_{b_2}(t_0))$, contradicting the equal-gradient condition at $t_0 + \epsilon$.\Halmos \end{proof}

\subsection{Proof of Lemma \ref{lem: K+ lambda big}}
\begin{proof}
We first show $N_b(t)$ is increasing
for all $b \in B, t\in [t_1, t_2]$.
\begin{align}
    \sum_{b \in B} \dot{N}_b(t)& \stackrel{(a)}{=} \sum_{b \in B} \left(\sum_{f\in F\cap \mathcal{F}(b)} \lambda_f x_{f,b}(t) - \mu_b(N_b(t))\right) = \sum_{f\in F} \lambda_f \sum_{b\in B\cap\mathcal{B}(f)} x_{f,b} - \sum_{b \in B} \mu_b(N_b(t)) \notag \\
    & \stackrel{(b)}{=}  \sum_{f\in F}\lambda_f - \sum_{b \in B} \mu_b(N_b(t))  \stackrel{(c)}{\geq} 0, \label{eq: sum of dotN in B}
\end{align}
where $(a)$ follows from the definition of the differential inclusion \eqref{eq: differential inclusion}; $(b)$ follows from $\sum_{b\in B\cap\mathcal{B}(f)} x_{f,b}(t)= 1$ for all $f\in F$; $(c)$ follows from the assumption.

Combined with Lemma \ref{lem: tier N same direction}, i.e., the change in workloads at the backends in the same tier must have the same direction, we have $\dot{N}_b(t) \geq 0$ for all $b \in B$.
Therefore, this means for $t\in [t_1, t_2]$, there must exist $\{x_{f,b}(t)\}_{f\in F\cap \mathcal{F}(b)}$ such that
\begin{align}
    & \frac{d}{dt}N_b(t) = \sum_{f\in F\cap \mathcal{F}(b)} \lambda_f x_{f,b}(t) -\mu_b(N_b(t)) \geq 0\label{eq: single ellb small} \\
    & 0\leq x_{f,b} \leq 1, \forall f \in F\cap \mathcal{F}(b), \quad \sum_{b\in B \cap \mathcal{B}(f)} x_{f,b} = 1, \forall f\in F, \notag
\end{align}
where the last two constraints follow from the definition of the differential inclusion \eqref{eq: differential inclusion}.

Therefore, for any subset of backends $Q\subseteq B$,
\[\sum_{b\in Q} \mu_b(N_b(t)) \stackrel{(d)}{\leq} \sum_{b\in Q} \sum_{f\in F\cap \mathcal{F}(b)} \lambda_f x_{f,b}(t) = \sum_{f\in  F\cap (\cup_{b\in Q}\mathcal{F}(b))}\lambda_f x_{f,b}(t) \stackrel{(e)}{\leq}  \sum_{f\in F\cap (\cup_{b\in Q}\mathcal{F}(b))}\lambda_f.\]
where $(d)$ follows inequality \eqref{eq: single ellb small}; $(e)$ holds as $x_{f,b}(t) \leq 1$.\Halmos \end{proof}

\subsection{Proof of Lemma \ref{lem: K- lambda small}}
\begin{proof}
    Similar to the proof for Lemma \ref{lem: K+ lambda big}, we can show that $N_b(\tau_1) \geq N_b(\tau_2)$ for any $\tau_1 \leq \tau_2 \in [t_1, t_2]$, for all $b\in B$. Therefore, by definition of differential inclusion, there must exist $\{x_{f,b}(t)\}_{f\in F\cap \mF(b)}$ such that
\begin{align}
    & \frac{d}{dt}N_b(t) = \sum_{f\in F \cap \mathcal{F}(b)} \lambda_f x_{f,b}(t) -\mu_b(N_b(t)) \leq 0, \label{eq: single ellb big}\\
    & 0\leq x_{f,b} \leq 1, \forall f_1\in F\cap \mF(b), \quad \sum_{b\in B \cap \mathcal{B}(f)} x_{f,b} = 1, \forall f\in F. \label{eq: feasible x sum to 1}
\end{align}

Therefore, for any subset of frontends $P\subseteq F$,
\begin{align*}
    \sum_{f\in P} \lambda_f &\stackrel{(a)}{=} \sum_{f\in P} \lambda_f \sum_{b\in B \cap \mathcal{B}(f)} x_{f,b}(t) = \sum_{b\in B\cap(\cup_{f\in P}\mathcal{B}(f)) } \sum_{f\in P\cap\mathcal{F}(b) } \lambda_f x_{f,b}(t) \\
    & \stackrel{(b)}{\leq}  \sum_{b\in B\cap(\cup_{f\in P}\mathcal{B}(f))} \sum_{f\in F \cap \mathcal{F}(b)} \lambda_f x_{f,b}(t) \stackrel{(c)}{\leq}  \sum_{b\in B\cap(\cup_{f\in P}\mathcal{B}(f)) } \mu_b(N_b(t)),
\end{align*}
where $(a)$ follows equality \eqref{eq: feasible x sum to 1}; $(b)$ holds as $P\subseteq F$; $(c)$ follows inequality \eqref{eq: single ellb big}.\Halmos \end{proof}

\subsection{Proof of Corollary \ref{cor: slide flow imbalance no change}}
When the tier structure remains unchanged, the total flow imbalance never flips sign from positive to negative or vice versa. {This is because flow imbalance drives workloads towards equilibrium, and when the flow is balanced, i.e., jobs arrive at the same rate as jobs depart, the workloads stabilize, and the total flow imbalance remains at zero. }This observation is formalized in the following corollary.

\begin{corollary}\label{cor: slide flow imbalance no change}
If there exists $t_1 < t_2$ and $F \subseteq \mathcal{F}, B \subseteq \mathcal{B}$ such that for $t\in [t_1, t_2]$, $(F,B)$ is a tier at $\bN(t)$, then for $\tau_1, \tau_2 \in [t_1, t_2]$,
    \[\left(\sum_{f\in F} \lambda_f - \sum_{b\in B}\mu_b(N_b(\tau_1))\right)\cdot \left(\sum_{f\in F} \lambda_f - \sum_{b\in B}\mu_b(N_b(\tau_2))\right) \geq 0.\]
\end{corollary}
\begin{proof}
Let $I(t) = \sum_{f\in F} \lambda_f - \sum_{b\in B}\mu_b(N_b(t))$ denote the flow imbalance. By continuity of each $\mu_b(N_b(\cdot))$, $I(\cdot)$ is continuous. We show $I$ cannot change sign on $[t_1, t_2]$.
Suppose for contradiction that $I(\tau_1) > 0$ and $I(\tau_2) < 0$ for some $\tau_1, \tau_2 \in [t_1, t_2]$. By the intermediate value theorem, there exists $\tau$ between $\tau_1$ and $\tau_2$ with $I(\tau) = 0$. We claim $I(t) = 0$ for all $t \in [\tau, t_2]$: at any time $t$ where $I(t) = 0$, the tier has zero flow imbalance, so both Lemma~\ref{lem: K+ lambda big} (with $\geq 0$) and Lemma~\ref{lem: K- lambda small} (with $\leq 0$) apply, giving $\frac{d}{dt} N_b(t) = 0$ for all $b\in B$. Therefore $I(t) = 0$ for all $t\in [\tau, t_2]$, contradicting $I(\tau_2) < 0$. The case $I(\tau_1) < 0, I(\tau_2) > 0$ is symmetric. \Halmos \end{proof}

\subsection{Proof of Lemma \ref{lem: slide}}\label{apx: proof of slide}

\begin{proof}
The total drift for the backends in $B$ is
\begin{align*}
V_B(\bN(t), \bx(t)) & = \sum_{b\in B} \left|\sum_{f\in \mathcal{F}(b)} \lambda_f x_{f,b}(t) - \mu_b(N_b(t))\right| \stackrel{(a)}{=}  \sum_{b\in B} \left|\sum_{f\in F \cap \mathcal{F}(b)} \lambda_f x_{f,b}(t) - \mu_b(N_b(t))\right| \\
& \stackrel{(b)}{=} \left|\sum_{b\in B} \sum_{f\in F \cap \mathcal{F}(b)} \lambda_f x_{f,b}(t) - \sum_{b\in B}\mu_b(N_b(t))\right|  = \left|\sum_{f\in F} \lambda_f \sum_{b\in B \cap \mathcal{B}(f)}  x_{f,b}(t) - \sum_{b\in B}\mu_b(N_b(t))\right| \\
    & \stackrel{(c)}{=}  \left| \sum_{f \in F} \lambda_f - \sum_{b\in B}\mu_b(N_b(t)) \right|,
\end{align*}
where $(a)$ follows as $(F,B)$ is a tier, thus $x_{f,b}(t)=0$ for $f \in \mF(b)\setminus F$; $(b)$ follows from Lemma \ref{lem: tier N same direction} because flow balance for backends in a tier have the same sign; $(c)$ follows as $\sum_{b\in B \cap \mathcal{B}(f)} x_{f,b}(t) = 1$ by flow balance at the frontends.

We conclude the proof by showing that $V_B(\bN(t), \bx(t))$ is non-increasing in $[t_1, t_2]$.
If $\sum_{f\in F}\lambda_f \geq \sum_{b\in B} \mu_b(N_b(t_1))$. Then by Corollary \ref{cor: slide flow imbalance no change}, we have $\sum_{f\in F}\lambda_f \geq \sum_{b\in B} \mu_b(N_b(t))$ for $t\in [t_1, t_2]$. Therefore, for $t\in [t_1, t_2]$,
\begin{align}
    \frac{d}{dt} V_B(\bN(t), \bx(t)) & = \frac{d}{dt} \left(\sum_{f\in F}\lambda_f - \sum_{b\in B} \mu_b(N_b(t))\right) = -\sum_{b\in B} \mu'_b(N_b(t)) \frac{d}{dt} N_b(t) \notag\\
    & \stackrel{(a)}{=} - \mu_B'(\bN(t)) \sum_{b\in B} \frac{d}{dt} N_b(t) \stackrel{(b)}{=} - \mu_B'(\bN(t))\left(\sum_{f\in F}\lambda_f - \sum_{b\in B} \mu_b(N_b(t))\right) \stackrel{(c)}{\leq} 0. \label{eq: derivative of VB +}
\end{align}
where $(a)$ follows as all backends in a tier share the same gradient value, denoted by $\mu_B'(\bN(t))$; $(b)$ follows \eqref{eq: sum of dotN in B}; $(c)$ follows as $\mu_b'(N_b) >0$ by assumption and $\sum_{f\in F}\lambda_f \geq \sum_{b\in B} \mu_b(N_b(t))$ for $t\in [t_1, t_2]$.

Similarly, if the tier has negative flow imbalance,  for $t\in [t_1, t_2]$,
\begin{align}
    \frac{d}{dt} V_B(\bN(t), \bx(t)) & = \frac{d}{dt} \left( \sum_{b\in B} \mu_b(N_b(t)) - \sum_{f\in F}\lambda_f \right) = \sum_{b\in B} \mu'_b(N_b(t)) \frac{d}{dt} N_b(t) \notag\\
    & \stackrel{(a)}{=} \mu_B'(\bN(t)) \sum_{b\in B} \frac{d}{dt} N_b(t) \stackrel{(b)}{=} \mu_B'(\bN(t))\left( \sum_{f\in F}\lambda_f  - \sum_{b\in B} \mu_b(N_b(t))\right) \notag\\
    & = - \mu_B'(\bN(t))\left(\sum_{b\in B} \mu_b(N_b(t)) - \sum_{f\in F}\lambda_f \right) \stackrel{(c)}{\leq} 0. \label{eq: derivative of VB -}
\end{align}
where $(a)$ follows as all backends in a tier share the same gradient value; $(b)$ follows \eqref{eq: sum of dotN in B}; $(c)$ follows as $\mu_b'(N_b) >0$ by assumption and $\sum_{f\in F}\lambda_f \leq \sum_{b\in B} \mu_b(N_b(t))$ for $t\in [t_1, t_2]$.
\Halmos \end{proof}

\subsection{Proof of Lemma \ref{lem:split}}\label{apx: proof of split}

\begin{proof}
First by Lemma \ref{lem: slide}, $V_B(\bN(t), \bx(t)) = \left|\sum_{f\in F} \lambda_f - \sum_{b\in B}\mu_b(N_b(t))\right|$ for $t\in [t_1, \tau)$. We proceed to show that this holds for $t\in (\tau, t_2]$.
    \begin{itemize}
        \item Suppose $\sum_{f\in F} \lambda_f - \sum_{b\in B}\mu_b(N_b(t)) \geq 0$ for $t\in [t_1, \tau)$, which we refer to as the tier having positive flow imbalance. We proceed to show that all tiers $(\hat{F}_1, \hat{B}_1), \ldots, (\hat{F}_{\hat{k}}, \hat{B}_{\hat{k}})$ also have positive flow imbalance for $t\in (\tau, t_2]$.

        Without loss of generality, let $\mathcal{T}(\bN(t))$ denote the TierGraph that is restricted to tiers $(\hat{F}_1, \hat{B}_1), \ldots, (\hat{F}_{\hat{k}}, \hat{B}_{\hat{k}})$ for $t \in (\tau, t_2]$.
        \begin{itemize}
            \item If $(\hat F_i, \hat B_i)$ corresponds to a source vertex in the TierGraph $\mathcal{T}(\bN(t))$\footnote{Note that it's possible for the TierGraph to have multiple source vertices. This argument holds for all such source vertices and thus all corresponding tiers.}, then
            \begin{align}\label{eq: first tier positive}
            \sum_{b\in \hat B_i} \mu_b(N_b(\tau)) \stackrel{(a)}{=} \sum_{b\in \hat B_i} \mu_b(N_b(\tau^-)) \stackrel{(b)}{\leq} \sum_{f\in F\cap (\cup_{b\in \hat B_i}\mathcal{F}(b))} \lambda_f \stackrel{(c)}{=} \sum_{f\in \hat F_i} \lambda_f,
        \end{align}
         where $(a)$ follows from the continuity of $\mu_b(N_b(t))$; $(b)$ follows from Lemma \ref{lem: K+ lambda big}; $(c)$ follows as $(\hat F_i, \hat B_i)$ corresponds to a source vertex, i.e., $F\cap (\cup_{b\in \hat B_i}\mathcal{F}(b))= \hat F_i$.

         Therefore, by Corollary \ref{cor: slide flow imbalance no change}, $\sum_{b\in \hat B_i} \mu_b(N_b(t)) \leq \sum_{f\in \hat F_i} \lambda_f$ for $t\in [\tau, t_2]$.

         \item If $(\hat F_i, \hat B_i)$ does not correspond to a source vertex in the TierGraph $\mathcal{T}(\bN(t))$, i.e., there exists a tier $(\hat F_j, \hat B_j)$ that corresponds to a source vertex such that $ (\hat F_j,\hat B_j) \prec (\hat F_i,\hat B_i)$.

         We prove by contradiction, i.e., suppose there exists $\hat{\tau} \in (\tau, t_2]$ such that $\sum_{b\in \hat B_i} \mu_b(N_b(\hat{\tau})) > \sum_{f\in \hat F_i} \lambda_f$. Then by Corollary \ref{cor: slide flow imbalance no change}, as $(\hat F_i, \hat B_i)$ is a tier at $\bN(t)$ for $t\in (\tau, t_2]$, we have that tier $(\hat F_i, \hat B_i)$ has negative flow imbalance in the whole interval $(\tau, t_2]$, i.e., $\sum_{b\in \hat B_i} \mu_b(N_b(t)) \geq \sum_{f\in \hat F_i} \lambda_f$ for $t\in (\tau, t_2]$. Therefore, for $b_i \in \hat B_i, b_j \in \hat B_j$,
            \[\mu'_{b_i}(N_{b_i}(t)) \stackrel{(d)}{>} \mu'_{b_i}(N_{b_i}(\tau)) \stackrel{(e)}{=} \mu'_{b_i}(N_{b_i}(\tau^-)) \stackrel{(f)}{=} \mu'_{b_j}(N_{b_j}(\tau^-)) \stackrel{(e)}{=} \mu'_{b_j}(N_{b_j}(\tau)) \stackrel{(g)}{\geq} \mu'_{b_j}(N_{b_j}(t)),\]
            where $(d)$ follows because by Lemma \ref{lem: K- lambda small} we have $N_{b_i}(t) < N_{b_i}(\tau)$ and $\mu_b'$ is strictly decreasing (by strict concavity); $(e)$ follows from the continuity of $\mu_b'(N_b(t))$; $(f)$ follows because $b_i, b_j$ are in the same tier for $t\in [t_1,\tau)$; $(g)$ follows from \eqref{eq: first tier positive}, i.e., tier $(\hat F_j, \hat B_j)$ corresponds to a source vertex and thus has positive flow imbalance, which implies $N_{b_j}(t) \geq N_{b_j}(\tau)$ by Lemma \ref{lem: K+ lambda big}.

            This contradicts $(\hat F_j,\hat  B_j) \prec (\hat F_i, \hat B_i)$ by Corollary \ref{cor:gradient decreases along tiers}.
        \end{itemize}

        Therefore, we have shown that all tiers $(\hat F_1, \hat B_1), \ldots, (\hat F_k, \hat B_k)$ have positive flow imbalance. Then for $t \in (\tau, t_2]$,
        \begin{align*}
            V_B(\bN(t), \bx(t)) & = \sum_{b\in B} \left|\sum_{f\in \mathcal{F}(b)} \lambda _f x_{f,b}(t) -\mu_b(N_b(t))\right| \stackrel{(a)}{=} \sum_{i\in [\hat{k}]} \sum_{b\in \hat B_i} \left|\sum_{f\in \mathcal{F}(b)} \lambda _f x_{f,b}(t) -\mu_b(N_b(t))\right| \\
            & \stackrel{(b)}{=} \sum_{i\in [\hat{k}]} \left|\sum_{f\in \hat{F}_i} \lambda _f -\sum_{b\in \hat B_i}\mu_b(N_b(t))\right| \stackrel{(c)}{=} \sum_{f\in F} \lambda_f - \sum_{b\in B}\mu_b(N_b(t)) = \left|\sum_{f\in F} \lambda_f - \sum_{b\in B}\mu_b(N_b(t)) \right|,
        \end{align*}
        where $(a)$ follows as $(\hat{F}_1, \hat{B}_1), \ldots, (\hat{F}_{\hat{k}}, \hat{B}_{\hat{k}})$ are tiers at $\bN(t), t\in (\tau, t_2]$ and $B = \hat{B}_1 \cup \cdots \cup \hat{B}_{\hat{k}}$; $(b)$ follows from Lemma \ref{lem: slide}; $(c)$ follows as all tiers have positive flow imbalance.

        \item Suppose $\sum_{f\in F} \lambda_f - \sum_{b\in B}\mu_b(N_b(t)) \leq 0$ for $t\in [t_1, \tau)$, which we refer to as the tier having negative flow imbalance. We proceed to show that all tiers also have negative flow imbalance for $t\in (\tau, t_2]$.

        Without loss of generality, let $\mathcal{T}(\bN(t))$ denote the TierGraph that is restricted to tiers $(\hat F_1, \hat B_1), \ldots, (\hat F_k, \hat B_k)$ for $t \in (\tau, t_2]$.
        \begin{itemize}
            \item If $(\hat F_i, \hat B_i)$ corresponds to a sink vertex in the TierGraph $\mathcal{T}(\bN(t))$\footnote{Note that it's possible for the TierGraph to have multiple sink vertices. This argument holds for all such sink vertices and thus all corresponding tiers.}, then
            \begin{align}\label{eq: last tier negative}
            \sum_{f \in \hat F_i} \lambda_f \stackrel{(a)}{\leq} \sum_{b\in B \cap (\cup_{f\in \hat F_i}\mathcal{B}(f))} \mu_b(N_b(\tau^-)) \stackrel{(b)}{=} \sum_{b\in \hat B_i} \mu_b(N_b(\tau^-))\stackrel{(c)}{=} \sum_{b\in \hat B_i} \mu_b(N_b(\tau)),
        \end{align}
         where $(a)$ follows from Lemma \ref{lem: K- lambda small}; $(b)$ follows as $(\hat F_i, \hat B_i)$ corresponds to a sink vertex, i.e., $B \cap (\cup_{f\in \hat F_i}\mathcal{B}(f)) = \hat B_i$; $(c)$ follows from the continuity of $\mu_b(N_b(t))$.

         Therefore, by Corollary \ref{cor: slide flow imbalance no change}, $\sum_{f\in \hat F_i} \lambda_f \leq \sum_{b\in \hat B_i} \mu_b(N_b(t)) $ for $t\in [\tau, t_2]$.
            \item If $(\hat F_i, \hat B_i)$ does not correspond to a sink vertex in the TierGraph $\mathcal{T}(\bN(t))$, i.e., there exists a tier $(\hat F_j, \hat B_j)$ that corresponds to a sink vertex such that $ (\hat F_i, \hat B_i) \prec (\hat F_j, \hat B_j)$.

         We prove by contradiction, i.e., suppose there exists $\hat{\tau} \in (\tau, t_2]$ such that $\sum_{b\in \hat B_i} \mu_b(N_b(\hat{\tau})) < \sum_{f\in F_i} \lambda_f$. Then by Corollary \ref{cor: slide flow imbalance no change}, as $(\hat F_i, \hat B_i)$ is a tier at $\bN(t)$ for $t\in (\tau, t_2]$, we have that $\sum_{b\in \hat B_i} \mu_b(N_b(t))\leq \sum_{f\in \hat F_i} \lambda_f$ for $t\in (\tau, t_2]$. Therefore, for $b_i \in \hat B_i, b_j \in \hat B_j$,
            \[\mu'_{b_i}(N_{b_i}(t)) \stackrel{(d)}{<} \mu'_{b_i}(N_{b_i}(\tau)) \stackrel{(e)}{=} \mu'_{b_i}(N_{b_i}(\tau^-)) \stackrel{(f)}{=} \mu'_{b_j}(N_{b_j}(\tau^-)) \stackrel{(e)}{=} \mu'_{b_j}(N_{b_j}(\tau)) \stackrel{(g)}{\leq} \mu'_{b_j}(N_{b_j}(t)),\]
            where $(d)$ follows as by Lemma \ref{lem: K+ lambda big}, we have $N_{b_i}(t) > N_{b_i}(\tau)$ and $\mu_b'$ is strictly decreasing (by strict concavity); $(e)$ follows from the continuity of $\mu_b'(N_b(t))$; $(f)$ follows as $b_i, b_j$ are in the same tier for $t\in [t_1,\tau)$; $(g)$ follows from \eqref{eq: last tier negative}, i.e., tier $(\hat F_j, \hat B_j)$ corresponds to a sink vertex and thus has negative flow imbalance, which implies $N_{b_j}(t) \leq N_{b_j}(\tau)$ by Lemma \ref{lem: K- lambda small}.

            This contradicts $(F_i, B_i) \prec (F_j, B_j)$ by Corollary \ref{cor:gradient decreases along tiers}.

        \end{itemize}
         Therefore, similarly to the first case, for $t\in (\tau, t_2]$,
         \[V_B(\bN(t), \bx(t))  = \sum_{b\in B}\mu_b(N_b(t))- \sum_{f\in F} \lambda_f  = \left|\sum_{f\in F} \lambda_f - \sum_{b\in B}\mu_b(N_b(t)) \right|.\]

    \end{itemize}

Therefore, we have shown that $V_B(\bN(t), \bx(t)) = \left|\sum_{f\in F} \lambda_f - \sum_{b\in B}\mu_b(N_b(t))\right|$. By \eqref{eq: derivative of VB +} and \eqref{eq: derivative of VB -}, we can show the total absolute drifts are decreasing in $[t_1, t_2]$. \Halmos \end{proof}

\subsection{Proof of Lemma \ref{lem:sync}}\label{apx: proof of sync lemma}


\begin{proof}
Let's define the following sets that depend on the flow imbalance of each frontend and backend experience at state $\bN(t_1)$.

\begin{align}
        & F^+ = \bigcup_{i \in [k]: \sum_{f\in F_i} \lambda_f > \sum_{b\in B_i} \mu_b(N_b(t_1))} F_i, \quad B^+ =  \bigcup_{i \in [k]: \sum_{f\in F_i} \lambda_f > \sum_{b\in B_i} \mu_b(N_b(t_1))} B_i, \label{eq: def of 8 subsets}\\
        & F^- =  \bigcup_{i \in [k]: \sum_{f\in F_i} \lambda_f \leq \sum_{b\in B_i} \mu_b(N_b(t_1))}F_i, \quad B^- = \bigcup_{i \in [k]: \sum_{f\in F_i} \lambda_f \leq \sum_{b\in B_i} \mu_b(N_b(t_1))} B_i.\notag
\end{align}
Similarly we can define $\hat{F}^+, \hat{F}^-, \hat{B}^+$ and $\hat{B}^-$ based on $(\hat{F}_1,\hat{B}_1), \ldots, (\hat{F}_{\hat{k}}, \hat{B}_{\hat{k}})$ based on $\bN(t_2)$.

By Corollary \ref{cor: slide flow imbalance no change}, the flow imbalance for each tier will never flip signs for the time interval $[t_1, \tau)$ and $(\tau, t_2]$, thus by Lemma \ref{lem: slide}, we have
        \begin{align*}
            V_B(\bN(t),\bx(t)) & = \sum_{f\in F^+} \lambda_f - \sum_{b\in B^+} \mu_b(N_b(t)) + \sum_{b\in B^-} \mu_b(N_b(t))-\sum_{f\in F^-} \lambda_f, \text{ for } t\in [t_1, \tau),\\
            V_B(\bN(t),\bx(t)) & = \sum_{f\in \hat{F}^+} \lambda_f - \sum_{b\in \hat{B}^+} \mu_b(N_b(t)) + \sum_{b\in \hat{B}^-} \mu_b(N_b(t))-\sum_{f\in \hat{F}^-} \lambda_f, \text{ for } t\in (\tau, t_2].
        \end{align*}
and $V_B(\bN(t),\bx(t))$ is decreasing within $[t_1, \tau)$ and $(\tau, t_2]$. Therefore, we proceed to prove that
\[V_B(\bN(\tau^-), \bx(\tau^-)) \geq V_B(\bN(\tau^+), \bx(\tau^+)).\]

We first present a useful result that states the difference between $V_B$ at $\tau^-$ and $\tau^+$ (see proof in Appendix \ref{apx: proof of difference in V_B}).

\begin{lemma}\label{lem: difference in V_B change}
The difference in $V_B$ is
\begin{align}\label{eq: difference in V_B change}
    &V_B(\bN(\tau^-), \bx(\tau^-))  -V_B(\bN(\tau^+),\bx(\tau^+)) \notag\\
    & \quad\quad=  2 \left(\sum_{f\in F^+ \cap \hat{F}^-}\lambda_f - \sum_{b\in B^+\cap \hat{B}^-}\mu_b(N_b(\tau))\right) + 2 \left(\sum_{b\in B^-\cap \hat{B}^+} \mu_b(N_b(\tau)) - \sum_{f\in F^-\cap\hat{F}^+} \lambda_f \right).
\end{align}
\end{lemma}

\noindent We proceed to discuss the following three cases:
\begin{itemize}
    \item if $F^- = \emptyset, B^- = \emptyset$, we can show $\hat{F}^- = \emptyset, \hat{B}^- = \emptyset$, thus $V_B(\bN(\tau^-), \bx(\tau^-))  -V_B(\bN(\tau^+),\bx(\tau^+)) = 0$.
    \item if $F^+ = \emptyset, B^+ = \emptyset$, we can show $\hat{F}^+ = \emptyset, \hat{B}^+ = \emptyset$, thus $V_B(\bN(\tau^-), \bx(\tau^-))  -V_B(\bN(\tau^+),\bx(\tau^+)) = 0$.
    \item otherwise, $V_B(\bN(\tau^-), \bx(\tau^-))  -V_B(\bN(\tau^+),\bx(\tau^+)) \geq 0$.
\end{itemize}

\paragraph{First case.}
If for $t\in [t_1,\tau)$, $F^- = \emptyset, B^- = \emptyset$, i.e., $\sum_{f\in F_i} \lambda_f \geq \sum_{b\in B_i} \mu_b(N_b(t))$ for all $i\in [k]$, we first show a generalization of Lemma \ref{lem: K+ lambda big} for multiple tiers.

For any subset of backends $Q\subseteq B$ and for $t\in [t_1,\tau)$,
\begin{align}\label{eq: generalization flow imbalance +}
    \sum_{b\in Q} \mu_b(N_b(t)) = \sum_{i \in [k]} \sum_{b\in Q\cap B_i} \mu_b(N_b(t)) \stackrel{(a)}{\leq} \sum_{i\in [k]}\sum_{f\in F_i\cap (\cup_{b\in Q\cap B_i} \mathcal{F}(b))} \lambda_f  \stackrel{(b)}{\leq} \sum_{f\in F\cap (\cup_{b\in Q} \mathcal{F}(b))} \lambda_f,
\end{align}
    where $(a)$ follows from Lemma \ref{lem: K+ lambda big}; $(b)$ follows as for sets $\{A_i\}_{i\in [n]}, \{B_i\}_{i\in [n]}$,
    \begin{align}\label{eq: union intersection ineq}
        \bigcup_{i\in [n]} (A_i \cap B_i) \subseteq \left(\bigcup_{i\in [n]} A_i\right) \cap \left(\bigcup_{i\in [n]} B_i\right).
    \end{align}

    Therefore similar to the proof for Lemma \ref{lem:split}, which is a special case with $K=1$, we can show that all tiers that correspond to a source vertex in the TierGraph restricted to tiers $(\hat F_1,\hat B_1), \ldots, (\hat F_{\hat{k}}, \hat B_{\hat{k}})$ have positive flow imbalance; and by the same proof by contradiction argument, all tiers at $t\in (\tau, t_2]$ have positive flow imbalance.

    Therefore, we have $F^- = \emptyset, B^- = \emptyset, \hat F^-= \emptyset, \hat B^- = \emptyset$. This means $F^+\cap \hat{F}^- = \emptyset, B^+\cap \hat{B}^- = \emptyset, B^-\cap \hat{B}^+ = \emptyset$ and $F^- \cap \hat{F}^+ = \emptyset$, so $V_B(\bN(\tau^-), \bx(\tau^-))  -V_B(\bN(\tau^+),\bx(\tau^+)) = 0$.

\paragraph{Second Case.}
If for $t\in [t_1,\tau)$, $F^+ = \emptyset, B^+ = \emptyset$, i.e., $\sum_{f\in F_i} \lambda_f \leq \sum_{b\in B_i} \mu_b(N_b(t))$ for all $i\in [k]$, we first show a generalization of Lemma \ref{lem: K- lambda small} for multiple tiers.

For any subset of frontends $P\subseteq F$ and for $t\in [t_1,t_2]$,
\begin{align}\label{eq: generalization flow imbalance -}
    \sum_{f\in P} \lambda_f = \sum_{i\in [k]} \sum_{f\in P\cap F_i}\lambda_f \stackrel{(a)}{\leq} \sum_{i\in [k]}\sum_{b\in B_i\cap(\cup_{f\in P\cap F_i}\mathcal{B}(f))} \mu_b(N_b(t)) \stackrel{(b)}{\leq}\sum_{b\in B\cap(\cup_{f\in P}\mathcal{B}(f))} \mu_b(N_b(t)),
\end{align}
    where $(a)$ follows from Lemma \ref{lem: K- lambda small}; $(b)$ follows as \eqref{eq: union intersection ineq}. Therefore similar to the proof for Lemma \ref{lem:split}, which is a special case with $K=1$, we can show that the tiers at $t\in (\tau, t_2]$ also have negative flow imbalance. Therefore, $V_B(\bN(\tau^-), \bx(\tau^-))  -V_B(\bN(\tau^+),\bx(\tau^+)) = 0$.

\paragraph{Third Case.}
If for $t\in [t_1,\tau)$, $\sum_{f\in F_i} \lambda_f \geq \sum_{b\in B_i} \mu_b(N_b(t))$ for a strict subset of tiers.
We first present a useful lemma, and the proof of this result is deferred to Appendix \ref{apx: proof of B- F-}.

\begin{lemma}\label{lem: hat B- connects to hat F-}
Let $\hat{F}^+, \hat{F}^-, \hat{B}^+, \hat{B}^-$ be defined as in \eqref{eq: def of 8 subsets}, and $F = \hat{F}^+ \cup \hat{F}^-, B = \hat{B}^+ \cup \hat{B}^-$.
If $b\in \hat{B}^-$, then $\mathcal{F}(b) \cap F \subseteq \hat{F}^-$. If $f\in \hat{F}^+$, then $\mathcal{B}(f) \cap B \subseteq \hat{B}^+$.
\end{lemma}
        Therefore,
        \begin{align}\label{eq:backends to frontends flow imbalance}
            \sum_{b\in B^+\cap \hat{B}^-} \mu_b(N_b(\tau)) \stackrel{(a)}{=} \sum_{b\in B^+\cap \hat{B}^-} \mu_b(N_b(\tau^-)) \stackrel{(b)}{<} \sum_{f\in  F^+ \cap (\cup_{b\in B^+\cap \hat{B}^-}\mathcal{F}(b))} \lambda_f \stackrel{(c)}{\leq} \sum_{f\in F^+\cap \hat{F}^-} \lambda_f,
        \end{align}
        where $(a)$ follows from the continuity of $\mu_b(N_b(t))$; $(b)$ follows from $b\in B^+$ at time $\tau^-$ and \eqref{eq: generalization flow imbalance +}; $(c)$ follows from $F \cap (\cup_{b\in B^+\cap \hat{B}^-}\mathcal{F}(b)) \subseteq\hat{F}^-$ by Lemma \ref{lem: hat B- connects to hat F-}.
        Similarly,
        \begin{align}\label{eq:frontends to backends flow imbalance}
            \sum_{f\in F^-\cap \hat{F}^+} \lambda_f  \stackrel{(a)}{\leq} \sum_{b \in B^-\cap (\cup_{f \in F^-\cap \hat{F}^+} \mathcal{B}(f))} \mu_b(N_b(\tau^-)) \stackrel{(b)}{\leq}\sum_{b\in B^-\cap \hat{B}^+} \mu_b(N_b(\tau^-)) \stackrel{(c)}{=} \sum_{b\in B^-\cap \hat{B}^+} \mu_b(N_b(\tau)),
        \end{align}
        where $(a)$ follows from $f\in F^-$ at time $\tau^-$ and \eqref{eq: generalization flow imbalance -}; $(b)$ follows from $B \cap (\cup_{f \in F^-\cap \hat{F}^+} \mathcal{B}(f)) \subseteq \hat{B}^+$ by Lemma \ref{lem: hat B- connects to hat F-}; $(c)$ follows from the continuity of $\mu_b(N_b(t))$.
        Therefore, by Lemma \ref{lem: difference in V_B change},
        \begin{align}
            V_B&(\bN(\tau^-), \bx(\tau^-))  -V_B(\bN(\tau^+),\bx(\tau^+)) \notag \\
            & = 2 \left(\sum_{f\in F^+\cap \hat{F}^-}\lambda_f - \sum_{b\in B^+\cap \hat{B}^-}\mu_b(N_b(\tau))\right) + 2 \left(\sum_{b\in B^-\cap \hat{B}^+} \mu_b(N_b(\tau)) - \sum_{f\in F^-\cap \hat{F}^+} \lambda_f \right)\quad > 0,\label{eq: lyapunov strictly decreasing}
        \end{align}
        where the inequality follows from \eqref{eq:backends to frontends flow imbalance} and \eqref{eq:frontends to backends flow imbalance}.\Halmos \end{proof}
\subsection{Proof of Lemma \ref{lem: difference in V_B change}}\label{apx: proof of difference in V_B}
\begin{proof}
The difference in $V_B$ is
        \begin{align*}
            V_B&(\bN(\tau^-), \bx(\tau^-))  -V_B(\bN(\tau^+),\bx(\tau^+))  \\
            & \stackrel{(a)}{=} \left(\sum_{f\in F^+} \lambda_f - \sum_{b\in B^+} \mu_b(N_b(\tau^-)) + \sum_{b\in B^-} \mu_b(N_b(\tau^-))-\sum_{f\in F^-} \lambda_f \right) \\
            & \quad\quad - \left(\sum_{f\in \hat{F}^+} \lambda_f - \sum_{b\in \hat{B}^+} \mu_b(N_b(\tau^+)) + \sum_{b\in \hat{B}^-} \mu_b(N_b(\tau^+))-\sum_{f\in \hat{F}^-} \lambda_f\right) \\
            & \stackrel{(b)}{=} \left(\sum_{f\in F^+} \lambda_f - \sum_{b\in B^+} \mu_b(N_b(\tau)) + \sum_{b\in B^-} \mu_b(N_b(\tau))-\sum_{f\in F^-} \lambda_f\right) \\
            & \quad\quad - \left(\sum_{f\in \hat{F}^+} \lambda_f - \sum_{b\in \hat{B}^+} \mu_b(N_b(\tau)) + \sum_{b\in \hat{B}^-} \mu_b(N_b(\tau))-\sum_{f\in \hat{F}^-} \lambda_f  \right)\\
            & = \sum_{f\in F^+ \setminus \hat{F}^+} \lambda_f - \sum_{f \in \hat{F}^+ \setminus F^+} \lambda_f + \sum_{b \in \hat{B}^+ \setminus B^+} \mu_b(N_b(\tau)) - \sum_{b\in B^+\setminus \hat{B}^+} \mu_b(N_b(\tau)) \\
            & \quad\quad + \sum_{b\in B^-\setminus \hat{B}^-} \mu_b(N_b(\tau)) - \sum_{b\in \hat{B}^- \setminus B^-} \mu_b(N_b(\tau)) + \sum_{f\in \hat{F}^- \setminus F^-} \lambda_f - \sum_{f\in F^- \setminus \hat{F}^-} \lambda_f \\
            & \stackrel{(c)}{=} \sum_{f\in F^+\cap \hat{F}^-} \lambda_f - \sum_{f\in F^-\cap \hat{F}^+} \lambda_f + \sum_{b\in B^-\cap \hat{B}^+} \mu_b(N_b(\tau)) - \sum_{b\in B^+\cap \hat{B}^-} \mu_b(N_b(\tau)) \\
            & \quad\quad + \sum_{b\in B^-\cap \hat{B}^+} \mu_b(N_b(\tau)) - \sum_{b\in B^+\cap \hat{B}^-} \mu_b(N_b(\tau)) + \sum_{f\in F^+\cap \hat{F}^-} \lambda_f - \sum_{f\in F^-\cap \hat{F}^+} \lambda_f \\
            & = 2 \left(\sum_{f\in F^+\cap \hat{F}^-}\lambda_f - \sum_{b\in B^+\cap \hat{B}^-}\mu_b(N_b(\tau))\right) + 2 \left(\sum_{b\in B^-\cap \hat{B}^+} \mu_b(N_b(\tau)) - \sum_{f\in F^-\cap \hat{F}^+} \lambda_f \right),
        \end{align*}
        where $(a)$ follows from Lemma \ref{lem: slide}; $(b)$ follows from the continuity of $\mu_b(N_b(t))$; $(c)$ follows from $F = F^+ \cup F^- = \hat{F}^+ \cup \hat{F}^-, B = B^+ \cup B^- = \hat{B}^+ \cup \hat{B}^-$.\Halmos \end{proof}

\subsection{Proof of Lemma \ref{lem: hat B- connects to hat F-}}\label{apx: proof of B- F-}
\begin{proof}
For the first half of the result, suppose not, i.e., there exists $\bar{f} \in \mathcal{F}(b) \cap F, \bar f \in \hat{F}^+$. Therefore, for $t\in (\tau, t_2]$, there exists $\bar{b} \in S_{\bar f}(\bN(t))$,
\[\mu'_{\bar{b}}(N_{\bar{b}} (t)) \stackrel{(a)}{<} \mu'_{\bar{b}}(N_{\bar{b}} (\tau))  \stackrel{(b)}{=} \mu'_{\bar{b}}(N_{\bar{b}} (\tau^-)) \stackrel{(c)}{=}  \mu'_{{b}}(N_{{b}} (\tau^-)) \stackrel{(b)}{=} \mu'_{b}(N_b(\tau)) \stackrel{(d)}{\leq} \mu'_b(N_b(t)),\]
where $(a)$ follows as $\bar f\in \hat{F}^+$ thus $\bar{b} \in \hat{B}^+$ and by Lemma \ref{lem: K+ lambda big} $N_{\bar{b}}(t) > N_{\bar{b}}(\tau)$ and $\mu_b'$ is decreasing; $(b)$ follows from the continuity of $\mu_b'(N_b(t))$; $(c)$ follows from the condition that all the backends share the same gradient at time $\tau$; $(d)$ follows as $b\in \hat{B}^-$ and by Lemma \ref{lem: K- lambda small} $N_{{b}}(t) \leq N_{{b}}(\tau)$ and $\mu_b'$ is decreasing. This contradicts $\bar{b} \in S_{\bar f}(\bN(t))$ but $b\notin S_{\bar f}(\bN(t))$ for $t\in (\tau, t_2]$.

The second half can be shown using a similar argument. Suppose not, i.e., there exists $\bar{b} \in \mathcal{B}(f) \cap B, \bar b \in \hat{B}^-$ and $b \in \mathcal{B}(f) \cap B, b \in \hat{B}^+$. Therefore, for $t\in (\tau, t_2]$,
\[\mu'_{\bar{b}}(N_{\bar{b}} (t)) \stackrel{(a)}{>} \mu'_{\bar{b}}(N_{\bar{b}} (\tau))  \stackrel{(b)}{=} \mu'_{\bar{b}}(N_{\bar{b}} (\tau^-)) \stackrel{(c)}{=}  \mu'_{{b}}(N_{{b}} (\tau^-)) \stackrel{(b)}{=} \mu'_{b}(N_b(\tau)) \stackrel{(d)}{\geq} \mu'_b(N_b(t)),\]
where $(a)$ follows as $\bar b\in \hat{B}^-$ and by Lemma \ref{lem: K- lambda small} $N_{\bar{b}}(t) < N_{\bar{b}}(\tau)$ and $\mu_b'$ is decreasing; $(b)$ follows from the continuity of $\mu_b'(N_b(t))$; $(c)$ follows as all the backends share the same gradient at time $\tau$; $(d)$ follows as $b\in \hat{B}^+$ and by Lemma \ref{lem: K+ lambda big} $N_{{b}}(t) \geq N_{{b}}(\tau)$ and $\mu_b'$ is decreasing. This contradicts $\bar{b} \notin S_{\bar f}(\bN(t))$ but $b\in S_{\bar f}(\bN(t))$ for $t\in (\tau, t_2]$.             \Halmos \end{proof}

\subsection{Proof of Corollary \ref{cor: non N star drift positive}}
\begin{proof}
The conditions imply that for all $f\in \mathcal{F}$ there exists $\bx_f \in X_f(\bN)$ such that $\sum_{b\in \mathcal{B}}\left|\sum_{f\in \mathcal{F}(b)}\lambda_f x_{f,b}  - \mu_b(N_b)\right| =0.$
Therefore, by Lemma \ref{lem: m V positive definite}, we must have $\bN = \bN^*$.\Halmos \end{proof}

\section{Convergence Rate}
\subsection{Proof of Lemma \ref{lem: capacity slack}}
\begin{proof} Let $\bN$ be a feasible solution, which is guaranteed to exist by Assumption~\ref{assmp: system stable}. We have for any subset $P \subseteq \mathcal{F}$, $\sum_{f\in P} \lambda_f \leq \sum_{b\in (\cup_{f\in P} \mathcal{B}(f))} \mu_b({N}_b)$.
    Note that $\limsup_{N\to\infty} \mu'_b(N) \leq \limsup_{N\to\infty} \mu_b(N)/N = 0,$
    where the inequality holds as the service rate functions are concave with $\mu_b(0)=0$, and the equality holds as the functions are bounded (Assumption \ref{assmp: increasing concave smooth ell}). Let $\kappa = \min_{b\in \mathcal{B}} \mu_b'(N_b)/2$. Since $\mu_b'$ is continuous (concavity and differentiability imply continuous differentiability) with $\mu_b'(0) > 0$ and $\mu_b'(N) \to 0$, by the intermediate value theorem there exists $\tilde{N}_b$ such that $\mu_b'(\tilde{N}_b) = \kappa$ for all $b\in \mathcal{B}$. Uniqueness follows by strict concavity.

    Then for any subset $P\subseteq \mathcal{F}$, $\sum_{f\in P} \lambda_f \leq \sum_{b\in (\cup_{f\in P} \mathcal{B}(f))} \mu_b({N}_b) < \sum_{b\in (\cup_{f\in P} \mathcal{B}(f))} \mu_b(\tilde{N}_b)$.
    Letting $\Delta=  \min_{b \in \mathcal{B}} \left\{ \mu_b(\tilde{N}_b) - \mu_b({N}_b) \right\}$
    finishes the proof.
\Halmos \end{proof}

\subsection{Proof of Proposition \ref{prop: K invariant set}}
\begin{proof}
    Suppose not, and let $b_1 \in \mathcal{B}$ be the first backend such that
    $\mu'_{b_1}(N_{b_1}(\tau^-)) > \kappa$ and $\mu'_{b_1}(N_{b_1}(\tau^+)) < \kappa$. Let $(F,B)$ be the tier in which $b_1$ lies. By definition of a tier, we have $\mu_b'(N_b(\tau)) \leq \kappa$ for all $b\in \cup_{f\in F} \mathcal{B}(f)$, otherwise the frontends in $F$ should send jobs to better connected backends than $b_1$. Meanwhile, for $b_1$ to be the first backend that moves out of $K$, we have that $\mu'_b(N_b(\tau)) \geq \kappa$ for all other backends $b\in \mathcal{B}$. Therefore, $\mu_b'(N_b(\tau)) = \kappa$ for all $b\in \cup_{f\in F} \mathcal{B}(f)$, i.e., $B = \cup_{f\in F} \mathcal{B}(f)$.

    As $\mu_b'(\cdot)$ is a decreasing function, this implies $\frac{d}{dt}N_{b_1}(\tau) >0$, i.e., the total job arrivals to backend $b_1$ must exceed its service rate. In other words, the tier that $b_1$ lies in, denoted by $(F,B)$, must have 1) $F\neq \emptyset$, and 2) by the contrapositive of Lemma~\ref{lem: K- lambda small}, we must have positive flow imbalance:
    $\sum_{f \in F} \lambda_f > \sum_{b\in B} \mu_b(N_b(\tau)) = \sum_{b\in B} \mu_b(\tilde{N}_b), $
    where the equality holds as 1) $\mu_{b}'(N_{b}(\tau)) = \kappa$ for all $b \in B$ as backends in the same tier share the gradient, and 2) $\kappa = \mu_b'(\tilde N_b)$ by definition of $\tilde N_b$. However, by Lemma \ref{lem: capacity slack}, we have $\sum_{f\in F}\lambda_f \leq \sum_{b\in B}\mu_b(\tilde{N}_b)$, a contradiction.
\Halmos \end{proof}

\subsection{Proof of Proposition \ref{prof: exponential convergence V in K}}
{
\begin{proof}
For state $\bN(t)$, let's define $F^+, F^-, B^+, B^-$ to be a partition of \textit{all} tiers (unlike in \eqref{eq: def of 8 subsets} where the partition is restricted to those involved in the tier reconfiguration):
\begin{align*}
        & F^+ = \bigcup_{i \in [k]: \sum_{f\in F_i} \lambda_f > \sum_{b\in B_i} \mu_b(N_b(t))} F_i, \quad B^+ =  \bigcup_{i \in [k]: \sum_{f\in F_i} \lambda_f > \sum_{b\in B_i} \mu_b(N_b(t))} B_i,\\
        & F^- =  \bigcup_{i \in [k]: \sum_{f\in F_i} \lambda_f \leq \sum_{b\in B_i} \mu_b(N_b(t))}F_i, \quad B^- = \bigcup_{i \in [k]: \sum_{f\in F_i} \lambda_f \leq \sum_{b\in B_i} \mu_b(N_b(t))} B_i.
\end{align*}
where $F_i \subseteq \mF, B_i \subseteq \mB$ and $(F_i, B_i)$ is a tier at $\bN(t)$ for all $i\in [k]$; $F_1 \cup \cdots \cup F_k = \mF, B_1 \cup \cdots B_k = \mB$. We have
    \begin{align*}
    V(\bN(t), \bx(t)) &  = \sum_{b\in \mB} \left|\sum_{f\in \mF} \lambda_f x_{f,b}(t) - \mu_b(N_b(t))\right|\\
    &  = \sum_{b\in B^+} \left(\sum_{f\in \mF} \lambda_f x_{f,b}(t) - \mu_b(N_b(t))\right) + \sum_{b\in B^-} \left(\mu_b(N_b(t)) - \sum_{f\in \mF} \lambda_f x_{f,b}\right)\\
    & = \sum_{f\in F^+}\lambda_f - \sum_{b\in B^+} \mu_b(N_b(t)) + \sum_{b\in B^-}\mu_b(N_b(t)) - \sum_{f\in F^-}\lambda_f
    \end{align*}

The partition $F^+, F^-, B^+, B^-$ remains unchanged in the constant tier case, single-tier splitting case, and the first two subcases of the tier reconfiguration case (where tiers involved all have positive/negative flow imbalances, see proof of Lemma \ref{lem:sync} in Appendix \ref{apx: proof of sync lemma}). As the service rate functions are differentiable, $V(\bN(t), \bx(t))$ is differentiable:
    \begin{align*}
    \frac{d}{dt} V(\bN(t), \bx(t)) & = -\sum_{b\in B^+} \mu_b'(N_b(t))\frac{d}{dt}N_b(t) + \sum_{b\in B^-} \mu_b'(N_b(t))\frac{d}{dt}N_b(t)\\
    & \stackrel{(a)}{=} -\sum_{b\in B^+} \mu_b'(N_b(t))\left|\frac{d}{dt}N_b(t)\right| - \sum_{b\in B^-} \mu_b'(N_b(t))\left|\frac{d}{dt}N_b(t)\right|\\
    & \stackrel{(b)}{\leq} -\sum_{b\in B^+} \kappa \left|\frac{d}{dt}N_b(t)\right|- \sum_{b\in B^-} \kappa\left|\frac{d}{dt}N_b(t)\right| \stackrel{(c)}{=} -\kappa V(\bN(t), \bx(t)),
    \end{align*}
    where $(a)$ follows from the definition of $B^+, B^-$ where backends in $B^+$ experience positive flow imbalance and vice versa; $(b)$ follows from $\mu_b'(N_b(t)) \geq \kappa$ as $\bN(t) \in K$; $(c)$ follows from the definition of the Lyapunov function.

    The remaining subcase of the tier reconfiguration case (where tiers involved don't all have positive/negative flow imbalances) is the only situation in which the partition $F^+, F^-, B^+, B^-$ can change. In this case, the Lyapunov function might not be differentiable. By \eqref{eq: lyapunov strictly decreasing} in Appendix \ref{apx: proof of sync lemma}, each such change causes a strict drop in the Lyapunov function. Because the Lyapunov function is non-increasing, it can have a countable number of discontinuities. Therefore, the number of times of the partition changes is countable and the condition on the statement of the proposition holds almost everywhere.
\Halmos \end{proof} }

\subsection{Proof of Proposition \ref{prop:exponential-convergence-sum-rifs}}
\begin{proof}
    For all $b \in \mathcal{B}$ let $\invrate_b$ denote the inverse function of $\mu_b$.
    Since $\mu_b$ is concave and strictly increasing, $\invrate_b$ is convex and strictly
    increasing on its domain of definition. The value $\invrate_b(y)$ is undefined when
    $y > \sup_{N} \{ \mu_b(N) \}$.

    For all $b \in \mathcal{B}$ let
    $\inflow^*_b = \sum_{f \in \mathcal{F}(b)} \lambda_f x^*_{f,b} = \mu_b(N^*_b)\, \inflow_b = \sum_{f \in \mathcal{F}(b)} \lambda_f x_{f,b}\, \tilde{\inflow}_b = \mu_b(N_b)\,.$
    We can interpret $\inflow^*_b$ as the inflow/service rate of backend $b$ at the optimal solution, $\inflow_b$ as the inflow at backend $b$ under $\bx$, and $\tilde{\inflow}_b$ as the service rate of backend $b$ at $\bN$.
    By the definition of the Lyapunov function $V$, we have
    \begin{equation} \label{eq:lyapunov-is-l1}
        V(\bN,\bx) = \| \binflow - \tilde{\binflow} \|_1 .
    \end{equation}
    By the chain rule and our assumption that $\bN \in K$, we have
    \begin{equation} \label{eq:chain-rule}
        \invrate'_b(\tilde{\inflow}_b) = \frac{1}{\mu'_b(N_b)} \in \left( 0, \frac{1}{\kappa} \right] .
    \end{equation}

    We first argue that $\sum_{b \in \mathcal{B}} N^*_b - \sum_{b \in \mathcal{B}} N_b \ge - V(\bN,\bx) / \kappa$. For the convex function
    $R(\bm{v}) = \sum_{b \in \mathcal{B}} \invrate_b(v)$
    we have
    \begin{align*}
        \sum_{b \in \mathcal{B}} N^*_b =
        R(\binflow^*) & \geq
        R(\tilde{\binflow}) +
        \langle \nabla R(\tilde{\binflow}), \, \binflow^* - \tilde{\binflow} \rangle  = R(\tilde{\binflow}) +
        \langle \nabla R(\tilde{\binflow}), \, \binflow^* - \binflow\rangle
        - \langle \nabla R(\tilde{\binflow}), \, \tilde{\binflow} - \binflow \rangle \\
        & \stackrel{(a)}{\geq}
        R(\tilde{\binflow})
        \, - \, \| \nabla R(\tilde{\binflow}) \|_{\infty}
        \| \binflow - \tilde{\binflow} \|_1 \, + \,
        \langle \nabla R(\tilde{\binflow}), \, \binflow^* - \binflow\rangle \\
        & \stackrel{(b)}{\geq}
        \sum_{b \in \mathcal{B}} N_b \, - \,
        \frac{V(\bN,\bx)}{\kappa} \, + \,
        \sum_{f \in \mathcal{F}} \lambda_f \langle \nabla R(\tilde{\binflow}), \, \bx^*_f - \bx_f \rangle  \stackrel{(c)}{\geq}
        \sum_{b \in \mathcal{B}} N_b \, - \,
        \frac{V(\bN,\bx)}{\kappa}\,,
    \end{align*}
    where (a) follows from H\"older inequality.
    Line (b) is justified by the following considerations.
    First, $R(\tilde{\binflow}) = \sum_{b \in \mathcal{B}} \invrate_b(\mu_b(N_b)) = \sum_{b \in \mathcal{B}} N_b$ by the definitions of $R(\cdot)$ and $\tilde{\binflow}$.
    Second, $\| \nabla R(\tilde{\binflow}) \|_{\infty} \| \binflow - \tilde{\binflow} \|_1 \leq V(\bN,\bx)/\kappa$ by equation~\eqref{eq:lyapunov-is-l1} and inequality~\eqref{eq:chain-rule}.
    Third, $\binflow^* = \sum_{f \in \mathcal{F}}\lambda_f \bx^*_f$ and $\binflow = \sum_{f \in \mathcal{F}} \lambda_f \bx_f$.
    Lastly, to justify line (c), we argue that for each $f \in \mathcal{F}$,
    \begin{equation} \langle \nabla R(\tilde{\binflow}), \, \bx_f \rangle
    \leq \langle \nabla R(\tilde{\binflow}), \, \bx^*_f \rangle.
    \label{eq:bxf-optimizes}
    \end{equation}
    In fact, the relation $\bx_f \in X_f(N_b)$ implies that
    $x_{f,b} > 0$ only for those $b \in \mathcal{B}(f)$ that
    maximize $\mu'_b(N_b)$. According to formula~\eqref{eq:chain-rule}
    these are exactly the $b \in \mathcal{B}(f)$ that
    minimize $\invrate'_b(\tilde{\inflow}_b)$. Accordingly,
    $\bx_f$ attains the minimum of
    $\langle \nabla R (\tilde{\binflow}), \, \bm{z} \rangle$
    over all $\bm{z}$ in the probability simplex
    $\bm{\Delta}(\mathcal{B}(f))$.
    Since $\bx^*_f$ also belongs to
    $\bm{\Delta}(\mathcal{B}(f))$, the
    inequality~\eqref{eq:bxf-optimizes} follows.

{
    We next argue that $\sum_{b \in \mathcal{B}} N_b - \sum_{b \in \mathcal{B}} N_b^* \ge - V(\bN,\bx) / \kappa$. As before, we have by convexity of $R(\cdot)$
    \begin{align*}
        \sum_{b \in \mathcal{B}} N_b =
        R(\tilde{\binflow}) & \geq
        R(\binflow^*) +
        \langle \nabla R(\binflow^*), \, \tilde{\binflow} - \binflow^*  \rangle = R(\binflow^*) -
        \langle \nabla R(\binflow^*), \, \binflow - \tilde{\binflow} \rangle
        + \langle \nabla R(\binflow^*), \,  \binflow -  \binflow^*\rangle \\
        & \stackrel{(a)}{\geq}
        R(\binflow^*)
        \, - \, \| \nabla R(\binflow^*) \|_{\infty}
        \| \binflow - \tilde{\binflow} \|_1 \, + \,
        \langle \nabla R(\binflow^*), \, \binflow - \binflow^* \rangle \\
        & \stackrel{(b)}{\geq}
        \sum_{b \in \mathcal{B}} N^*_b \, - \,
        \frac{V(\bN,\bx)}{\kappa} \, + \,
        \sum_{f \in \mathcal{F}} \lambda_f \langle \nabla R(\binflow^*), \, \bx_f - \bx^*_f \rangle \stackrel{(c)}{\geq}
        \sum_{b \in \mathcal{B}} N_b^* \, - \,
        \frac{V(\bN,\bx)}{\kappa}\,,
    \end{align*}
    where (a) follows from H\"older inequality.

    Line (b) is justified by the following considerations.
    First, $R(\binflow^*) = \sum_{b \in \mathcal{B}} \invrate_b(\mu_b(N_b^*)) = \sum_{b \in \mathcal{B}} N_b^*$ by the definitions of $R(\cdot)$ and $\binflow^*$.
    Second, $\| \nabla R(\binflow^*) \|_{\infty} \| \binflow - \tilde{\binflow} \|_1 \leq V(\bN,\bx)/\kappa$ by equation~\eqref{eq:lyapunov-is-l1} and inequality~\eqref{eq:chain-rule}.
    Third, $\binflow^* = \sum_{f \in \mathcal{F}} \lambda_f \bx^*_f$ and $\binflow = \sum_{f \in \mathcal{F}} \lambda_f \bx_f$.
    Lastly, to justify line (c), we argue that for each $f \in \mathcal{F}$,
    \begin{equation} \langle \nabla R(\binflow^*), \, \bx_f^* \rangle
    \leq \langle \nabla R(\binflow^*), \, \bx_f \rangle.
    \label{eq:bxf-optimizes-opp}
    \end{equation}
    This is because $\bx^*_f$ attains the minimum of
    $\langle \nabla R (\binflow^*), \, \bm{z} \rangle$
    over all $\bm{z}$ in the probability simplex
    $\bm{\Delta}(\mathcal{B}(f))$.
    Since $\bx_f$ also belongs to
    $\bm{\Delta}(\mathcal{B}(f))$, the
    inequality~\eqref{eq:bxf-optimizes-opp} follows.}
\Halmos \end{proof}

\subsection{Proof of Proposition \ref{prop: linear converges to K}}
\begin{proof}
Define $Q(t) = \{b \in \mB: N_b(t) > \tilde{N}_b\}\cup \{b \in \mB: N_b(t) = \tilde{N}_b, \dot{N}_b(t) \geq 0\}$.
When $\bN(t) \notin K$, $Q(t)$ is non-empty, and for all $b \in Q(t)$  the service rates satisfy
\begin{align}\label{eq:monotonicity-of-ell}
    \mu_{b}(N_{b}(t)) \geq \mu_{b}(\tilde{N}_{b})\,,
\end{align}
because $\mu_{b}(\cdot)$ is increasing.

Consider a Lyapunov function $J(\bN) = \sum_{b\in \mB} \max\{N_b - \tilde{N}_b, 0\} $, then we have $J(\bN) \geq 0$ for all $\bN$, and $J(\bN) = 0$ if and only if $\bN \in K$.
We proceed to show that $J(\bN(t))$ decreases (at least) linearly fast with rate $\min\{\Delta,  \min_{b\in \mB} \mu_b(\tilde{N}_b)\}$ when $\bN(t) \notin K$, thus it takes no more than $J(\bN(0))/\min\{\Delta,  \min_{b\in \mB} \mu_b(\tilde{N}_b)\}$ for $\bN(t)$ to reach the invariant set $K$ starting from $\bN(0)$ at time $t=0$.
By definition of $J(\bN(t))$ and $Q(t)$, we have $\frac{d}{dt} J(\bN(t)) {=} \sum_{b\in Q(t)} \frac{d}{dt} N_b(t).$

Let $P(t)$ denote the set of frontends that share tiers with backends $b\in Q(t)$, which could be an empty set.
Under GMSR, frontends always send arrivals to connected backends with highest gradient, thus no frontend in $P(t)$ is connected to a backend not in $Q(t)$ ($(f,b) \notin \mathcal{E}$ for $f\in P(t), b \notin Q(t)$), otherwise it does not belong to the current tier.
\begin{itemize}
    \item If $P(t) = \emptyset$,
    \[    \frac{d}{dt} J(\bN(t)) = \sum_{b\in Q(t)} \frac{d}{dt} N_b(t) \stackrel{(a)}{=}  - \sum_{b\in Q(t)} \mu_b(N_b(t)) \stackrel{(b)}{\leq} - \sum_{b\in Q(t)} \mu_b(\tilde{N}_b) \leq - \min_{b\in \mB} \mu_b(\tilde{N}_b), \]
    where $(a)$ holds as $P(t) = \emptyset$; $(b)$ follows from \eqref{eq:monotonicity-of-ell}.
    \item If $P(t) \neq \emptyset$,
    \[    \frac{d}{dt} J(\bN(t)) = \sum_{b\in Q(t)} \frac{d}{dt} N_b(t) \stackrel{(c)}{=} \sum_{f\in P(t)} \lambda_f - \sum_{b\in Q(t)} \mu_b(N_b(t)) \stackrel{(d)}{\leq} \sum_{f\in P(t)} \lambda_f - \sum_{b\in Q(t)} \mu_b(\tilde{N}_b) \stackrel{(e)}{\leq} -\Delta, \]
where $(c)$ follows from \eqref{eq: sum of dotN in B} in Appendix~\ref{sec:appendix-stability} that holds for each tier and here we sum over all tiers that are formed with frontends $P(t)$ and backends $Q(t)$; $(d)$ follows from \eqref{eq:monotonicity-of-ell}; and $(e)$ follows from Lemma \ref{lem: capacity slack}.
\end{itemize} \Halmos \end{proof}

\input{overloaded_appendix}

\section{Simulation}\label{apx: simulation}

In this section, we perform some simulations to illustrate the behavior of the GMSR policy in the discrete-time stochastic system and its convergence to the fluid model. We consider an example with two frontends connected to two backends as depicted in the ``N" model shown in Figure~\ref{fig: N model}. 

\begin{figure}[!htbp]
    \centering
\begin{tikzpicture}
    \node[circle, draw] (n1) at (0,0) {$f_1$};
    \node[circle, draw] (n2) at (0,-2) {$f_2$};
    \node[circle, draw] (l1) at (2,0) {$b_1$};
    \node[circle, draw] (l2) at (2,-2) {$b_2$};
    \node at (4,0) {$\mu_1(N_1) = \frac{N}{N+1}$};
    \node at (4,-2) {$\mu_2(N_2) = \frac{N}{N+2}$};
    \draw[->] (n1) -- (l1);
    \draw[->] (n2) -- (l2);
    \draw[->] (n2) -- (l1);
    \node (lambda1) at (-2,0) {$\lambda_1 = 0.4$};
    \node (lambda2) at (-2,-2) {$\lambda_2 = 0.6$};
    \draw[->] (lambda1) -- (n1);
    \draw[->] (lambda2) -- (n2);
\end{tikzpicture}
    \caption{``N'' Model}
    \label{fig: N model}
\end{figure}

We consider a sequence of discrete systems indexed by $(c)$ whose dynamics within time interval $[0,T]$ are as follows:
\begin{align*}
    &N_1^{(c)}(i+1) = N_1^{(c)}(i) + A_{1,1}^{(c)}(i) + A_{2,1}^{(c)}(i) - D_1^{(c)}(i), \quad N_2^{(c)}(i+1) = N_2^{(c)}(i) + A_{2,2}^{(c)}(i) - D_2^{(c)}(i), \\
    & A_{1,1}^{(c)}(i) = \text{Poisson}(0.4), \quad A_{2,1}^{(c)}(i) + A_{2,2}^{(c)}(i) = \text{Poisson}(0.6), \\
    & D_1^{(c)}(i) = 0 \text{ with probability } \frac{1}{N+1}, D_1^{(c)}(i) = 1 \text{ with probability } \frac{N}{N+1}, \\
    & D_2^{(c)}(i) = 0 \text{ with probability } \frac{2}{N+2}, D_2^{(c)}(i) = 1 \text{ with probability } \frac{N}{N+2}.
\end{align*}

We will plot the normalized workload defined as
\begin{align*}
    &Y_1^{(c)}(i+1) = Y_1^{(c)}(i) + \frac{1}{c}\left(A_{1,1}^{(c)}(i) + A_{2,1}^{(c)}(i) - D_1^{(c)}(i)\right), \quad Y_2^{(c)}(i+1) = Y_2^{(c)}(i) + \frac{1}{c}\left(A_{2,2}^{(c)}(i) - D_2^{(c)}(i)\right),
\end{align*}
and the solution, denoted by $N_1(t), N_2(t)$ to the differential inclusion
\[\dot{N}_1(t) = \lambda_1 + \lambda_2 x_{2,1}(t)- \mu_1(N_1(t)), \dot{N}_2(t) = \lambda_2 x_{2,2}(t)- \mu_2(N_2(t)),\]
where $ (x_{2,1}(t), x_{2,2}(t)) \in X_2((N_1(t), N_2(t))$. Note that the physical time for one discrete time step is $1/c$ in the $c^{th}$ step, thus the normalized workload at time $t$ with $t = i/c$ is $Y_b^{(c)}(i)$. The equilibrium of the system is $N_1^* = 2, N_2^* = 1$.

In panels (a)--(c) of Figure \ref{fig: workloads}, we plot the workloads at the two backends as a function of time until $T = 50$, each panel starting with a different initial state. These plots show both the workloads in the discrete-time stochastic model as well as the solution to the differential inclusion. We vary the system scaling parameter $c$, which affects the physical length of the time step and job size. We observe 1) as time progresses, the system workloads approach the equilibrium values and oscillate around it; and 2) as $c$ increases, the fluctuations in the stochastic system decrease and fluid model better approximates the stochastic system.

In panel (d) of Figure \ref{fig: workloads}, we present a scatter plot of the workloads at the two backends for various values of $c$. Here the system starts at the equilibrium point $N_1 = 2, N_2 = 1$, and we can see the system states oscillate around the equal gradient curve (the dotted curve) in the plot. As $c$ increases, the range of oscillation reduces, i.e., the system state orbits closely around the equilibrium point.

\begin{figure}[!htbp]
    \centering
    \captionsetup[subfigure]{font=footnotesize}
    \subcaptionbox{Initial state $(0,0)$}[0.48\textwidth]
    {%
    \resizebox{0.46\textwidth}{!}{%
    \begin{tikzpicture}[every node/.style={scale=0.8}]
        \begin{axis}[
            width=0.7\textwidth,
            height=0.4\textwidth,
            xlabel={Time $t$},
            ylabel={Workload},
            xmin = -2, xmax = 52,
            legend style={at={(1.05,1)}, anchor=north west},
            grid=major,
            cycle list name=color list,
        ]
        \pgfplotsset{
            cycle list={blue, red, green, orange, violet, cyan, teal, magenta, brown, gray}
        }
        \addplot table [x=Time, y={Queue 1 (c=20)}, col sep=comma] {data/queue_lengths_and_diff_inclusion.csv};
        \addlegendentry{$Y_1^{(20)}(t)$}
        \addplot table [x=Time, y={Queue 2 (c=20)}, col sep=comma] {data/queue_lengths_and_diff_inclusion.csv};
        \addlegendentry{$Y_2^{(20)}(t)$}
        \addplot table [x=Time, y={Queue 1 (c=100)}, col sep=comma] {data/queue_lengths_and_diff_inclusion.csv};
        \addlegendentry{$Y_1^{(100)}(t)$}
        \addplot table [x=Time, y={Queue 2 (c=100)}, col sep=comma] {data/queue_lengths_and_diff_inclusion.csv};
        \addlegendentry{$Y_2^{(100)}(t)$}
        \addplot[thick, color=black] table [x=Time, y={N1 (Differential Inclusion)}, col sep=comma] {data/queue_lengths_and_diff_inclusion.csv};
        \addlegendentry{$N_1(t)$}
        \addplot[thick, color=violet] table [x=Time, y={N2 (Differential Inclusion)}, col sep=comma] {data/queue_lengths_and_diff_inclusion.csv};
        \addlegendentry{$N_2(t)$}
        \end{axis}
    \end{tikzpicture}
    }%
    }%
    \hfill
    \subcaptionbox{Initial state $(1,2)$}[0.48\textwidth]
    {%
    \resizebox{0.46\textwidth}{!}{%
    \begin{tikzpicture}[every node/.style={scale=0.8}]
        \begin{axis}[
            width=0.7\textwidth,
            height=0.4\textwidth,
            xlabel={Time $t$},
            ylabel={Workload},
            xmin = -2, xmax = 52,
            legend style={at={(1.05,1)}, anchor=north west},
            grid=major,
            cycle list name=color list,
        ]
        \pgfplotsset{
            cycle list={blue, red, green, orange, violet, cyan, teal, magenta, brown, gray}
        }
        \addplot table [x=Time, y={Queue 1 (c=20)}, col sep=comma] {data/queue_lengths_and_diff_inclusion12.csv};
        \addlegendentry{$Y_1^{(20)}(t)$}
        \addplot table [x=Time, y={Queue 2 (c=20)}, col sep=comma] {data/queue_lengths_and_diff_inclusion12.csv};
        \addlegendentry{$Y_2^{(20)}(t)$}
        \addplot table [x=Time, y={Queue 1 (c=100)}, col sep=comma] {data/queue_lengths_and_diff_inclusion12.csv};
        \addlegendentry{$Y_1^{(100)}(t)$}
        \addplot table [x=Time, y={Queue 2 (c=100)}, col sep=comma] {data/queue_lengths_and_diff_inclusion12.csv};
        \addlegendentry{$Y_2^{(100)}(t)$}
        \addplot[thick] table [x=Time, y={N1 (Differential Inclusion)}, col sep=comma] {data/queue_lengths_and_diff_inclusion12.csv};
        \addlegendentry{$N_1(t)$}
        \addplot[thick, color=violet] table [x=Time, y={N2 (Differential Inclusion)}, col sep=comma] {data/queue_lengths_and_diff_inclusion12.csv};
        \addlegendentry{$N_2(t)$}
        \end{axis}
    \end{tikzpicture}
    }%
    }%
    \\[-0.2em]
    \subcaptionbox{Initial state $(2,4)$}[0.48\textwidth]
    {%
    \resizebox{0.46\textwidth}{!}{%
    \begin{tikzpicture}[every node/.style={scale=0.8}]
        \begin{axis}[
            width=0.7\textwidth,
            height=0.4\textwidth,
            xlabel={Time $t$},
            ylabel={Workload},
            xmin = -2, xmax = 52,
            legend style={at={(1.05,1)}, anchor=north west},
            grid=major,
            cycle list name=color list,
        ]
        \pgfplotsset{
            cycle list={blue, red, green, orange, violet, cyan, teal, magenta, brown, gray}
        }
        \addplot table [x=Time, y={Queue 1 (c=20)}, col sep=comma] {data/queue_lengths_and_diff_inclusion24.csv};
        \addlegendentry{$Y_1^{(20)}(t)$}
        \addplot table [x=Time, y={Queue 2 (c=20)}, col sep=comma] {data/queue_lengths_and_diff_inclusion24.csv};
        \addlegendentry{$Y_2^{(20)}(t)$}
        \addplot table [x=Time, y={Queue 1 (c=100)}, col sep=comma] {data/queue_lengths_and_diff_inclusion24.csv};
        \addlegendentry{$Y_1^{(100)}(t)$}
        \addplot table [x=Time, y={Queue 2 (c=100)}, col sep=comma] {data/queue_lengths_and_diff_inclusion24.csv};
        \addlegendentry{$Y_2^{(100)}(t)$}
        \addplot[thick] table [x=Time, y={N1 (Differential Inclusion)}, col sep=comma] {data/queue_lengths_and_diff_inclusion24.csv};
        \addlegendentry{$N_1(t)$}
        \addplot[thick, color=violet] table [x=Time, y={N2 (Differential Inclusion)}, col sep=comma] {data/queue_lengths_and_diff_inclusion24.csv};
        \addlegendentry{$N_2(t)$}
        \end{axis}
    \end{tikzpicture}
    }%
    }%
    \hfill
    \subcaptionbox{Scatter around equilibrium\label{fig: workloads scatter 12}}[0.48\textwidth]
    {%
    \resizebox{0.46\textwidth}{!}{%
    \begin{tikzpicture}
        \begin{axis}[
            width=0.8\textwidth,
            height=0.5\textwidth,
            xlabel={$Y_1^{(c)}$ (Queue 1 Length)},
            ylabel={$Y_2^{(c)}$ (Queue 2 Length)},
            xmin=1.6, xmax=2.6,
            ymin=0.6, ymax=1.7,
            legend style={at={(0.95,0.05)}, anchor=south east},
            grid=major
        ]
        \addplot[only marks, mark=square*, color=pink] table [x=N1_c20, y=N2_c20, col sep=comma] {data/queue_scatter_data_combined.csv};
        \addlegendentry{$c=20$}
        \addplot[only marks, mark=triangle*, color=orange] table [x=N1_c50, y=N2_c50, col sep=comma] {data/queue_scatter_data_combined.csv};
        \addlegendentry{$c=50$}
        \addplot[only marks, mark=diamond*, color=teal] table [x=N1_c100, y=N2_c100, col sep=comma] {data/queue_scatter_data_combined.csv};
        \addlegendentry{$c=100$}
        \addplot[only marks, mark=pentagon*, color=cyan] table [x=N1_c500, y=N2_c500, col sep=comma] {data/queue_scatter_data_combined.csv};
        \addlegendentry{$c=500$}
        \addplot[only marks, mark=*, color=blue] table [x=N1_c1000, y=N2_c1000, col sep=comma] {data/queue_scatter_data_combined.csv};
        \addlegendentry{$c=1000$}
        \addplot[only marks, mark=*, color=red] coordinates {(2,1)};
        \addlegendentry{$\bN^*$}
        \addplot[dashed, gray] coordinates {(1,0) (3,2)};
        \end{axis}
    \end{tikzpicture}
    }%
    }%
    \caption{Finite-system behavior in the stochastic model. Panels (a)--(c) compare normalized workloads for different values of $c$ with the fluid model as a function of physical time $t$. Panel (d) shows the scatter plot of normalized workloads for different $c$ values from initial state $(2,1)$; the dotted line is the equal-gradient curve $\{(N_1,N_2): \mu_1'(N_1)=\mu_2'(N_2)\}$.}
    \label{fig: workloads}
\end{figure}

\end{document}

%% file: overloaded_appendix.tex
\section{Preliminary for Network Flow Terminologies}\label{apx: graph def preliminary}
\begin{definition}[Maximum s-t Flow, Definition 2.3 \cite{williamson2019network}]
Given a directed graph $G = (V,A)$, and capacities on the arcs $u(i,j)\geq 0$. An s-t flow $w(i,j)$ is an assignment of reals to the arcs such that the following three properties are obeyed:
\begin{itemize}
    \item for all arcs $(i,j)\in A, w(i,j)\leq u(i,j)$;
    \item for all $i\in V$ such that $i\neq s, t$, the net flow leaving $i$ is zero; that is
    \[\sum_{k:(i,k)\in A} w(i,k) = 0;\]
    \item for all $(i,j) \in A$,
    \[w(i,j) = -w(j,i).\]
\end{itemize}
The value of an s-t flow $w$ is
\[|w| = \sum_{k:(s,k)\in A} w(s,k) - \sum_{k:(k,s) \in A} w(k,s).\]
A maximum s-t flow is an s-t flow of maximum value.
\end{definition}

\begin{definition}[Minimum s-t Cut]
Given a directed graph $G = (V,A)$, and capacities on the arcs $u(i,j)\geq 0$. An s-t cut $(S,T)$ is a partition of $V$ such that $s\in S$ and $t \in T$. The capacity of an s-t cut $(S,T)$ is the sum of the capacities of all the arcs leaving $S$ to $T$,
    \[c(S,T) = \sum_{(i,j) \in A: i\in S, j\in T}u(i,j).\]
A minimum s-t cut is an s-t cut of minimum capacity.
\end{definition}

\begin{definition}[Residual Graph]
Given flow $w$ on a directed graph $G = (V,A)$, the residual graph with respect to flow $w$ is a graph $G_w = (V,A)$, where each arc has a residual capacity $u_w(i,j) = u(i,j) - w(i,j)$. Let $A_w = \{(i,j)\in A: u_w(i,j) >0\}$ denote the set of all arcs with positive residual capacity.
Arcs $(i,j)$ with zero residual capacity are said to be saturated. Note that with the definition that uses reverse arcs, the residual capacity of such reverse arcs $(j,i)$ is $w(i,j)$.
\end{definition}

\section{Proofs and Details for the Overloaded System}\label{apx: overload system}

\subsection{Stability Decomposition}
When the total processing capacity of the backends cannot satisfy all incoming requests, the total workload in the system will explode, i.e.,
\[\lim_{t\to\infty}\sum_{b\in \mB} N_b(t) = \infty.\]

To identify the subset of backends that can still be stabilized, we introduce the following stability decomposition.

\begin{definition}[Stability Decomposition]\label{def: stability decomposition}
    Given the connectivity graph $\mathcal{G} = (\mathcal{F},\mathcal{B},\mathcal{E})$, arrival rates $\{\lambda_f\}_{f\in \mathcal{F}}$ and service rate functions $\{\mu_b(\cdot)\}_{b\in \mathcal{B}}$, we define $\tilde{\mathcal{F}}\subseteq \mathcal{F}, \tilde{\mathcal{B}}\subseteq \mathcal{B}$ to be the sets that satisfy the following properties:
\begin{itemize}
    \item $\tF = \cup_{b\in \tB} \mF(b)$;
    \item for every subset $P\subseteq \tilde{\mathcal{F}}$,
    \[\sum_{f\in P}\lambda_f < \sum_{b\in \tilde{\mathcal{B}} \cap (\cup_{f\in P} \mathcal{B}(f))} \mu_b(\infty);\]
    \item for every subset $Q\subseteq \mathcal{B}\setminus\tilde{\mathcal{B}}$,
    \[\sum_{f\in \cup_{b\in Q}\mathcal{F}(b) \setminus \tilde{\mathcal{F}}}\lambda_f \geq \sum_{b\in Q} \mu_b(\infty).\]
\end{itemize}
\end{definition}
The first condition establishes the ``inverted-Z'' structure between $\tF, \tB, \mF\setminus \tF, \mB\setminus\tB$ as shown in Figure \ref{fig:maxflowmincutQ}.
Specifically, the first condition requires $\tF$ to be the set of frontends that connect to backends in $\tB$.

The last two conditions establish contrasting stability properties for the two subgraphs:
\begin{itemize}
    \item First subgraph: $(\tF, \tB, \{(f,b) \in \mathcal{E}: f\in \tF, b\in \tB\})$.
    \item Second subgraph: $(\mF\setminus\tF, \mB\setminus\tB, \{(f,b) \in \mathcal{E}: f\in \mF\setminus\tF, b\in \mB\setminus\tB\})$.
\end{itemize}
The second condition ensures the workload in the first subgraph can be stabilized by requiring that for any subset of frontends in $\tF$, the maximum service capacity of their connected backends in $\tB$ strictly exceed the total arrival rates.
The last condition establishes that the second subgraph is unstable, i.e., all backends' workload will grow unbounded by requiring that for any subset of backends in $\mB\setminus\tB$, their total service capacity be insufficient to handle the incoming jobs from frontends in $\mF\setminus\tF$ connected to them.

Next, we establish the equivalence between stability decomposition and minimum cuts by constructing an extended flow network.
Consider a graph, denoted by $\tilde{\mathcal{G}}$, with a source $s$ connected to all frontends in $\mF$ and a sink $t$ connected to all backends in $\mB$:
\begin{align}\label{eq: def of tG}
    &\tilde{\mathcal{G}} = (\{s, t\}\cup \mF\cup \mB, \{(f,b):(f,b) \in \mathcal{E}\}\cup\{(s,f): f\in \mF\}\cup\{(b,t): b \in \mB\}),\\
    & \text{ with }u(f,b) = \infty \text{ for }(f,b) \in \mathcal{E}, u(s,f) = \lambda_f, \text{ for } f\in \mF, u(b,t) = \mu_b(\infty), \text{ for } b\in \mB. \notag
\end{align}
where $u(i,j)$ denote the capacity on the arcs $(i,j)$.

\begin{figure}[htbp]
    \centering
    \includegraphics[width=0.5\linewidth]{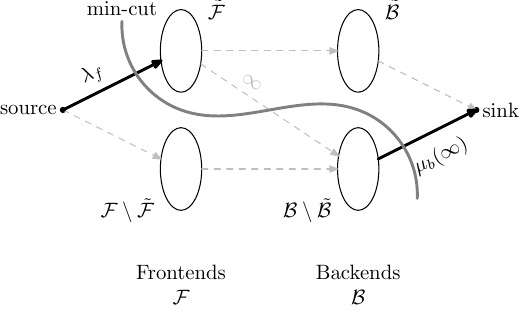}
    \caption{Illustration of the Minimum s-t Cut Given by the Stability Decomposition in $\tilde{\mathcal{G}}$ (defined in \eqref{eq: def of tG}).}
    \label{fig:maxflowmincutQ}
\end{figure}

\begin{lemma}[A Stability Decomposition is a Minimum s-t Cut with Least Balanced Cardinality]\label{lem: stability decomposition is a cut}
Let $\tF, \tB$ be the sets defined in Definition \ref{def: stability decomposition}, then $(\{s\}\cup\{\mF\setminus\tF\}\cup\{\mB\setminus\tB\}, \{t\}\cup \tF\cup\tB)$ is a minimum s-t cut for the directed graph $\tilde{\mathcal{G}}$ defined in \eqref{eq: def of tG} with least balanced cardinality. That is, for every minimum s-t cut $(S,T)$ of graph $\tilde{\mathcal{G}}$, we have $\{\{t\}\cup \tF \cup \tB \} \subseteq T$.
\end{lemma}
See Figure \ref{fig:maxflowmincutQ} for illustration of the minimum s-t cut.
\begin{proof}
The capacity of the s-t cut is given by the sum of the capacities of all the arcs connecting the two parts, i.e., $\sum_{f\in \tF} \lambda_f + \sum_{b\in \mB\setminus\tB} \mu_b(\infty)$.

\textbf{Part 1.}
We prove the first part of the result by showing the max-flow of this graph is $\sum_{f\in \tF} \lambda_f + \sum_{b\in \mB\setminus\tB} \mu_b(\infty)$, which by the Max-Flow Min-Cut Theorem implies that $(\{s\}\cup\{\mF\setminus\tF\}\cup\{\mB\setminus\tB\}, \{t\}\cup \tF\cup\tB)$ is a minimum s-t cut.

We first state two useful lemmas.
\begin{lemma}
If for every subset of frontends $P \subseteq \mF$,
\[\sum_{f\in P} \lambda_f < \sum_{b\in \cup_{f\in P}\mB(f)} \mu_b(\infty),\]
the max flow of the graph $\tilde{\mathcal{G}}$ is $\sum_{f\in \mF}\lambda_f$.
\end{lemma}
\begin{proof}
If we relax the strict inequality condition, the result follows by the same max-flow argument used in the proof of Lemma~\ref{lem: OPT lower bound}. \Halmos
\end{proof}

\begin{lemma}
If for every subset of backends $Q \subseteq \mB$,
\[\sum_{f\in \cup_{b\in Q} \mF(b)} \lambda_f \geq \sum_{b\in Q} \mu_b(\infty),\]
the max flow of the graph $\tilde{\mathcal{G}}$ is $\sum_{b\in \mB} \mu_b(\infty)$.
\end{lemma}
\begin{proof}
    The result is a direct application of the above lemma due to symmetry, i.e., switching frontends and backends. \Halmos
\end{proof}

Applying the two lemmas into different subgraphs of the connectivity graph $\tilde{\mathcal{G}}$, we have:

\begin{enumerate}
    \item the max flow of the subgraph
\[(\tF, \tB, \{(f,b)\in \mathcal{E}: f\in \tF, b\in \tB\})\]
is $\sum_{f\in \tF} \lambda_f$;
    \item the max flow of the subgraph
\[(\mF\setminus\tF, \mB\setminus\tB, \{(f,b)\in \mathcal{E}: f\in \mF\setminus\tF, b\in \mB\setminus\tB\})\]
is $\sum_{b\in \mB\setminus\tB} \mu_b(\infty)$.
    \item the max flow of the subgraph
\[(\mF\setminus \tF, \mB, \{(f,b)\in \mathcal{E}: f\in \mF\setminus \tF, b\in \mB\})\]
is equivalent to the one in 2), i.e, $\sum_{b\in \mB\setminus\tB} \mu_b(\infty)$, as $\tB$ and $\mF\setminus\tF$ are not connected.
\end{enumerate}
Using 1) and 2), as the disjoint union of the two subgraphs is a subgraph of the connectivity graph $\mathcal{G}$, the max flow of the whole graph is no smaller than $\sum_{f\in \tF} \lambda_f + \sum_{b\in \mB\setminus \tB} \mu_b(\infty)$.

Using 3) and the fact that adding the set of frontends $\tF$ back (along with all edges) can at most add $\sum_{f\in \tF}\lambda_f$ flow to the graph, the max flow of the whole graph is no bigger than $\sum_{f\in \tF} \lambda_f + \sum_{b\in \mB\setminus \tB} \mu_b(\infty)$.

Therefore, the max flow of the graph is $\sum_{f\in \tF} \lambda_f + \sum_{b\in \mB\setminus \tB} \mu_b(\infty)$.

\textbf{Part 2.} We proceed to show the stability decomposition is the minimum s-t cut with least balanced cardinality.
First, note that arc capacity $u(f,b) = \infty$ for frontend $f$ and backend $b$ connected in $\tilde{\mathcal{G}}$. Thus, for any minimum s-t cut $(S,T)$, if $f\in {S}$, $\mB(f) \subseteq S$ as otherwise the cut includes arcs with infinite capacity.
In other words, for any minimum s-t cut $(S,T)$, there exists $B\subseteq \mB$ such that $T = \{\{t\}\cup \{\cup_{b\in B} \mF(b)\} \cup B\}$.

We prove the result by contradiction. Suppose there exists $\hat{B} \subseteq\mB$ such that
\begin{itemize}
    \item there exists $b\in \tB$, $b\notin \hat{B}$,
    \item $(\hat{S}, \hat{T})$ is a minimum s-t cut with $\hat{T} = \{\{t\}\cup \{\cup_{b\in \hat{B}} \mF(b)\} \cup \hat{B}\}$
\end{itemize}
Let $\hat{B}_1 = \hat{B}\cap \tB, \hat{B}_2 = \hat{B} \setminus \tB$. Then we have $\hat{B}_1 \subset \tB$ and $\hat{B}_2 \subseteq \mB\setminus \tB$.
The capacity of cut $(\hat{S}, \hat{T})$ is:
\begin{align*}
    \sum_{f\in \cup_{b\in \hat{B}}\mF(b)} \lambda_f + \sum_{b\in \mB\setminus \hat{B}} \mu_b(\infty) & \stackrel{(a)}{=} \sum_{f\in \cup_{b\in \hat{B}_1}\mF(b)} \lambda_f + \sum_{f\in \cup_{b\in \hat{B}_2}\mF(b)} \lambda_f + \sum_{b\in \mB}\mu_b(\infty) -
    \sum_{b\in \hat{B}_1} \mu_b(\infty) - \sum_{b\in \hat{B}_2} \mu_b(\infty) \\
    & \stackrel{(b)}{\geq} \sum_{f\in \cup_{b\in \hat{B}_1}\mF(b)} \lambda_f + \sum_{b\in \mB}\mu_b(\infty) -
    \sum_{b\in \hat{B}_1} \mu_b(\infty) \\
    & \stackrel{(c)}{=} \left(\sum_{f\in \tF} \lambda_f - \sum_{f\in \tF \setminus \{\cup_{b\in \hat{B}_1}\mF(b)\}}\lambda_f \right)+ \sum_{b\in \mB}\mu_b(\infty) -
    \left(\sum_{b\in \tB} \mu_b(\infty) - \sum_{b\in \tB \setminus \hat{B}} \mu_b(\infty))\right) \\
    & \stackrel{(d)}{>} \sum_{f\in \tF} \lambda_f+ \sum_{b\in \mB}\mu_b(\infty) - \sum_{b\in \tB} \mu_b(\infty) {=} \sum_{f\in \tF} \lambda_f+ \sum_{b\in \mB\setminus \tB}\mu_b(\infty),
\end{align*}
thus, its capacity is strictly larger than the max flow of the graph $\tilde{\mathcal{G}}$, which is a contradiction.
Here $(a)$ follows from definition of $\hat{B}_1, \hat{B}_2$, and the fact that $\tF = \cup_{b\in \tB}\mF(b)$;
$(b)$ follows from the third property of stability decomposition in Definition \ref{def: stability decomposition} as $\hat{B}_2 \subseteq \mB\setminus\tB$; $(c)$ follows from $\hat{B}_1 = \tB \cap \hat{B} = \tB \setminus \{\tB \setminus \hat{B}\}$ thus $\cup_{b\in \hat{B}_1} \mF(b) \subseteq \tF$; $(d)$ follows from 1) $\tB \setminus \hat{B} \neq \emptyset$ by assumption that there exists $b\in \tB, b\notin \hat{B}$ and 2) $\cup_{f\in \tF \setminus \{\cup_{b\in \hat{B}_1} \mF(b)\} }\mB(f)= \tB\setminus\hat{B}$, thus by the second property of stability decomposition,
\[\sum_{f\in \tF \setminus \{\cup_{b\in \hat{B}_1} \mF(b)\}} \lambda_f < \sum_{b\in \tB\setminus\hat{B}}\mu_b(\infty).\Halmos\]
\end{proof}

Based on the proof of the equivalence between a stability decomposition and a minimum s-t cut, we have the following corollary which will be useful in future proofs.
\begin{corollary}\label{cor: max flow increase iff increase tF lambda}
Fix the maximum service rates $\{\mu_b(\infty)\}_{b\in \mB}$. The max flow of $\tilde{\mathcal{G}}$ increases when we increase $\lambda_f$ if and only if $f\in \tF$.
\end{corollary}
\begin{proof}
Let $\bm{w}$ be the max flow of $\tilde{\mathcal{G}}$.
If we increase $\lambda_f$ for $f\in \tF$, this implies that the residual capacity of this arc is positive under flow $\bm{w}$, allowing us to find an augmenting path that increases the max flow of $\tilde{\mathcal{G}}$.

If we increase $\lambda_f$ for $f\notin \tF$, by the third condition in Definition \ref{def: stability decomposition} and the fact that they are not connected to backends in $\tB$, this will only further overload backends in $\mB\setminus\tB$. Thus, we can repeat the above argument and get the same max flow.\Halmos
\end{proof}

Having established the correspondence between stability decomposition and minimum s-t cut. Two natural questions arise: Is this decomposition unique? Can we compute it effectively? The following results address these questions.

\begin{corollary}[Uniqueness of Stability Decomposition]\label{cor: unique stability composition}
The stability decomposition in Definition \ref{def: stability decomposition} is unique.
\end{corollary}
\begin{proof}
Minimum s-t cuts have an important property:
\begin{lemma}[Corollary 3 \cite{picard1980structure}]
    If $(S,T)$ and $(S',T')$ are minimum s-t cuts in $G = (V,A)$, then $(S\cup S', V\setminus \{S\cup S'\})$ and $(S\cap S', V\setminus \{S\cap S'\})$ are also minimum s-t cuts in $G$.
\end{lemma}
Uniqueness follows from combining Lemma \ref{lem: stability decomposition is a cut} with the above property of minimum s-t cuts.\Halmos
\end{proof}

\begin{corollary}[Construction of $\tF$ using Maximum s-t Flow]\label{cor: construct mF using max flow}
Let $R\subseteq \mF$ denote the set that can reach the sink node $t$ in the residual graph of $\tilde{\mathcal{G}}$ given a maximum s-t flow $w$ (i.e., there exists a directed path from the node to $t$ with non-saturated arcs $A_w$), then $R =  \tF$.
\end{corollary}
\begin{proof}
If $f\in  R$, if we increase $\lambda_f$, i.e., the arc connecting $s$ and $f$ with positive residual capacity, we can then find an augmenting path on the residual graph so that we increase the max flow.

If $f\notin R$, if we increase $\lambda_f$, this does not create an augmenting path on the residual graph so that we keep the max flow fixed.

Putting them together, by Corollary \ref{cor: max flow increase iff increase tF lambda}, we have $R = \tF$.\Halmos
\end{proof}
\begin{remark}
We can construct $\tF$ and thus $\tB$ (these are the backends that only connect to frontends in $\tF$) using algorithms such as the Ford--Fulkerson algorithm and the Edmonds--Karp algorithm. When the arc capacities are integral, the computational complexity is bounded by $\mathcal{O}(|A||w|)$ where $|w|$ denotes the maximum flow value. Otherwise, we can use the Edmonds--Karp algorithm with runtime $\mathcal{O}(|V||A|^2)$.
\end{remark}

\subsection{The Lyapunov Function}
To establish the convergence result, we use the same Lyapunov function as in Section~\ref{sec: stability analysis}.
    \[V(\bN,\bx) = \sum_{b\in \mB}\left|\sum_{f\in \mathcal{F}(b)} \lambda_f x_{f,b} - \mu_b(N_b)\right|,\]
\begin{lemma}\label{lem: property of V general}
For every state $\bN$ and job assignment $\bx \in  X_f(\bN)$,
    \begin{itemize}
        \item $V(\bN,\bx)\geq 0$;
        \item $V(\bN, \bx) \geq \sum_{f \in \mathcal{F} \setminus \tilde{\mathcal{F}}} \lambda_f - \sum_{b\in \mathcal{B}\setminus \tilde{\mathcal{B}}} \mu_b(\infty)$.
    \end{itemize}
\end{lemma}
The second property captures a lower bound that reflects the total excess arrivals to the overloaded backend subset $\mB\setminus \tB$.

\begin{proof}
The first property holds as $V(\bN,\bx)$ is the sum of absolute values.

For the second property,
\begin{align}
    V(\bN,\bx) & = \sum_{b\in \mathcal{B}}\left|\sum_{f\in \mathcal{F}(b)} \lambda_f x_{f,b} - \mu_b(N_b)\right|  \geq \sum_{b\in \mathcal{B}\setminus \tilde{\mathcal{B}}}\left|\sum_{f\in \mathcal{F}(b)} \lambda_f x_{f,b} - \mu_b(N_b)\right| \label{eq: tB flow balanced}\\
    & {\geq} \sum_{b\in \mathcal{B}\setminus \tilde{\mathcal{B}}}\left(\sum_{f\in \mathcal{F}(b)} \lambda_f x_{f,b} - \mu_b(N_b)\right) = \sum_{b\in \mathcal{B}\setminus \tilde{\mathcal{B}}} \sum_{f\in \mathcal{F}(b)} \lambda_f x_{f,b} - \sum_{b\in \mathcal{B}\setminus \tilde{\mathcal{B}}} \mu_b(N_b)  \label{eq: not tB flow positive}\\
    & \geq \sum_{b\in \mathcal{B}\setminus \tilde{\mathcal{B}}} \sum_{f\in \mathcal{F}(b) \setminus \tilde{\mathcal{F}}} \lambda_f x_{f,b} - \sum_{b\in \mathcal{B}\setminus \tilde{\mathcal{B}}} \mu_b(N_b) \stackrel{(a)}{=}  \sum_{f\in \mathcal{F}\setminus \tilde{\mathcal{F}}} \lambda_f \sum_{b\in \mathcal{B}\setminus \tilde{\mathcal{B}}} x_{f,b} - \sum_{b\in \mathcal{B}\setminus \tilde{\mathcal{B}}} \mu_b(N_b) \label{eq: restrict frontends not in tF} \\
    & \stackrel{(b)}{=}  \sum_{f\in \mathcal{F}\setminus \tilde{\mathcal{F}}} \lambda_f  - \sum_{b\in \mathcal{B}\setminus \tilde{\mathcal{B}}} \mu_b(N_b) \stackrel{(c)}{\geq}  \sum_{f\in \mathcal{F}\setminus \tilde{\mathcal{F}}} \lambda_f  - \sum_{b\in \mathcal{B}\setminus \tilde{\mathcal{B}}} \mu_b(\infty), \label{eq: Nb infinity}
\end{align}
where $(a)$ follows as $\cup_{b\in \mathcal{B}\setminus \tilde{\mathcal{B}}} \mathcal{F}(b) \setminus \tilde{\mathcal{F}} = \mathcal{F}\setminus \tilde{\mathcal{F}}$; $(b)$ follows as for every $f\in \mathcal{F}\setminus\tilde{\mathcal{F}}, \mathcal{B}(f) \subseteq  \mathcal{B}\setminus \tilde{\mathcal{B}}$; $(c)$ follows as $\mu_b(\cdot)$ is increasing in workload. \Halmos
\end{proof}

We next establish the following global convergence theorem:
\begin{theorem}\label{thm: general N convergence result}
Each solution $\bN(\cdot)$ of the differential inclusion with any initial state $\bN(0)$ satisfies
\[\lim_{t\to\infty} \mu_b(N_b(t)) = \mu_b(\tilde{N}^*_b), \forall b \in \tB, \lim_{t\to\infty} \mu_b(N_b(t)) = \mu_b(\infty), \forall b \in \mB\setminus \tB,\]
where $\tilde{\bN}^*$ is the unique optimal solution to the following optimization problem for subgraph consisting only of frontends $\tF$, backends $\tB$ and the edges connecting them $\tilde{\mathcal{E}} = \{(f,b) \in \mathcal{E}: f\in \tF, b\in \tB\}$:
\begin{align}
\min_{ \{N_b\}_{b\in \tB}, \{x_{f,b}\}_{f\in \tF, b\in \tB}} \quad & \sum_{b\in \tB} N_b \label{eq: FLU two parts 1}\\
    \text{s.t.} \quad &  \sum_{f\in \tF } \lambda_f x_{f,b}= \mu_b(N_b), \forall b\in \tilde{\mathcal{B}}, \notag\\
    & \sum_{b\in \tB }x_{f,b} = 1, \forall f \in \tilde{\mathcal{F}}, \notag\\
    & x_{f,b} \geq 0, \forall (f,b) \in \tilde{\mathcal{E}}, \notag \\
    & x_{f,b} = 0, \forall (f,b) \notin \tilde{\mathcal{E}}. \notag
\end{align}
\end{theorem}
When the fluid optimization problem \eqref{eq: FLU} is feasible, $\tB = \mB$, the above result is a restatement of Theorem \ref{thm:stability}.
Otherwise, $\tB$ is a strict subset of $\mB$, and this result formalizes the intuition that the backends in $\tB$ converge to a finite workload, whereas the remaining backends asymptotically operate at their maximum service levels $\mu_b(\infty)$. Recall that we let $\bm{L}^*$ denote the equilibrium service rates at the backends under GMSR; the above theorem characterizes the exact values of $\bm{L}^*$:
\[L_b^* = \mu_b(\tilde{N}^*_b), b\in \tB, L_b^* = \mu_b(\infty), b\in \mB\setminus\tB.\]

\begin{proof}
First the feasibility of the optimization problem \eqref{eq: FLU two parts 1} is guaranteed by the second property in Definition \ref{def: stability decomposition}.
The uniqueness of the $\tilde{\bN}^*$ follows from the uniqueness of solution result for the fluid optimization problem \eqref{eq: FLU} to the subgraph with $\tilde{\mathcal{F}}, \tilde{\mathcal{B}}$.

When the fluid optimization problem is infeasible, that is, $\tB$ is a strict subset of $\mB$, there must exist a tier with an imbalanced flow at any time $t$. Following the same argument as in the stability analysis, the Lyapunov function is strictly decreasing whenever there exists a tier with imbalanced flow. Therefore, the Lyapunov function is strictly decreasing at any time $t$.

Therefore,
\[\lim_{t\to\infty} V(\bN(t),\bx(t)) = \sum_{f \in \mathcal{F} \setminus \tilde{\mathcal{F}}} \lambda_f - \sum_{b\in \mathcal{B}\setminus \tilde{\mathcal{B}}} \mu_b(\infty),\]
where the lower bound is provided in Lemma \ref{lem: property of V general}.

We finish the proof by showing $V(\bN, \bx) = \sum_{f \in \mathcal{F} \setminus \tilde{\mathcal{F}}} \lambda_f - \sum_{b\in \mathcal{B}\setminus \tilde{\mathcal{B}}} \mu_b(\infty)$ if and only if $\bN, \bx = \hat{\bN}^*, \hat{\bx}^*$ that are defined as follows:
\begin{itemize}
    \item $\hat{N}^*_b = \tilde{N}_b^*$ for $b\in \tB, \hat{N}_b^* = \infty$ for $b\in \mB\setminus\tB$;
    \item for $f\in \tF, b\in \tB$, $\hat{x}^*_{f,b}$ is the optimal assignment matrix for optimization problem \eqref{eq: FLU two parts 1};
    \item for $f\in \tF, b \in \mB\setminus\tB$, $\hat{x}^*_{f,b} =0$;
    \item for $f\in \mF\setminus \tF, b \in \mB \setminus \tB$, $\hat{x}^*_{f,b}$ is a job assignment rule such that
\[\sum_{f \in \mF\setminus \tF} \lambda_f x_{f,b} \geq \mu_b(\infty), \forall b \in \mB\setminus \tB; \sum_{b\in \mB\setminus \tB} x_{f,b} = 1, \forall f\in \mF\setminus \tF; x_{f,b} \geq 0, \forall (f,b) \in \mathcal{E}; x_{f,b} = 0, \forall (f,b) \notin \mathcal{E}.\]
The existence of such $\hat{x}^*_{f,b}$ is guaranteed by the third property of stability decomposition in Definition~\ref{def: stability decomposition}.
\end{itemize}

The ``if" part is straight-forward, while on the other hand, if $V(\bN, \bx) = \sum_{f \in \mathcal{F} \setminus \tilde{\mathcal{F}}} \lambda_f - \sum_{b\in \mathcal{B}\setminus \tilde{\mathcal{B}}} \mu_b(\infty)$, we must have
\begin{itemize}
    \item $\sum_{b\in \tilde{\mathcal{B}}}\left|\sum_{f\in \mathcal{F}(b)} \lambda_f x_{f,b} - \mu_b(N_b)\right|=0$ by \eqref{eq: tB flow balanced}, i.e., flow balance at each backend $b\in \tB$,
    \item $x_{f,b} = 0$ for $f\in \tilde{\mathcal{F}}$, $b\in \mathcal{B}\setminus\tilde{\mathcal{B}}$ by \eqref{eq: restrict frontends not in tF}, i.e., all jobs from frontends in $\tF$ are sent to backends in $\tB$.
\end{itemize}
These together imply $N_b = \tilde{N}_b$ for $b\in \tB$ using a similar argument in the proof of Lemma \ref{lem: m V positive definite}.
Meanwhile, $N_b = \infty$ for $b\in \mathcal{B}\setminus \tilde{\mathcal{B}}$ by \eqref{eq: Nb infinity} and $\sum_{f\in \mF} \lambda_f x_{f,b} - \mu_b(N_b)\geq 0$ by \eqref{eq: not tB flow positive}.
Thus we have $\bN = \hat{\bN}^*, \bx = \hat{\bx}^*$.
\Halmos
\end{proof}

\subsection{Proof of Lemma \ref{lem: OPT TP bounds service rate}}
\begin{proof}
    To see $\textsf{OPT-TP}  \geq \limsup_{k\to\infty}\frac{1}{k}\sum_{i=0}^{k-1}\mathbb{E}_{\pi}\left[\sum_{b\in \mB}\mu_b(N_b(i))\right]$, by the system dynamics,
\[N_b(k) = N_b(0) + \sum_{i=0}^{k-1} \sum_{f\in \mF(b)} A_{f,b}(i) - \sum_{i=0}^{k-1} D_b(i)\geq 0,\]
thus
\[\frac{N_b(0)}{k} + \frac{1}{k}\sum_{i=0}^{k-1} \sum_{f\in \mF(b)} \mathbb{E}[A_{f,b}(i)] \geq \frac{1}{k}\sum_{i=0}^{k-1} \mathbb{E}[D_b(i)] = \frac{1}{k}\sum_{i=0}^{k-1} \mathbb{E}[\mu_b(N_b(i))].\]

On the other hand,
\[G_f(i+1) = G_f(i) + W_f(i) -\sum_{b\in \mB(f)} A_{f,b}(i),\]
thus
\[\frac{1}{k}\sum_{i=0}^{k-1} \sum_{b\in \mB(f)} \mathbb{E}[A_{f,b}(i)] = \frac{1}{k}\sum_{i=0}^{k-1} \mathbb{E}[W_f(i)] + \frac{G_f(0) - \mathbb{E}[G_f(k)]}{k} = \lambda_f +  \frac{G_f(0) - \mathbb{E}[G_f(k)]}{k} \leq \lambda_f + \frac{G_f(0)}{k}.\]

Let $\bar{L}_b(k) := \sum_{i=0}^{k-1} \mathbb{E}[\mu_b(N_b(i))]/k$, $\bar{x}_{f,b}(k) := \sum_{i=0}^{k-1} \mathbb{E}[A_{f,b}(i)]/(k\lambda_f)$, then we have $\bar{\bx}(k), \bar{\bm{L}}(k)$ is a feasible solution to the following optimization problem:
\begin{align}
\textsf{OPT-TP}(k) = \max_{\bm{L}, \bx} \quad & \sum_{b\in \mathcal{B}} L_b \label{eq: FLU TP k}\\
    \text{s.t.} \quad & \frac{N_b(0)}{k} + \sum_{f\in \mathcal{F}} \lambda_f x_{f,b} \geq L_b , \forall b\in\mathcal{B}, \notag\\
    & L_b \leq \mu_b(\infty), \forall b\in \mB, \notag\\
    & \sum_{b\in \mathcal{B}}x_{f,b} \leq 1 + \frac{G_f(0)}{k\lambda_f}, \forall f \in \mathcal{F}, \notag\\
    & x_{f,b} \geq 0, \forall (f,b) \in \mathcal{E}, \notag \\
    & x_{f,b} = 0, \forall (f,b) \notin \mathcal{E}. \notag
\end{align}
Thus \[\sum_{b\in \mB}\bar{L}_b(k) \leq \textsf{OPT-TP}(k).\]

We proceed to show the following:
\[\textsf{OPT-TP}(k) \leq \textsf{OPT-TP} + \sum_{b\in \tB}\frac{N_b(0)}{k} +  \sum_{f\in \tF} \frac{G_f(0)}{k},\]
then the result follows by taking $\limsup$ over $k$.

To see this, for any feasible solution $\bm{L}$ to \eqref{eq: FLU TP k}, we have
\begin{align*}
    \sum_{b\in \mB} L_b & = \sum_{b\in \tB} L_b + \sum_{b\in \mB\setminus\tB} L_b \stackrel{(a)}{\leq} \sum_{b\in \tB}\mu_b(\infty) + \sum_{b\in \tB}\frac{N_b(0)}{k} + \sum_{b\in \tB} \sum_{f\in \mF}\lambda_f x_{f,b} \\
    & \stackrel{(b)}{=} \sum_{b\in \tB}\mu_b(\infty) + \sum_{b\in \tB}\frac{N_b(0)}{k} + \sum_{f\in \tF}\lambda_f \sum_{b\in \tB}x_{f,b} \stackrel{(c)}{\leq}  \sum_{b\in \tB}\mu_b(\infty) + \sum_{b\in \tB}\frac{N_b(0)}{k} + \sum_{f\in \tF}\lambda_f \left(1+ \frac{G_f(0)}{k}\right)\\
    & \stackrel{(d)}{=} \textsf{OPT-TP} + \sum_{b\in \tB}\frac{N_b(0)}{k} +  \sum_{f\in \tF} \frac{G_f(0)}{k}
\end{align*}
where $(a)$ follows from 1) $L_b \leq \mu_b(\infty)$ for $b\in \mB\setminus \tB$ and 2) first constraint in \eqref{eq: FLU TP k}; $(b)$ follows as $\tF = \cup_{b\in \tB}\mF(b)$; $(c)$ follows from the third constraint in \eqref{eq: FLU TP k}; $(d)$ follows as $\textsf{OPT-TP} = \sum_{f\in \tF} \lambda_f + \sum_{b\in \mB\setminus\tB} \mu_b(\infty)$.
\end{proof}

\subsection{Proof of Proposition \ref{pro: maximum throughput}}
\begin{proof}
Theorem \ref{thm: general N convergence result} describes the precise values of the equilibrium service rates: $L_b^* = \mu_b(\tilde{N}_b^*)$ for $b\in \tB$ and $L_b^* = \mu_b(\infty)$ for $b\in \mB\setminus\tB$, with $\sum_{b\in \tB}\mu_b(\tilde{N}_b^*) = \sum_{f\in \tF} \lambda_f$. We need to show $\bm{L}^*$ is an optimal solution to \textsf{OPT-TP}.

First $(\bm{L}^*, \hat{\bx}^*)$ is feasible to $\textsf{OPT-TP}$ ($\hat{\bx}^*$ defined in the proof of Theorem \ref{thm: general N convergence result} as the equilibrium job assignment matrix).
To see $\sum_{b\in \mB} L_b^* = \textsf{OPT-TP}$, suppose not. Let ${\bm{L}}'$ denote the optimal service rates.
Then $\sum_{b\in \tB} {L}'_b > \sum_{b\in \tB} {L}^*_b$ as $L_b^* = \mu_b(\infty)$ for $b\in \mB\setminus\tB$ cannot be further improved. As $\sum_{b\in \tB}L_b^* = \sum_{f\in \tF} \lambda_f$, we must send more jobs to backends in $\tB$ to achieve service rates ${\bm{L}'}$. This is impossible as backends in $\tB$ are only connected to frontends in $\tF$.
\Halmos
\end{proof}

\subsection{Proof of Proposition \ref{pro: maximum stabilizing backend}}
\begin{proof}
Suppose not. There exists an optimal workload $\bm{L}' \in \mathcal{L}^*$ such that there exists $b \in B(\bm{L}'), b \notin B(\bm{L}^*)$. That is, there exists $b\in \mB\setminus\tB$ such that $L_b' < \mu_b(\infty)$.

Let $\bx'$ denote the corresponding optimal workload assignment matrix, we have
\begin{align*}
    \sum_{b\in \mB} L_b' &= \sum_{b\in \tB} L_b' + \sum_{b\in \mB\setminus\tB} L_b' \stackrel{(a)}{<}  \sum_{b\in \tB} L_b'  + \sum_{b\in \mB\setminus\tB} \mu_b(\infty) \stackrel{(b)}{\leq}  \sum_{b\in \tB} \sum_{f\in \mF} \lambda_f x_{f,b}'  + \sum_{b\in \mB\setminus\tB} \mu_b(\infty) \\
    & \stackrel{(c)}{\leq}  \sum_{f\in \tF}  \lambda_f \sum_{b\in \tB} x_{f,b}'  + \sum_{b\in \mB\setminus\tB} \mu_b(\infty) \stackrel{(d)}{\leq} \sum_{f\in \tF} \lambda_f + \sum_{b\in \mB\setminus\tB} \mu_b(\infty) \stackrel{(e)}{=} \textsf{OPT-TP},
\end{align*}
where $(a)$ follows as there must exist $b\in \mB\setminus \tB$ such that $L_b' < \mu_b(\infty)$; $(b)$ follows from the first constraint in \eqref{eq: FLU TP}; $(c)$ follows as $\tF = \cup_{b\in \tB}\mF(b)$; $(d)$ follows from the third constraint in \eqref{eq: FLU TP}; $(e)$ follows from Proposition \ref{pro: maximum throughput}. This contradicts $\bm{L}' \in \mathcal{L}^*$.
\Halmos
\end{proof}

\subsection{Proof of Proposition \ref{prof: minimum stabilized workload}}
\begin{proof}
The result follows directly from the definition of $\tilde{N}$ (restricted to $\tB$), which is the optimal solution to the optimization problem \eqref{eq: FLU two parts 1}.
\Halmos
\end{proof}